\documentclass[3p,times]{elsarticle}
\usepackage{ecrc}

\usepackage{graphicx}

\usepackage{amssymb}
\usepackage{latexsym}
\usepackage{amsmath}
\usepackage{amsfonts}
\usepackage{amsthm}
\usepackage{hhline}
\usepackage{units}

\usepackage{multirow}
\usepackage{rotating}

\def\I{\mathcal{I}}
\def\M{\mathcal{M}}
\def\D{\mathcal{D}}
\def\IC{\mathcal{IC}}
\def\cc{\textsc{cc}}
\def\qa{\textsc{qa}}
\def\mp{\textsc{mp}}

\def\nmp{\overline{\textsc{mp}}}

\newtheorem{definition}{Definition}
\newtheorem{example}{Example}
\newtheorem{theorem}{Theorem}
\newtheorem{lemma}{Lemma}
\newtheorem{proposition}{Proposition}

\newtheorem{fact}{Fact}

\newcommand{\la}{\langle}
\newcommand{\ra}{\rangle}
\newcommand{\nop}[1]{}
\newcommand{\boxx}{\ \hfill$\Box$}

\usepackage{ecrc}

\volume{00}

\firstpage{1}

\journalname{submitted to Journal of Computer and System Sciences}

\runauth{S. Flesca, F. Furfaro, F. Parisi}

\begin{document}
\begin{frontmatter}



\title{Consistency Checking and Querying in Probabilistic Databases under Integrity Constraints}


\author[label]{Sergio Flesca}
\ead{flesca@dimes.unical.it}
\author[label]{Filippo Furfaro}
\ead{furfaro@dimes.unical.it}
\author[label]{Francesco Parisi}
\ead{fparisi@dimes.unical.it}
\address[label]{DIMES, University of Calabria, Via Bucci - Rende (CS), Italy}

\begin{abstract}
We address the issue of incorporating a particular yet expressive form of integrity constraints (namely, denial constraints) into probabilistic databases.
To this aim, we move away from the common way of giving semantics to probabilistic 
databases, which relies on considering a unique interpretation of the data, and address 
two fundamental problems: \emph{consistency checking} and \emph{query evaluation}.
The former consists in verifying whether there is an 
interpretation which conforms to both the marginal probabilities of the tuples and 
the integrity constraints.
The latter is the problem of answering queries under a 
``cautious'' paradigm, taking into account all interpretations of the data in accordance with the constraints.
In this setting, we investigate the complexity of the above-mentioned problems, and identify 
several tractable cases of practical relevance.
\end{abstract}


\end{frontmatter}

\section{Introduction}
Probabilistic databases (PDBs) are widely used to represent uncertain information in several contexts,
ranging from data collected from sensor networks, data integration from heterogeneous sources,
bio-medical data, and, more in general, data resulting from statistical analyses.
In this setting, several relevant results have been obtained regarding the evaluation of 
conjunctive queries, thanks to the definition of probabilistic frameworks dealing with two
substantially different scenarios:
the case of \emph{tuple-independent} PDBs~\cite{Suciu04, SuciuEDBT2010}, where all the tuples of the database are considered independent one from another, and 
the case of PDBs representing probabilistic networks encoding even 
complex forms of correlations among the data~\cite{Sen09}.
However, none of these frameworks takes into account integrity constraints in the same way as
it happens in the deterministic setting, where constraints are used to enforce the consistency of the data.
In fact, the former framework strongly relies on the independence assumption (which clearly
is in contrast with the presence of the correlations entailed by integrity constraints).
The latter framework is closer to an AI perspective of representing the information, as it 
requires the correlations among the data to be represented as data themselves.
This is different from the DB perspective, where constraints are part of the schema, and not of the data.

\looseness-1
In this paper, we address the issue of incorporating integrity constraints into probabilistic
databases, with the aim of extending the classical semantics and usage of integrity constraints of the deterministic setting to the probabilistic one.
Specifically, we consider one of the most popular logical models for the probabilistic data, where information is represented into tuples associated with probabilities, and give the
possibility of imposing \emph{denial constraints} on the data, i.e., constraints forbidding the
co-existence of certain tuples.
In our framework, the role of integrity constraints is the same as in the deterministic setting: they can be used to decide whether a new tuple can be inserted in the database,
or to decide (a posteriori w.r.t. the generation of the data) if the data are consistent.

Before explaining in detail the main contribution of our work, we provide a motivating example,
which clarifies the impact of augmenting a PDB with (denial)
constraints.
In particular, we focus on the implications on the consistency of the probabilistic data, and 
on the evaluation of queries.
We assume that the reader is acquainted with the data representation model where uncertainty is
represented by associating tuples with a probability, and with the notion of possible world. 
(however, these concepts will be formally recalled in the first sections of the paper).

\subsection*{Motivating Example}
Consider the PDB schema $\D^p$ consisting of the relation schema 
\emph{Room$^p$(Id, Hid, Price, Type, View, P)}, and its instance \emph{room}$^p$ 
in Figure~\ref{fig:room}.

\begin{figure}[!h]
\small
\centering
\begin{tabular}{l||c|c|c|c|c|c||}
  \hhline{~|t:======:t|}
& \raisebox{0cm}[3.2mm][1.1mm]{\textbf{\textit{Id}}} &  
\textbf{\textit{Hid}}   & 
\textbf{\textit{Price}} & 
\textbf{\textit{Type}} &
\textbf{\textit{View}} &
\textbf{\textit{P}}\\
  \hhline{~|:======:|}
$t_1$ & \raisebox{0cm}[3mm][1mm]{1} & 1 &  120 & Std  & Sea  &$p_1$ \\
  \hhline{~||------||}
$t_2$ & \raisebox{0cm}[3mm][1mm]{2}  & 1 & 70 & Suite & Courtyard &$p_2$\\
  \hhline{~||------||}
$t_3$ & \raisebox{0cm}[3mm][1mm]{3}  & 1 &  120 & Std & Sea &$p_3$  \\
\hhline{~|b:======:b|}
\end{tabular}
\label{fig:room}
\caption{Relation instance \textit{room}$^p$}
\end{figure}

Every tuple in \textit{room}$^p$ is characterized by the room
identifier \emph{Id}, the identifier \emph{Hid} of the hotel owning the room,
its price per night,
its type (e.g., ``Standard'', ``Suite''), and the attribute \emph{View} describing the room 
view.
The attribute $P$ specifies the probability that the tuple is true.
For now, we leave the probabilities of the three tuples as parameters ($p_1$, $p_2$, $p_3$), 
as we will consider different values to better explain the main issues 
related to the consistency and the query evaluation.

Assume that the following constraint $ic$ is defined over $\D^p$:
``\emph{in the same hotel, standard rooms cannot be more expensive than suites}''.
This is a denial constraint, as it forbids the coexistence of tuples not satisfying 
the specified property.
In particular, $ic$ entails that $t_1$ and $t_2$ are mutually exclusive, as, according to $t_1$,
the standard room $1$ would be more expensive than the suite room $2$ belonging to the same hotel as room $1$. 
For the same reason, $ic$ forbids the coexistence of $t_2$ and $t_3$.

 %
Finally, consider the following query $q$ on $\D^p$:
``\emph{Are there two standard rooms with sea view in hotel $1$?}''.
We now show how the consistency of the database and the answer of $q$ vary when changing the probabilities of \emph{room}$^p$'s tuples.

\noindent
\textbf{\textit{Case 1 (No admissible interpretation):}}\ \
$p_1=\frac{3}{4}$;\ \
$p_2=\frac{1}{2}$;\ \
$p_3=\frac{1}{2}$.\\
In this case, we can conclude that the database is inconsistent. 
In fact, $ic$ forbids the coexistence of $t_1$ and $t_2$, which means that the possible worlds
containing $t_1$ must be distinct from those containing $t_2$.
But the marginal probabilities of $t_1$ and $t_2$ do not allow this:
the fact that $p_1\!=\!\frac{3}{4}$ and $p_2\!=\!\frac{1}{2}$ implies that the sum of the probabilities of the worlds containing either $t_1$ or $t_2$ would be $\frac{3}{4}\!+\!\frac{1}{2}$, which is
greater than $1$.\\[4pt]
\noindent
\textbf{\textit{Case 2 (Unique admissible interpretation):}}\ \
$p_1=\frac{1}{2}$;\ \
$p_2=\frac{1}{2}$;\ \
$p_3=\frac{1}{2}$.\\
In this case, the database is consistent, as it represents two possible worlds:
$w_1= \{t_1, t_3\}$ and $w_2= \{t_2\}$, both with probability $\frac{1}{2}$ (correspondingly, the possible worlds representing the other subsets of $\{t_1, t_2, t_3\}$ have probability 
$0$).
Observe that  there is no other way to interpret the database, while making the constraint 
satisfied in each possible world, and the probabilities of the possible worlds compatible w.r.t. the marginal probabilities of $t_1$, $t_2$, $t_3$.
Thus, the database is consistent and has a unique admissible interpretation.\\
Now, evaluating the above-defined query $q$ over all the admissible interpretations of  the database yields the answer \emph{true} with probability $\frac{1}{2}$ 
(which is the probability of $w_1$, the only
non-zero-probability world, in the unique admissible interpretation, where $q$ evaluates to \emph{true}).
Note that, if $ic$ were disregarded and $q$ were evaluated using the independence
assumption, the answer of $q$ would be \emph{true} with probability $\frac{1}{4}$.\\[4pt]\noindent
\textbf{\textit{Case 3 (Multiple admissible interpretations):}}\ \
$p_1=\frac{1}{2}$;\ \
$p_2=\frac{1}{4}$;\ \
$p_3=\frac{1}{2}$.\\
In this case, we can conclude that the database is consistent, as it admits at least 
the interpretations $I_1$ and $I_2$ represented in the two rows of the following table (each cell 
is the probability of the possible world reported in the column header).

\begin{figure*}[!h]
\small\centering
\begin{tabular}{r||c|c|c|c|c|c|c|c||}
\hhline{~|t:========:t|}
&\raisebox{0cm}[3.2mm][1.1mm]{\textbf{\textit{$\hspace*{-1mm}\emptyset\hspace*{-1mm}$}} }     & 
\textbf{\textit{$\hspace*{-1mm}\{ t_1 \}\hspace*{-1mm}$}}      & 
\textbf{\textit{$\hspace*{-1mm}\{ t_2 \}\hspace*{-1mm}$}}   & 
\textbf{\textit{$\hspace*{-1mm}\{ t_3 \}\hspace*{-1mm}$}} & 
\textbf{\textit{$\hspace*{-1mm}\{ t_1, t_2 \}\hspace*{-1mm}$}}     & 
\textbf{\textit{$\hspace*{-1mm}\{ t_1, t_3 \}\hspace*{-1mm}$}}   & 
\textbf{\textit{$\hspace*{-1mm}\{ t_2, t_3 \}\hspace*{-1mm}$ }}   & 
\textbf{\textit{$\hspace*{-1mm}\{ t_1, t_2, t_3 \}\hspace*{-1mm}$}}\\
\hhline{~|:========:|}
$\hspace*{-1mm}I_1$&
{$0$} & 
{$\nicefrac{1}{4}$} &
$\nicefrac{1}{4}$ & 
$\nicefrac{1}{4}$ & 
$0$ & 
$\nicefrac{1}{4}$ & 
$0$ & 
$0$\\
\hhline{~||--------||}
$\hspace*{-1mm}I_2$&
{$\nicefrac{1}{4}$} & 
$0$ & 
$\nicefrac{1}{4}$ & 
$0$ & 
$0$ & 
$\nicefrac{1}{2}$ & 
$0$ & 
$0$\\
\hhline{~|b:========:b|}
\end{tabular}
\end{figure*}

With a little effort, the reader can check that there are infinitely many ways of interpreting
the database while satisfying the constraints:
each interpretation can be obtained by assigning to the possible world $\{t_1, t_3\}$ a
different probability in the range $[\frac{1}{4}, \frac{1}{2}]$, 
and then suitably modifying the probabilities of the other possible worlds where $ic$ is 
satisfied.
Basically, the interpretations $I_1$ and $I_2$ correspond to the two extreme possible scenarios
where, compatibly with the integrity constraint $ic$,
a strong negative or positive correlation exists between $t_1$ and $t_3$.
The other interpretations correspond to scenarios where an
``intermediate'' correlation exists between $t_1$ and $t_3$.
Thus, differently from the previous case, there is now more than one admissible
interpretation for the database.\\
Observe that, in the absence of any additional information about the actual correlation among the tuples of \emph{room}$^p$, all of the above-described admissible interpretations are equally 
reasonable.
Hence, when evaluating queries, we use a ``cautious'' paradigm, where all the admissible interpretations are taken into account -- meaning that no assumption on the actual correlations
among tuples is made, besides those which are derivable from the integrity constraints.
Thus, according to this paradigm, 
the answer of query $q$ is \textit{true} 
with a probability range $[\frac{1}{4}, \frac{1}{2}]$ (where the boundaries of this range are 
the overall probabilities assigned to the possible worlds containing both $t_1$ and $t_3$ by 
$I_1$ and $I_2$).
As pointed out in the discussion of Case 2, if the independence assumption were adopted (and $ic$
disregarded), the answer of $q$ would be \emph{true} with probability $\frac{1}{4}$, which is 
the left boundary of the probability range got as cautious answer.

  
\vspace*{-1mm}
\subsection*{Main contribution}
We address the following two fundamental problems:
\begin{enumerate}
\item[$1)$] 
\emph{\underline{Consistency checking}}: 
the problem of deciding the consistency of a PDB w.r.t. a given set of denial constraints,
that is deciding if there is at least one admissible interpretation of the data. 
This problem naturally arises when integrity constraints are considered over PDBs:
the information encoded in the data (which are typically uncertain) may be in contrast with the information encoded in the constraints (which are typically certain, as they express 
well-established knowledge about the data domain).
Hence, detecting possible inconsistencies arising from the co-existence of certain and
uncertain information is relevant in several contexts,
such as query evaluation, data cleaning and repairing.

In this regard, our contribution consists in a thorough characterization of the complexity 
of this problem.
Specifically, after noticing that, in the general case, this problem is $N\!P$-complete (owing 
to its interconnection to the probabilistic version of SAT),
we identify several islands of tractability,
which hold when either:
\begin{enumerate}
\item[$i)$] 
the conflict hypergraph (i.e., the hypergraph whose edges are the sets of tuples which 
can not coexist according to the constraints) has some structural property (namely, it is a 
hypertree or a ring), or
\item[$ii)$]
the constraints have some syntactic properties (independently from the shape of the conflict
hypergraph).
\end{enumerate}

\item[$2)$] 
\emph{\underline{Query evaluation}}:
the problem of evaluating queries over a database which is consistent w.r.t. a
given set of denial constraints.
Query evaluation relies on the ``cautious'' paradigm described in Case 3 of the motivating 
example above, which takes into account all the possible ways of interpreting the data in 
accordance with the constraints.
Specifically, query answers consist of pairs $\langle t, r_p\rangle$, 
where $t$ is a tuple and $r_p$ a range of probabilities.
Therein, $r_p$ is the narrowest interval containing all the probabilities which
would be obtained for $t$ as an answer of the query when considering all the admissible 
interpretations of the data (and, thus, all the correlations among the data compatible with 
the constraints).

For this problem, we address both its decisional and search versions, 
studying the sensitivity of their complexity to the specific constraints imposed on the data
and the characteristics of the query.
We show that, in the case of general conjunctive queries, the query evaluation
problem is $F\!P^{N\!P[\log n]}$-hard and in $F\!P^{N\!P}$ (note that $F\!P^{N\!P}$ is contained 
in $\#P$, the class for which the query evaluation problem under the independence assumption 
is complete).
Moreover, we identify tractable cases where the query evaluation problem is in \emph{PTIME}, 
which depend on the characteristics of the query and, analogously to the case of the consistency checking problem, on either the syntactic form of the constraints or on some structural properties of the conflict hypergraph.
\end{enumerate}


Moreover, we consider the following extensions of the framework and discuss their impact on the above-summarized results:
\begin{enumerate}
\item[$A)$] 
\emph{tuples are associated with probability ranges, rather than single probabilities}:
this is useful when the data acquisition process is not able to assign a precise probability
value to the tuples~\cite{LakshmananLRS97,Luk01};
\item[$B)$]  
\emph{also denial constraints are probabilistic}: 
this allows also the domain knowledge encoded by the constraints to be taken into account as 
uncertain;
\item[$C)$] 
\emph{pairs of tuples are considered independent unless this contradicts the constraints}:
this is a way of interpreting the data in between adopting tuple-independence and rejecting it, and is well suited for those cases where one finds it reasonable to assume some groups of
tuples as independent from one another.
For instance, if we consider further tuples pertaining to a different  hotel in the introductory example (where constraints involve tuples over the same hotel), it may be reasonable to assume that these tuples encode events independent from those pertaining 
hotel $1$.
\end{enumerate}

\section{Fundamental notions}
\label{sec:FundamentalNotions}
\subsection{Deterministic Databases and Constraints}
We assume classical notions of database schema, relation schema, and relation instance.
Relation schemas will be represented by sorted predicates of the form $R(A_1, \dots, A_n)$, 
where $R$ is said to be the name of the relation schema and $A_1, \dots, A_n$ are attribute
names composing the set denoted as \emph{Attr}$(R)$.
A tuple over a relation schema $R(A_1, \dots, A_n)$ is a member of 
$\Delta_1\times \dots \times \Delta_n$, where each $\Delta_i$ is the domain of 
attribute $A_i$ (with $i\in [1..n]$).
A relation instance of $R$ is a set $r$ of tuples over $R$.
A database schema $\D$ is a set of relation schemas, and a database instance $D$ of 
$\D$ is a set of relation instances of the relation schemas of $\D$.
Given a tuple $t$, the value of attribute $A$ of $t$ will be denoted as $t[A]$.

A \emph{denial constraint} over a database schema $\D$ is of the form
$\forall \vec{x}. \neg [R_1(\vec{x}_1)\wedge\dots\wedge R_m(\vec{x}_m)\wedge\phi(\vec{x})]$, where:\vspace*{-2mm}
\begin{list}{--}
     {\setlength{\rightmargin}{0mm}
      \setlength{\leftmargin}{3mm}
      \setlength{\itemindent}{-1mm} 
			\setlength{\itemsep}{0mm}}
\item
$R_1,\dots,R_m$ are name of relations in $\D$;
\item
$\vec{x}$ 
is a tuple of variables
and $\vec{x}_1,\dots,\vec{x}_m$ are tuples of variables and constants such that 
$\vec{x}=\mbox{Var}(\vec{x}_1)\cup\dots\cup\mbox{Var}(\vec{x}_m)$, 
where $\mbox{Var}(\vec{x}_i)$ denotes the set of variables in $\vec{x}_i$;
\item
$\phi(\vec{x})$ is a conjunction of built-in predicates of the form
$x \diamond y$, where $x$ and $y$ are either variables in $\vec{x}$ or constants, and 
$\diamond$ is a comparison operator in $\{=, \neq,\leq,\geq,<,> \}$.
\end{list}

\vspace*{-1mm}$m$ is said to be the \textit{arity} of the constraint.
Denial constraints of arity $2$ are said to be \emph{binary}.
For the sake of brevity, 
constraints will be written in the form: 
$\neg [R_1(\vec{x}_1)\wedge\dots\wedge R_m(\vec{x}_m)\wedge\phi(\vec{x})]$, thus omitting 
the quantification $\forall \vec{x}$.

We say that a denial constraint $ic$ is \emph{join-free} if no 
variable occurs in two distinct relation atoms of $ic$,
and, for each built-in predicate occurring in $\phi$,
at least one term is a constant.
Observe that join-free constraints allow multiple occurrences of the same relation name.

It is worth noting that denial constraints enable \emph{equality generating dependencies} 
(EGDs) to be expressed: an EGD is a denial constraint where all the conjuncts of $\phi$
are not-equal predicates.
Obviously, this means that a denial constraints enables also a functional dependency (FD) to be expressed, as an FD is a binary EGDs over a unique relation (when referring to FDs, we consider also non-canonical ones, i.e., FDs whose RHSs contain multiple attributes).

Given an instance $D$ of the database schema $\D$ and an integrity constraint $ic$ over $\D$, 
the fact that $D$ satisfies (resp., does not satify) $ic$ is denoted as 
$D\models ic$ (resp., $D\not\models ic$) and is defined in the standard way.
$D$ is said to be consistent w.r.t. a set of integrity constraints $\IC$, denoted 
with $D\models \IC$, iff $\forall ic \in\IC\,D\models ic$ .

\begin{example}\label{ex:deterministicDB-IC}
Let $\D$ be the (deterministic) database schema consisting of the relation schema 
\emph{Room(Id, Hid, Price, Type, View)}, obtained by removing the probability attribute 
from the relation schema of our motivating example.
Assume the following denial constraints over $\D$:\vspace*{-1mm}
\begin{list}{--}
     {\setlength{\rightmargin}{0mm}
      \setlength{\leftmargin}{3mm}
      \setlength{\itemindent}{-1mm} 
			\setlength{\itemsep}{-1mm}}
\item[$ic$:]
$\neg [$\emph{Room}$(x_1,x_2,x_3,$ \emph{`Std'}$,x_4)\wedge$
        \emph{Room}$(x_5,x_2,x_6,$ \emph{`Suite'} $,x_7)\wedge\ x_3>x_6]$,
saying that, in the same hotel, 
there can not be standard rooms more expensive than suites;
\item[$ic'$:]
$\neg [$\emph{Room}$(x_1,x_2,x_3,x_4,x_5)\wedge$ \emph{Room}$(x_6,x_2, x_7,x_4,x_8)\wedge x_3\neq x_7]$,
imposing that rooms of the same type and hotel have the same price.
Thus, $ic'$ is the FD: \emph{HId}, \emph{Type}$\rightarrow$ \emph{Price}.
\end{list}
\vspace*{-1mm}
\noindent
where $ic$ is the constraint presented in the introductory example.
Consider the relation instance \emph{room} of \emph{Room}, obtained from the 
instance \emph{room}$^p$ of the motivating example by removing column \emph{P}.
It is easy to see that \emph{room} satisfies $ic'$, but does not 
satisfy $ic$, since, for the same hotel, the price of standard rooms 
(rooms $1$ and $3$) is greater than that of suite room $2$.
\boxx
\end{example}

\subsection{Hypergraphs and hypertrees}
\label{sec:hypergraphs}
A hypergraph is a pair $H = \la N, E \ra$, where $N$ is a set of \emph{nodes}, 
and $E$ a set of subsets of $N$, called the \emph{hyperedges} of $H$. 
The sets $N$ and $E$ will be also denoted as $N(H)$ and $E(H)$, respectively.
Hypergraphs generalize graphs, as graphs are hypergraphs 
whose hyperedges have exactly two elements (and are called \emph{edges}).
Examples of hypergraphs are depicted in Figure~\ref{fig:esempioipergrafi}.

Given a hypergraph $H = \la N, E \ra$ and a pair of its nodes $n_1$, $n_2$, a 
\textit{path} connecting  $n_1$ and $n_2$ is a sequence $e_1$, $\dots$, $e_m$ of 
distinct hyperedges of $H$ (with $m\geq 1$) such that 
$n_1\in e_1$, $n_2\in e_m$ and, 
for each $i\in [1..m-1]$, $e_i\cap e_{i+1}\neq \emptyset$.
A path connecting $n_1$ and $n_2$ is said to be \emph{trivial} if $m=1$, that is,
if it consists of a single edge containing both nodes. 

Let $\mathcal{R}=e_1, \dots, e_m$ be a sequence of hyperedges.
We say that $e_i$ and $e_j$ are \emph{neighbors} if $j=i+1$, or $i=m$ and $j=1$ 
(or: if $i=j+1$, or $i=1$ and $j=m$).
The sequence $\mathcal{R}$ is said to be a \emph{ring} if:
$i)$ $m \geq 3$;
$ii)$ for each pair $e_i$, $e_j$ ($i\neq j$), it holds that 
$e_i \cap e_j\neq \emptyset$ if and only if $e_i$ and $e_j$ are neighbors.
An example of ring is depicted in Figure~\ref{fig:esempioipergrafi}(b).
It is easy to see that the definition of ring collapses to the definition of cycle in the
case that the hypergraph is a graph.

The nodes appearing in a unique edge will be said to be \emph{ears} of that edge. 
The set of ears of an edge $e$ will be denoted as \emph{ears}$(e)$.
For instance, in Figure~\ref{fig:esempioipergrafi}(a), 
\emph{ears}$(e_1)=\{t_2\}$ and \emph{ears}$(e_3)=\emptyset$.

A set of nodes $N'$ of $H$ is said to be an \emph{edge-equivalent set} if all the nodes 
in $N'$ appear altogether in the edges of $H$.
That is, for each $e\in E$ such that $e\cap N'\neq \emptyset$, it holds that $e\cap N'=N'$.
Equivalently, the nodes in $N'$ are said to be \emph{edge-equivalent}.
For instance, in the hypergraph of Figure~\ref{fig:esempioipergrafi}(b), $\{t_1, t_2\}$ is an 
edge-equivalent set, as both $t_1$ and $t_2$ belong to the edges $e_1$, $e_2$ only.
Analogously, in the hypergraph of Figure~\ref{fig:esempioipergrafi}(c), nodes $t_2$ and $t_3$ are
edge equivalent, while $\{t_2, t_3, t_4\}$ is not an edge-equivalent set.
Observe that sets of nodes which do not belong to any edge, as well as the ears of an edge (which belong to one edge only), are particular cases of edge-equivalent sets.

A hypergraph is said to be \textit{connected} iff, for each pair of its nodes,
a path connects them.
A hypergraph $H$ is a \textit{hypertree} iff it is connected and it 
satisfies the following \emph{acyclicity} property:
there is no pair of edges $e_1$, $e_2$ such that removing the nodes composing their 
intersection from every edge of $H$ results in a new hypergraph where the remaining nodes 
of $e_1$ are still connected to the remaining nodes of $e_2$.
An instance of hypertree is depicted in Figure~\ref{fig:esempioipergrafi}(c).
Observe that the hypergraph in Figure~\ref{fig:esempioipergrafi}(a) is not a hypertree, as 
the nodes $t_2$ and $t_6$ of $e_1$ and $e_2$, respectively, are still connected (through 
the path $e_1$, $e_3$, $e_2$) even if we remove node $t_1$, which is shared by $e_1$ and $e_2$.
It is easy to see that hypertrees generalize trees.
Basically, the acyclicity property of hypertrees used in this paper is the 
well-known $\gamma$-acyclicity property introduced in~\cite{Fagin83}.
In~\cite{DAtri84,Fagin83}, polynomial time algorithms for checking that a hypergraph is
$\gamma$-acyclic (and thus a hypertree) are provided.

\begin{figure}[!t]
\centering
\begin{tabular}{ccc}
\hspace*{5mm}
\includegraphics{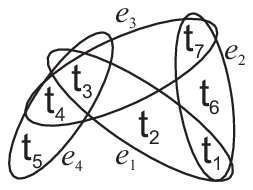} 
\hspace*{5mm}&
\hspace*{5mm}
\includegraphics{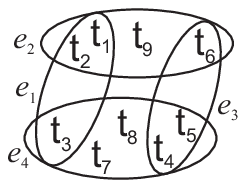}
\hspace*{5mm} &
\hspace*{5mm}
\includegraphics{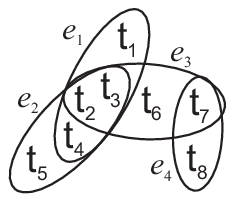}
\hspace*{5mm}\\
$(a)$ & $(b)$ & $(c)$
	\end{tabular}	
	\caption{An example of hypergraph (a), ring (b), hypertree (c)}
	\label{fig:esempioipergrafi}
\end{figure}

\section{PDB\lowercase{s} under integrity constraints}
\label{sec:PDBunderconstraints}

\subsection{Probabilistic Databases (PDBs)}
A probabilistic relation schema is a classical relation schema with a distinguished 
attribute $P$, called \emph{probability}, whose domain is the real interval $[0,1]$ 
and which functionally depends on the set of the other attributes.
Hence, a probabilistic relation schema has the form  $R^p(A_1, \dots, A_n,P)$.
A PDB schema $\D^p$ is a set of probabilistic relation schemas.
A probabilistic relation instance $r^p$ is an instance of $R^p$ and a PDB instance $D^p$ is an instance of $\D^p$.
We  use the superscript $p$ to denote probabilistic relation and database schemas, and their instances.
For a tuple $t\in D^p$, the value $t[P]$ is the probability that 
$t$ belongs to the real world.
We also denote $t[P]$ as $p(t)$.

Given a probabilistic relation schema $R^p$ (resp., relation instance $r^p$, probabilistic 
tuple $t$), we write $det(R^p)$ (resp., $det(r^p)$, $det(t)$) to denote its 
``deterministic'' part.
Hence, given $R^p(A_1, \dots, A_n,P)$, $det(R^p)=R(A_1, \dots, A_n)$, and
$det(r^p)=\pi_{Attr(det(R^p))}(r^p)$, and
$det(t)=$ $\pi_{Attr(det(R^p))}(t)$.
This definition is extended to deal with the deterministic part of PDB 
schemas and instances in the obvious way.

\subsubsection{Possible world semantics}
\label{sec:possworldsemantics}
The semantics of a PDB is based on \textit{possible worlds}.
Given a PDB $D^p$, a possible world is any subset of its 
deterministic part $det(D^p)$.
The set of possible worlds of $D^p$ is as follows:
$pwd(D^p)=\{w \ |\ w\subseteq det(D^p) \}$.
An $Pr$ interpretation of $D^p$ is a \emph{probability distribution function} (PDF) over 
the set of possible worlds $pwd(D^p)$ which satisfies the following property:
$$(i)\ \ 
\forall t\in D^p,\ \ 
p(t)=\hspace*{-3mm}\sum_{\scriptsize\begin{array}{c}w\in pwd(D^p)\\ \wedge\ det(t)\in w\end{array}}\hspace*{-3mm}Pr(w).$$



Condition $(i)$ imposes that the probability of each tuple $t$ of $D^p$ coincides with
that specified in $t$ itself.
Observe that, from definition of PDF, $Pr$ must also satisfy the following conditions:
$$(ii)\ 
\sum_{\scriptsize\begin{array}{c}w\in pwd(D^p)\end{array}} \hspace*{-2mm}Pr(w)=1;\hspace*{10mm}
(iii)\ \
\forall w\in pwd(D^p),\ Pr(w)\geq 0;\hspace*{10mm}$$
\noindent
meaning that $Pr$ assigns a non-negative probability to each possible world, and that the probabilities assigned by $Pr$ to the possible worlds sum up to $1$.

The set of interpretations of a PDB $D^p$ will be denoted as $\I(D^p)$.

Observe that, strictly speaking, possible worlds are sets of deterministic counterparts of
probabilistic tuples.
However, for the sake of simplicity, with a little abuse of notation, in the following 
we will say that a probabilistic tuple $t$ belongs (resp., does not belong) to a possible 
world $w$ -- written $t\in w$ (resp., $t\not\in w$) --  
if $w$  contains (resp., does not contain) the deterministic counterpart of $t$, 
i.e., $det(t)\in w$ (resp., $det(t)\not\in w$).
Moreover, given a deterministic tuple $t$, we will write
$p(t)$ to denote the probability associated with the probabilistic counterpart of $t$.
Thus, $p(t)$ will denote either $t[P]$, in the case that $t$ is a probabilistic tuple, or 
$t'[P]$, in the case that $t$ is deterministic and $t'$ is its probabilistic counterpart.

If independence among tuples is assumed, only one interpretation of $D^p$ is considered,
assigning to each possible world $w$ the probability
$Pr(w)\!=\!\prod_{t\in w} p(t) \times \prod_{t\not\in w} (1\!-\!p(t)).$
In fact, under the independence assumption, the probability of a conjunct of events 
is equal to the product of their probabilities.
In turn, queries over the PDB are evaluated by considering this unique 
interpretation.
In this paper, we consider a different framework, where independence among tuples is not
assumed, and all the possible interpretations are considered. 

\vspace*{-1mm}
\begin{example}\label{ex:probabilisticDB}
Consider the PDB schema $\D^p$ and its instance $D^p$ introduced in our
motivating example.
$D^p$ consists of the relation instance \emph{room}$^p$ reported in Figure~\ref{fig:room}.
Assume that $t_1$, $t_2$, $t_3$ have probabilities 
$p_1\!=\!p_2\!=\!p_3\!=\!1/2$, and disregard the integrity constraint defined in the motivating example.

Table~\ref{tab:PDFs} shows some interpretations of $D^p$.
$Pr_1$ corresponds to the interpretation obtained by assuming tuple independence.
Interpretation $Pr_5$, where $\epsilon$ is any real number in $[0,1/4]$, suffices to show that there are infinitely many interpretations of $D^p$. 
\boxx
\end{example}

\begin{table*}[!h]
\centering
\small
\begin{tabular}{c||l||c|c|c|c|c|c|c|c||l}
\multicolumn{2}{c}{}& \multicolumn{8}{c}{Possible worlds ($w$)} &\\
\hhline{~~|t:========:t|~}
\multicolumn{2}{c||}{}& \raisebox{0cm}[3.2mm][1.1mm]{\textbf{\textit{$\emptyset$}}} &
\textbf{\textit{$\{ t_1 \}$}}      & 
\textbf{\textit{$\{ t_2 \}$}}   & 
\textbf{\textit{$\{ t_3 \}$}} & 
\textbf{\textit{$\{ t_1, t_2 \}$}}     & 
\textbf{\textit{$\{ t_1, t_3 \}$}}   & 
\textbf{\textit{$\{ t_2, t_3 \}$ }}   & 
\textbf{\textit{$\{ t_1, t_2, t_3 \}$}}&
 \\
\hhline{~|t:=::========:|~}
\multirow{6}{*}{\begin{sideways}Interpretations\end{sideways}}&
\raisebox{0cm}[3.2mm][1.1mm]{\textbf{\textit{$Pr_1(w)$}}}&
$1/8$ & $1/8$ & $1/8$ & $1/8$ & $1/8$ & $1/8$ & $1/8$ & $1/8$&
\textbf{\large\}} \  Assuming tuple independence\\
\hhline{~||---------||~}
&\raisebox{0cm}[3.2mm][1.1mm]{\textbf{\textit{$Pr_2(w)$}}}&
$0$ & $1/2$ & $0$ & $0$ & $0$ & $0$ & $1/2$ & $0$&
\multirow{4}{*}{
\hspace*{-3mm}
\begin{tabular}{rl}
\multirow{4}{*}{
\hspace*{-6mm}
$\left. \begin{array}{c} 
					\raisebox{0cm}[3.2mm][1.1mm]{}\\ 
					\raisebox{0cm}[3.2mm][1.1mm]{}\\ 
					\raisebox{0cm}[3.2mm][1.1mm]{}\\ 
					\raisebox{0cm}[3.2mm][1.1mm]{}
				\end{array}  
 \right\rbrace$} &\\
& \hspace*{-4mm} Further interpretations\\ 
& \hspace*{-4mm} corresponding to different\\
& \hspace*{-4mm} correlations among tuples \\
& 
\end{tabular}
}\\
\hhline{~||---------||~}
&\raisebox{0cm}[3.2mm][1.1mm]{\textbf{\textit{$Pr_3(w)$}}}&
$0$ & $0$ & $1/2$ & $0$ & $0$ & $1/2$ & $0$ & $0$& \\
\hhline{~||---------||~}
&\raisebox{0cm}[3.2mm][1.1mm]{\textbf{\textit{$Pr_4(w)$}}}&
$0$ & $0$ & $0$ & $1/2$ & $1/2$ & $0$ & $0$ & $0$& \\
\hhline{~||---------||~}
&\raisebox{0cm}[3.2mm][1.1mm]{\textbf{\textit{$Pr_5(w)$}}}&
$1/2- 2\epsilon$ & $\epsilon$ & $\epsilon$ & $\epsilon$ & $0$ & $0$ & $0$ & $1/2- \epsilon$& \\
\hhline{~|b:=:b:========:b|~}
\end{tabular}
\vspace*{-2mm}\caption{Some interpretations of $D^p$}
\label{tab:PDFs}
\end{table*}



\subsection{Imposing denial constraints over PDBs}
An integrity constraint over  a PDB schema $\D^p$ is written as an integrity constraint over its deterministic part $det(\D^p)$.
Its impact on the semantics of the instances of $\D^p$ is as follows.
As explained in the previous section, a PDB $D^p$, instance of $\D^p$, may
have several interpretations, all equally sound.
However, if some constraints are known on its schema $\D^p$, some interpretations may
have to be rejected. 
The interpretations to be discarded are those ``in contrast'' with the domain
knowledge expressed by the constraints, that is, those assigning a non-zero probability
to worlds violating some constraint.

Formally, given a set of constraints $\IC$ on $\D^p$, an interpretation $Pr\in \I(D^p)$ is \emph{admissible} (and said to be a \emph{model for $D^p$ w.r.t. $\IC$}) if
$\sum_{w\in pwd(D^p)\wedge w\models \IC} Pr(w)=1$
(or, equivalently, if $\sum_{w\in pwd(D^p)\wedge w\not\models \IC} Pr(w)=0$).
The set of models of $D^p$ w.r.t. $\IC$ will be denoted as $\M(D^p,\IC)$.
Obviously, $\M(D^p,\IC)$ coincides with the set of interpretations $\I(D^p)$ if no integrity constraint is imposed ($\IC=\emptyset$), while, in general,  
$\M(D^p,\IC)\subseteq \I(D^p)$.
\begin{example}\label{ex:probabilisticIC}
Consider the PDB $D^p$ and the integrity constraint $ic$ introduced in our
motivating example.
Assume that all the tuples of \emph{room}$^p$ have probability $1/2$.
Thus, the interpretations for $D^p$ are those discussed in Example~\ref{ex:probabilisticDB} 
(see also Table~\ref{tab:PDFs}).
It is easy to see that \emph{room}$^p$ admits at least one model, namely
$Pr_3$ (shown in Table~\ref{tab:PDFs}), which assigns non-zero probability only to  
$w_1\!=\!\{t_2\}$ and $w_2\!=\!\{t_1,t_3\}$.
In fact, it can be proved that $Pr_3$ is the unique model of \emph{room}$^p$ w.r.t. $ic$,
since every other interpretation of \emph{room}$^p$, including $Pr_1$ where tuple independence is assumed, makes the constraint $ic$ violated in some non-zero probability world.
This example shows an interesting aspect of  denial constraints.
Although denial constraints only explicitly forbid the co-existence of tuples, they may implicitly entail the co-existence of tuples:
for instance, for the given probabilities of $t_1$, $t_2$, $t_3$, constraint $ic$
implies the coexistence of $t_1$ and $t_3$.
\boxx
\end{example}

\looseness-1
Example~\ref{ex:probabilisticIC} re-examines Case 2 of our motivating example, and shows a case where
the PDB is consistent and admits a unique model.
The reader is referred to the discussions of Case 1 and Case 3 of the motivating example 
to consider different scenarios, where the PDB is not consistent (Case 1), or is consistent and admits several models (Case 3).



\subsubsection{Modeling denial constraints as hypergraphs}
Basically, a denial constraint over a PDB restricts its models w.r.t. 
the set of interpretations, as it expresses the fact that some sets of tuples of $D^p$ are 
\emph{conflicting}, that is, they cannot co-exist: 
an interpretation is not a model if it assigns a non-zero probability to a possible
world containing these tuples altogether.
Hence, a set of denial constraints $\IC$ can be naturally represented as a 
\emph{conflict hypergraph}, whose nodes are the tuples of $D^p$ and where each hyperedge
consists of a set of tuples whose co-existence is forbidden by a denial constraint in $\IC$
(in fact, hypergraphs were used to model denial constraints also in several works
dealing with consistent query answers in the deterministic setting~\cite{ChomickiMS04}).
The definitions of \emph{conflicting tuples} and \emph{conflict hypergraph} are as
follows.

\begin{definition}[Conflicting set of tuples]
Let $\D^p$ be a PDB schema, 
$\IC$ a set of denial constraints on $\D^p$, and
$D^p$ an instance of $\D^p$.
A set $T$ of tuples of $D^p$ is said to be a \emph{conflicting set} w.r.t. $\IC$ if 
it is a minimal set such that any possible world  
containing all the tuples in $T$ violates $\IC$.
\end{definition}
\begin{example}
In Example~\ref{ex:probabilisticIC}, both $\{t_1,t_2\}$ and $\{t_2,t_3\}$ are conflicting sets of tuples w.r.t. $\IC=\{ic\}$, while $\{t_1,t_2,t_3\}$ is not, as it is not minimal.
\boxx
\end{example}

\begin{definition}[Conflict hypergraph]
Let $\D^p$ be a PDB schema, 
$\IC$ a set of denial constraints on $\D^p$, and
$D^p$ an instance of $\D^p$.
The conflict hypergraph  of $D^p$ w.r.t. $\IC$ 
is the hypergraph $HG(D^p, \IC)$ whose nodes are the tuples of $D^p$ and whose 
hyperedges are the conflicting sets of $D^p$ w.r.t. $\IC$.
\end{definition}

\begin{example}
Consider a database instance $D^p$ having tuples $t_1,\dots,t_9$, and a set of
denial constraints $\IC$ stating that 
$e_1\!=\!\{t_1,t_4, t_7\}$,
$e_2\!=\!\{t_1, t_2, t_3, t_4, t_5, t_6\}$,
$e_3\!=\!\{t_3,t_6, t_9\}$, and
$e_4\!=\!\{t_6,t_8\}$ are conflicting sets of tuples.
The conflict hypergraph $HG(D^p, \IC)$ in Figure~\ref{fig:ConflictHG} concisely represents this fact.\boxx
\end{example}

\noindent
It is easy to see that, if $\IC$ contains binary denial constraints only, 
then the conflict hypergraph collapses to a graph.

\begin{example}
\label{ex:HGesempioprima}
Consider $D^p$ and $\IC=\{ic\}$ of our motivating example -- observe that $ic$ is a binary denial 
constraint.
The graph representing $HG(D^p,\IC)$ is shown in Figure~\ref{fig:ConflictG}.
\boxx
\end{example}

\begin{figure}[h!]
\centering
\begin{minipage}[b]{5cm}
\includegraphics{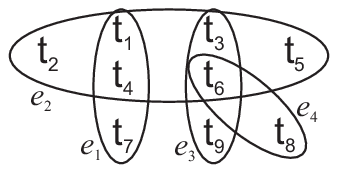}
\caption{A conflict hypergraph.}
\label{fig:ConflictHG}
\end{minipage}
\hspace*{0mm}
\begin{minipage}[b]{3.2cm}
\hspace*{4mm}
\includegraphics[width=2cm]{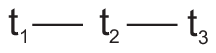}
\caption{Conflict graph of the motivating example}
\label{fig:ConflictG}
\end{minipage}
\end{figure}

It is easy to see that the size of the conflict hypergraph is polynomial w.r.t. 
the size of $D^p$ (in particular, its number of nodes is bounded by the number of tuples 
of $D^p$) and can be constructed in polynomial time w.r.t. the size of $D^p$.\\[-2pt]

\noindent
\textbf{Remark 1.}
Observe that the conflict hypergraph $H(D^p, \IC)$ corresponds to a representation 
of the dual lineage of the \emph{constraint query} $q_{\mathcal{IC}}$, i.e., 
the boolean query $q_{\scriptsize\mathcal{IC}}=\bigvee_{ic\in\IC} (\neg ic)$
which basically asks whether there is \emph{no} model for $D^p$ w.r.t. $\IC$.
For instance,  consider the case of Example~\ref{ex:probabilisticIC}.
A lineage of $q_{\scriptsize\mathcal{IC}}$ is the DNF expression: 
$(X_1 \wedge X_2) \vee (X_2 \wedge X_3)$, where each $X_i$ corresponds to tuple $t_i$.
Thus, the semantics of the considered constraints is captured by the dual lineage,
that is the CNF expression $(Y_1 \vee Y_2) \wedge (Y_2 \vee Y_3)$, where
each $Y_i=$ not$(X_i)$.
It is easy to see that the conflict hypergraph (as described in 
Example~\ref{ex:HGesempioprima}) is the hypergraph of this CNF expression.
In the conclusions (Section~\ref{sec:conclusions}), we will elaborate more on this  
relationship between conflict hypergraphs and (dual) lineages of constraint queries: exploiting this relationship may help 
to tackle the problems addressed in this paper from a different perspective.

%

\section{Consistency checking}
\label{sec:consistencychecking}
Detecting inconsistencies is fundamental for certifying the quality of the data and extracting 
reliable information from them.
In the deterministic setting, inconsistency typically arises from errors that occurred during
the generation of the data, as well as during their acquisition.
In the probabilistic setting, there is one more possible source of inconsistency, coming 
from the technique adopted for estimating the ``degree of uncertainty'' of the acquired 
information, which determines the probability values assigned to the probabilistic 
tuples.
Possible bad assignments of probability values can turn out when integrity constraints on the
data domain (which typically encode certain information coming from well-established knowledge of the domain) are considered.

In this section, we address the problem of checking this form of consistency, that is, the 
problem of checking whether the probabilities associated with the tuples are ``compatible'' 
with the integrity constraints defined over the data. 
It is worth noting that the study of this problem has a strong impact in several aspects of 
the management of probabilistic data: checking the consistency can be used during the
data acquisition phase (in order to ``certify'' the validity of the model applied for
determining the probabilities of the tuples), as well as a preliminary step of the computation
of the query answers. 
Moreover, it is strongly interleaved with the problem of repairing the data, whose  study is deferred to future work.

Before providing the formal definition of the consistency checking problem, we introduce some 
basic notions and notations.
Given a PDB schema $\D^p$, a set of integrity constraint $\IC$, and an instance $D^p$ of $\D^p$, we say that $D^p$ satisfies (resp., does not satisfy) $\IC$, denoted 
as $D^p\models \IC$ (resp., $D^p\not\models \IC$ ) iff the set of models $\M(D^p,\IC)$ is not
empty.
In the following, we will say ``\emph{consistent w.r.t.}'' 
(resp.,  ``\emph{inconsistent w.r.t.}'') meaning the same as ``\emph{satisfies}'' 
(resp.,  ``\emph{does not satisfy}'').



We are now ready to provide the formal definition of the consistency checking problem.
In this definition, as well as in the rest of the paper, we assume that a PDB schema $\D^p$ and a set of denial constraints $\IC$ over $\D^p$ are given.

\begin{definition}[Consistency Checking Problem (\cc)]\label{def:cc}
Given a PDB instance $D^p$ of $\D^p$,
the consistency checking problem (\cc) is deciding whether $D^p\models \IC$.
\end{definition}

We point out that, in our complexity analysis, $\D^p$ and $\IC$ will be assumed of fixed size, 
thus we refer to data complexity.

The following theorem states that \cc\ is $NP$-complete, and it easily derives from the 
interconnection of \cc\ with the $NP$-complete problem 
PSAT~\cite{Pap88} (\emph{Probabilistic satisfiability}), which is the generalization of SAT defined as follows:
``\emph{Let $S=\{C_1, \dots, C_m\}$ be a set of $m$ clauses, where each $C_i$ is a disjunction of literals (i.e, possibly negated propositional variables $x_1,\dots, x_n$)
and each $C_i$ is associated with a probability $p_i$. 
Decide whether $S$ is \emph{satisfiable}, that is, whether there is a probability distribution $\pi$ over all the $2^n$ possible truth assignments over $x_1,\dots, x_n$ 
such that, for each $C_i$, the sum of the probabilities assigned 
by $\pi$ to the truth assignments satisfying $C_i$ is equal to $p_i$.}''
Basically, the membership in $NP$ of \cc\ derives from the fact that any instance of \cc\
over a PDB $D^p$ can be reduced to an equivalent PSAT instance where:
$a)$ the propositional variables correspond to the tuples of $D^p$, 
$b)$ the constraints of \cc\ are encoded into clauses with probability $1$, 
$c)$ the fact that the tuples are assigned a probability is encoded into a clause for each tuple, with probability equal to the tuple probability.
As regards the hardness of \cc\ for $NP$, it intuitively derives from the fact
that the hardness of PSAT was shown in~\cite{Pap88} for the case that only unary clauses 
have probabilities different from $1$: 
thus, this proof can be applied on \cc, by mapping unary clauses to tuples and 
the other clauses (which are deterministic) to constraints
\footnote{
However, we will not provide a formal proof of the $N\!P$-hardness of \cc\ based on this reasoning, that is, based on reducing hard instances of PSAT to \cc\ instances. 
Indeed, a formal proof of the hardness will be provided for the theorems~\ref{theo:ic-arity-3}
and~\ref{theo:2FDs-cc} introduced in Section~\ref{sec:Syntactic-tractable-cases}, which 
are more specific in stating the hardness of \cc\ in that they say that \cc\ is $N\!P$-hard in the presence of denial constraints of some syntactic forms.}.

\begin{theorem}[Complexity of \cc]
\label{theo:CCisNPcomplete}
\cc\  is $NP$-complete.
\end{theorem}

\looseness-1
In the following, we devote our attention to determining tractable cases of \cc, from two 
different perspectives.
First, in Section~\ref{sec:tractablecases}, we will show tractable cases which depend from 
the structural properties of the conflict hypergraph, and, thus, from how the 
data combine with the constraints.
The major results of this section are that \cc\ is tractable if the conflict hypergraph is 
either a hypertree or ring.
Then, in Section~\ref{sec:Syntactic-tractable-cases}, we will show syntactic conditions on the
constraints which make \cc\ tractable, independently from the shape of the conflict hypergraph.
At the end of the latter section, we also discuss the relationship between these two kinds 
of tractable cases.

\subsection{Tractability arising from the structure of the conflict hypergraph}
\label{sec:tractablecases}

It is worth noting that, since there is a polynomial-time reduction from \cc\ to PSAT,
the tractability results for PSAT may be exploited for devising efficient strategy for solving \cc.
In fact, in~\cite{Pap88}, it was shown that 2PSAT (where clauses are binary) can be solved in
polynomial time if the graph of clauses (which contains a node for each literal and an edge
for each pair of literals occurring in the same clause) is outerplanar.
This result relies on a suitable reduction of 2PSAT to a tractable instance of 2MAXSAT
(maximum weight satisfiability with at most two literals per clause).
Since, in the case of binary denial constraints, the conflict hypergraph is a graph and 
the above-discussed reduction of \cc\ to PSAT results in an instance of 2PSAT where the graph
of clauses has the same ``shape'' of our conflict graph, we have that \cc\ is polynomial-time solvable if denial constraints are binary and the conflict graph is outerplanar.
However, on the whole,  reducing 2PSAT to 2MAXSAT and then solving the obtained 2MAXSAT 
instance require a high polynomial-degree computation (specifically, the complexity is 
$O(n^6 \log n)$, where $n$ is the number of literals in the PSAT formula, corresponding
to the number of tuples in our case).

Here, we detect tractable cases of \cc, which, up to our knowledge, are not subsumed by 
any known tractability result for PSAT.
Our tractable cases have the following amenities:\vspace*{-1.5mm}
\begin{list}{--}
     {\setlength{\rightmargin}{0mm}
      \setlength{\leftmargin}{3mm}
      \setlength{\itemindent}{-1mm} 
			\setlength{\itemsep}{-1mm}}
\item
no limitation is put on the arity of the constraints;
\item
instead of exploiting reductions of \cc\ to other problems, we determine necessary and sufficient conditions which can be efficiently checked (in linear time) by only examining
the conflict hypergraph and the probabilities of the tuples. 
\end{list}

Our main results regarding the tractability arising from the structure of the conflict hypergraph (which will be given in sections~\ref{sec:TractabilityHypertrees} and 
\ref{sec:TractabilityRings}) are that consistency can be checked in linear time over the 
conflict hypergraph if it is either a hypertree or a ring.

\subsubsection{New notations and preliminary results}\label{sec:newNotations}
Before providing our characterization of 
tractable cases arising from the structure of the conflict hypergraph, we introduce some 
preliminary results and new notations.
Given a hypergraph $H=\langle N, E\rangle$ and a hyperedge $e\in E$, 
the set of intersections of $e$ with the other hyperedges of $H$ is denoted as 
$Int(e,H)=\{s\: |\: \exists e'\in E \mbox{ s.t. } e'\neq e \wedge s=e\cap e'\}$.
For instance, for the hypertree $H$ in Figure~\ref{fig:esempioipergrafi}(c), 
$Int(e_1,H)=\left\{  \{t_2, t_3\}, \{t_2, t_3, t_4\} \right\}$.
Moreover, given a set of sets $S$, we call $S$ a \emph{matryoshka} if 
there is a total ordering $s_1,\dots, s_n$ of its elements such that, 
for each $i,j \in [1..n]$ with $i<j$ it holds that $s_1\subset s_2 \subset \dots \subset s_n$.
For instance, the above-mentioned set $Int(e_1,H)$ is a matryoshka.
Finally, given a set of hyperedges $S$, we denote as $H^{-S}$ the hypergraph obtained from
$H$ by removing the edges of $S$ and the nodes in the edges of $S$ which do not belong to 
any other edge of the remaining part of $H$.
That is, $H^{-S}=\langle N', E'\rangle$, where $E'=E\setminus S$, 
$N'=\bigcup_{e\in E'} e$.
For instance, for the hypergraph $H$ in Figure~\ref{fig:esempioipergrafi}(a), 
$H^{-\{e_1\}}$ is obtained by removing $e_1$ from the set of edges of $H$, 
and $t_2$ from the set of its nodes.
Analogously, $H^{-\{e_1,e_2\}}$ will not contain edges $e_1$ and $e_2$, as well as 
nodes $t_1$, $t_2$, $t_6$.

The first preliminary result (Proposition \ref{prop:esistenzamatrioska}) states 
a general property of hypertrees:
any hypertree $H$ contains at least one edge $e$ which is attached to the rest of $H$ 
so that the set of intersections of $e$ with the other edges of $H$ is a matryoshka.
Moreover, removing this edge from $H$ results in a new hypergraph which is still a hypertree.
This result is of independent interest, as it allows for reasoning on hypertrees 
(conforming to the $\gamma$-acyclicity property) by using induction on the number of 
hyperedges: 
any hypertree with $x$ edges can be viewed as a hypertree with $x-1$ edges which has been augmented with a new edge, attached to the rest of the hypertree by means of sets
of nodes encapsulated one to another.

\begin{proposition}
\label{prop:esistenzamatrioska}
Let $H=\langle N,E\rangle$ be a hypertree.
Then, there is at least one hyperedge $e\in E$ such that 
$Int(e, H)$ is a matryoshka.
Moreover, $H^{-\{e\}}$ is still a hypertree.
\end{proposition}

As an example, consider the 
hypertree in Figure~\ref{fig:esempioipergrafi}(c).
As ensured by Proposition~\ref{prop:esistenzamatrioska}, this hypertree contains the 
edge $e_1$ whose set of intersections with the other edges is 
$\{\{t_2, t_3\},$ $\{t_2,t_3,t_4\}\}$, which is a matryoska.
Moreover, removing $e_1$ from the set of hyperedges, and the ears of $e_1$ from the set of
nodes, still yields a hypertree.
The same holds for $e_2$ and $e_4$, but not for $e_3$.

The second preliminary result (which will be stated as Lemma~\ref{lem:pmin1/2})
regards the minimum probability that a set of tuples co-exist
according to the models of the given PDB.
Specifically, given a set of tuples $T$ of the PDB 
$D^p$, we denote this minimum probability as $p^{\min}(T)$, whose formal definition is as
follows:
$$p^{\min}(T)=
\begin{array}{c}
 \min\\ 
 \mbox{\scriptsize$Pr\in\M(D^p, \IC)$}
\end{array}
\left\{
\begin{array}{cl}
\sum & \hspace*{-2mm}Pr(w)\\
\mbox{\scriptsize$w\in pwd(D^p)\wedge T\!\subseteq\!w$} & \\
\end{array}
\right\}
$$
The following example clarifies the semantics of $p^{\min}$. 
\begin{example}\label{ex:pmin}
Consider the case discussed in Example~\ref{ex:probabilisticDB} (the same as Case 2 of our
motivating example, but with $\IC=\emptyset$). 
Here, every interpretation is a model.
Hence, $p^{\min}(t_1,t_3)=0$, as there is an interpretation 
(for instance, $Pr_2$ or $Pr_4$ in Table~\ref{tab:PDFs}) which assigns probability 
$0$ to both the possible worlds $\{t_1,t_3\}$ and $\{t_1,t_2,t_3\}$ 
-- the worlds containing both $t_1$ and $t_3$.
However, if we impose $\IC=\{ic\}$ of the motivating example, 
we have that $p^{\min}(t_1,t_3)=1/2$, as according to $Pr_3$ 
(the unique model for the database w.r.t. $\IC$) the probabilities of worlds $\{t_1,t_3\}$ 
and $\{t_1,t_2,t_3\}$ are, respectively, $1/2$ and $0$ (hence, their sum is $1/2$).
\boxx
\end{example}


Lemma~\ref{lem:pmin1/2} states that, for any set of tuples $T=\{t_1,\dots,t_n\}$, independently 
from how they are connected in the conflict hypergraph, the probability that 
they co-exist, 
for every model, has a lower bound which is implied by their marginal probabilities.
This lower bound is $\max\left\{0, \sum_{i=1}^n p(t_i)-n+1\right\}$, which is exactly 
the minimum probability of the co-existence of $t_1,\dots,t_n$ in two cases:
$i)$ the case that $t_1,\dots,t_n$ are pairwise disconnected in the conflict hypergraph 
(which happens, for instance, in the very special case that $t_1,\dots,t_n$ are not involved in any constraint);
$ii)$
the case that the set of intersections of $T$ with the edges of $H$ is a matryoshka.
This is interesting, as it depicts a case of tuples correlated through constraints which 
behave similarly to tuples among which no correlation is expressed by any constraint.

\begin{lemma}
\label{lem:pmin1/2}
Let $D^p$ be an instance of $\D^p$ consistent w.r.t. $\IC$, 
$T$ a set of tuples of $D^p$, and let 
$H$ denote the conflict hypergraph $HG(D^p, \IC)$.
If either 
\emph{i)} the tuples in $T$ are pairwise disconnected in $H$, 
or 
\emph{ii)} $Int(T,H)$ is a matryoshka,
then
$p^{\min}(T)= \max\left\{0, \sum_{t\in T} p(t)-|T|+1\right\}$.
Otherwise, this formula provides a lower bound for $p^{min}(T)$.
\end{lemma}

\subsubsection{Tractability of hypertrees}\label{sec:TractabilityHypertrees}
We are now ready to state our first result on \cc\ tractability.
\begin{theorem}\label{theo:tractableCC}
Given an instance $D^p$ of $\D^p$, if $HG(D^p, \IC)$ is a hypertree, then $D^p\models \IC$
iff, for each hyperedge $e$ of $HG(D^p, \IC)$, it holds that
\begin{equation}\label{eq:tractableCC}
\sum_{t\in e} p(t) \leq |e|-1
\end{equation}
\end{theorem}
\begin{proof}
$(\Rightarrow)$:
We first show that if there is a model for $D^p$ w.r.t. $\IC$, then 
inequality (\ref{eq:tractableCC}) holds for each hyperedge of $HG(D^p, \IC)$.
Reasoning by contradiction, assume that 
$D^p\models\IC$ and there is an hyperedge $e=\{t_1,\dots, t_n\}$ 
of $HG(D^p, \IC)$ such that $\sum_{i=1}^n p(t_i)-n+1>0$.
Since this value is a lower bound for $p^{min}(t_1, \dots, t_n)$ (due to Lemma~\ref{lem:pmin1/2}),
it holds that every model $M$ for $D^p$ w.r.t. $\IC$ assigns a non-zero probability
to some possible world containing all the tuples $t_1, \dots, t_n$.
This contradicts that $M$ is a model, since any possible world containing $t_1,\dots,t_n$ 
does not satisfy $\IC$.\\
\noindent
$(\Leftarrow)$: We now prove that if inequality (\ref{eq:tractableCC}) holds for each 
hyperedge of $HG(D^p, \IC)$, then there is a model for $D^p$ w.r.t. $\IC$.
We reason by induction on the number of hyperedges of $HG(D^p, \IC)$.

The base case is when $HG(D^p, \IC)$ consists of a single hyperedge $e=\{t_1, \dots, t_{k}\}$.
Consider the same database $D^p$, but impose over it the empty set of denial constraints, instead of $\IC$.
Then, from Lemma~\ref{lem:pmin1/2} (case $i)$), we have that 
there is at least one model $M$ for $D^p$ (w.r.t. the empty set of constraints) such that
$\sum_{w\supseteq\{t_1,\dots,t_k\}} M(w)=\max\left\{0, \sum_{i=1}^k p(t_i)-k+1\right\}.$
The term on the right-hand side evaluates to $0$, as, from the hypothesis, we have that 
$\sum_{i=1}^k p(t_i)\leq k-1$.
Hence, $M$ is a model for $D^p$ also w.r.t. $\IC$, since the only constraint entailed by 
$\IC$ is that the tuples $t_1,\dots,t_k$ can not be altogether in any possible world with
non-zero probability.

We now prove the induction step. 
Consider the case that $H= HG(D^p, \IC)$ is a hypertree with $n$ hyperedges.
The induction hypothesis is that the property to be shown holds in the presence of any conflict 
hypergraph consisting of a hypertree with $n-1$ hyperedges.
Let $e$ be a hyperedge of $H$ such that $Int(e,H)$ is a matryoshka, and $H'=H^{-\{e\}}$ is a hypertree.
The existence of $e$ and the fact that $H'$ is a hypertree are guaranteed by 
Proposition~\ref{prop:esistenzamatrioska}.
We denote the nodes in $e$ as $t'_1,\dots, t'_m, t''_1, \dots, t''_n$, where 
$T'=\{t'_1,\dots, t'_m\}$ is the set of nodes of $e$ in $H'$, 
and $T''=\{t''_1, \dots t''_n\}$ are the ears of $e$.
Correspondingly, $D''$ is the portion of $D^p$ containing only the tuples $t''_1, \dots t''_n$, and $D'$ is the portion of $D^p$ containing all the other tuples (that is, the tuples 
corresponding to the nodes of $H'$).
We consider $D'$ associated with the set of constraints imposed by $H'$, and $D''$ associated 
with an empty set of constraints.

Thanks to the induction hypothesis, and to the fact that inequality (\ref{eq:tractableCC}) 
holds, we have that $D'$ is consistent w.r.t. the set of constraints encoded by $H'$.
Moreover, 
since $Int(e,H)$ is a matryoshka, we have that the set $T'$ is such that $Int(T',H')$ is 
a matrioshka too.
Hence, from Lemma~\ref{lem:pmin1/2} (case $ii$) we have that there is a model $M'$ for 
$D'$ w.r.t. $H'$ such that 
$\sum_{w\supseteq \{t'_1,\dots,t'_m\}} M'(w)= \max\left\{0,\right.$ $\left.\sum_{i=1}^m p(t'_i)-m+1\right\}$.
We denote this value as $p'$, and consider the case that $p'>0$ (that is, 
$p'= \sum_{i=1}^m p(t'_i)-m+1$ as the case that $p'=0$ can be proved analogously).
Since inequality (\ref{eq:tractableCC}) holds for every edge of $HG(D^p,\IC)$, the following
inequality holds for the tuples of $e$:
$\sum_{i=1..m} p(t'_i)+\sum_{i=1..n} p(t''_i)-m-n+1\leq 0$.
The quantity $m-\sum_{i=1..m} p(t'_i)$ is equal to $1-p'$, that is the overall probability 
assigned by $M'$ to the possible worlds of $D'$ not containing at least one tuple 
$t'_1,\dots, t'_m$. 
Denoting the probability $1-p'$ as $\overline{p'}$, the above inequality becomes 
$\sum_{i=1..n} p(t''_i)-n+1\leq \overline{p'}$.
Owing to Lemma~\ref{lem:pmin1/2} (case $i$), the term on the left-hand side corresponds to $p^{min}(t''_1,\dots,t''_n)$. 

Intuitively enough, this suffices to end the proof, as it means that, if we arrange the tuples 
$t''_1,\dots,t''_n$ according to a model $M''$ for $D''$ which minimizes the overall probability
of the possible worlds of $D''$ containing $t''_1,\dots,t''_n$ altogether, the portion of the
probability space invested to represent these worlds is less than the portion of the probability
space invested by $M'$ to represent the possible worlds of $D'$ not containing at least one
tuple among $t'_1,\dots,t'_m$.
For the sake of completeness, we formally show how to obtain a model for $D^p$ w.r.t. $\IC$ starting from $M'$ and $M''$.

First of all, observe that any interpretation $Pr$ can be represented as a sequence 
$S(Pr)=(w_1, p1), \dots, (w_k,p_k)$ where:
\begin{itemize}
\item
$w_1,\dots, w_k$ are all the possible worlds such that $Pr(w_i)\neq 0$ for each $i\in [1..k]$;
\item
$p_1=Pr(w_1)$;
\item
for each $i\in [2..n]$ $p_i=p_{i-1}+Pr(w_i)$ (that is, $p_i$ is the cumulative probability
of all the possible worlds in $S(M)$ occurring in the positions not greater than $i$).
In particular, this entails that $p_n=1$.
\end{itemize}
It is easy to see that many sequences can represent the same interpretation $Pr$, each
corresponding to a different permutation of the set of the possible worlds which are assigned 
a non-zero probability by $Pr$.

Consider the model $M'$, 
and let $\alpha$ 
be the number of possible worlds which are assigned by $M'$ a non-zero probability and which do not contain at least one tuple among $t'_1, \dots, t'_m$.
Then, take a sequence $S(M')$ such that the first $\alpha$ pairs are possible worlds not containing at least one tuple among $t'_1,\!\dots,\!t'_m$. 
In this sequence, denoting the generic pair occurring in it as $(w'_i,p'_i)$, 
it holds that $p'_{\alpha}= \overline{p'}$.

Analogously, consider the model $M''$ , and take any sequence $S(M'')$ where the first pair 
contains the possible world containing all the tuples $t''_1, \dots, t''_n$.
Obviously, denoting the generic pair occurring in $S(M'')$ as $(w''_i,p''_i)$
it holds that $p''_1= p^{min}(t''_1,\dots, t''_n)$ is less than or equal to $\overline{p'}$.

Now consider the sequence $S'= (w'''_1, p'''_1), \dots, '(w'''_k, p'''_k)$ defined as follows:
\begin{itemize}
\item
$p'''_1,\dots, p'''_k$ are the distinct (cumulative) probability values occurring in $S(M')$ and $S(M'')$, ordered by their values;
\item
for each $i\in[1..k]$, $w'''_i=w'_j\cup w''_l$, where $w'_j$ (resp., $w''_l$) is the 
possible world occurring in the left-most pair of $S(M')$ (resp., $S(M'')$) containing a 
(cumulative) probability value not less than $p'''_i$.
\end{itemize}

Consider the function $f$ over the set of possible worlds of $D^p$ defined as follows:
$$f(w)= \left\{\begin{array}{ll}
							 0 & \mbox{if $w$ does not occur in any pair of }S'\\
							p'''_1 & \mbox{if $w$ occurs in the first pair of }S'\\
							p'''_i-p'''_{i-1} & \mbox{if $w$ occurs in the } i\mbox{-th pair of }S' (i>1)
							\end{array}
							\right.$$
							
It is easy to see that $f$ is an interpretation for $D^p$.
In fact, by construction, it assigns to each possible world of $D^p$ a value in $[0,1]$, and 
the sum of the values assigned to the possible worlds is $1$. 
Moreover, the values assigned by $f$ to the possible worlds are compatible with the marginal 
probabilities of the tuples, since, for each tuple $t$ of $D'$, 
$\sum_{w'''|t\in w'''} f(w''')= \sum_{w'|t\in w'} M'(w')= p(t)$, as well as
for each tuple $t$ of $D''$, 
$\sum_{w'''|t\in w'''} f(w''')= \sum_{w''|t\in w''} M''(w'')= p(t)$.

In particular, $f$ is also a model for $D^p$ w.r.t. $\IC$:
on the one hand, $f$ assigns $0$ to every possible world containing tuples 
which are conflicting according to $H'$ (this follows from how $f$ was 
obtained starting from $M'$).
Moreover, $f$ assigns $0$ to every possible world containing tuples which are 
conflicting according to the hyperedge $e$.
In fact, the worlds containing all the tuples $t'_1,\dots,t'_m, t''_1,\dots,t''_n$ 
are assigned $0$ by $f$, since the worlds occurring in $S'$ containing 
$t''_1,\dots,t''_n$ do not contain at least one tuple among $t'_1,\dots,t'_m$ 
(this trivially follows from the fact that 
$\overline{p'}> p^{min}(t''_1, \dots, t''_n)$).
The fact that $f$ is a model for $D^p$ w.r.t. $\IC$ means that $D^p\models \IC$.
\end{proof}

The above theorem entails that, if $HG(D^p,\IC)$ is a hypertree, then \cc\
can be decided in time $O(|E|\cdot k)$ over $HG(D^p,$ $\IC)$, where $E$ is the set 
of hyperedges of $HG(D^p,\IC)$ and $k$ is the maximum arity of the constraints 
(which bounds the number of nodes in each hyperedge).
The number of hyperedges in a hypertree is bounded by the number of 
nodes $|N|$ (this easily follows from Proposition~\ref{prop:esistenzamatrioska}), thus $O(|E|\cdot k)=O(|N|\cdot k)$.
Interestingly, even if denial constraints of any arity were allowed, the consistency check
could be still accomplished over the conflict hypertree in polynomial time (that is, 
replacing $k$ with $|N|$, we would get the bound $O(|N|^2)$).

\begin{example}\label{ex:modelChecking}
Consider the PDB schme $\D^p$ consisting of 
relation scheme \emph{Person}$^p($\emph{Name}, \emph{Age}, \emph{Parent}, \emph{Date}, \emph{City}, \emph{P}$)$ representing some personal data obtained by integrating various sources. 
A tuple over \emph{Person}$^p$ refers to a person, and, in particular, attribute \emph{Parent}
references the name of one of the parents of the person, while \emph{City} is the city of residence of the person in the date specified in \emph{Date}.
Consider  the PDB instance $D^p$ consisting of the instance 
\emph{person}$^p$ of \emph{Person}$^p$ shown in Figure~\ref{fig:IperalberoReale}(a).

\begin{figure}[h!]
\centering
\begin{tabular}{cc}
{\small
\begin{tabular}{l||c|c|c|c|c|c||}
  \hhline{~|t:======:t|}
& \raisebox{0cm}[3.2mm][1.1mm]{\textbf{\textit{Name}}} &  
\textbf{\textit{Age}} & 
\textbf{\textit{Parent}} &
\textbf{\textit{Date}} &
\textbf{\textit{City}} &
\textbf{\textit{P}}\\
  \hhline{~|:======:|}
$t_1$ & \raisebox{0cm}[3mm][1mm]{A} & 40 &  B & 2010  & NY  &$p_1$ \\
  \hhline{~||------||}
$t_2$ & \raisebox{0cm}[3mm][1mm]{A}  & 40 & B & 2012  & Rome &$p_2$\\
  \hhline{~||------||}
$t_3$ & \raisebox{0cm}[3mm][1mm]{A}  & 40 & C & 2010  & NY &$p_3$  \\
  \hhline{~||------||}
$t_4$ & \raisebox{0cm}[3mm][1mm]{A}  & 40 & D & 2010  & NY &$p_4$  \\
  \hhline{~||------||}
$t_5$ & \raisebox{0cm}[3mm][1mm]{C} & 30 & E  & 2010  & NY   & $p_5$ \\
\hhline{~|b:======:b|}
\end{tabular}
}
&
\hspace*{10mm}
\raisebox{-.5\height}{\includegraphics[scale=0.9]{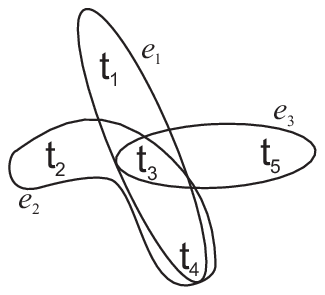}}\\
(a) & (b)\\
\end{tabular}
\caption{(a) PDB instance $D^p$; (b) Conflict hypergraph $HG(D^p,\IC)$}
\label{fig:IperalberoReale}
\end{figure}


Assume that $\IC$ consists of the following constraints defined over \emph{Person}$^p$:

\begin{itemize}
\item[$ic_1$:] 
$\neg \big[$ 
Person$(x_1,y_1,z_1,v_1,w_1)\,\wedge$ 
Person$(x_1,y_2,z_2,v_2,w_2)\,\wedge$ 
Person$(x_1,y_3,z_3,v_3,w_3)\,\wedge 
z_1\!\neq\! z_2 \wedge z_1\!\neq\!z_3 \wedge z_2\!\neq\!z_3
\big]$, 
imposing that no person has more than 2 parents;

\item[$ic_2$:]  
$\neg \big[$ 
Person$(x_1,y_1,z_1,v_1,w_1)\,\wedge$
Person$(z_1,y_2,z_2,v_2,w_2)\,\wedge y_1\!>\!y_2 \big]$, 
imposing that no person is older than any of her parents.
\end{itemize}
The conflict hypergraph $HG(D^p,\IC)$ is shown in Figure~\ref{fig:IperalberoReale}(b).
Here, the conflicting sets $e_1$, $e_2$ are originated by violations of $ic_1$, while $e_3$ 
is originated by the violation of $ic_2$.
It is easy to check that $HG(D^p,\IC)$ is a hypertree. 
In particular, observe that set of intersections of $e_1$ with the other hyper-edges of 
$HG(D^p,\IC)$, that is \emph{Int}$(e_1,HG(D^p,\IC))=\{\{t_3\}, \{t_3, t_4\} \}$, is a
matryoshka. 
Analogously, \emph{Int}$(e_2,HG(D^p,\IC))$ is matryoshka as well. 

Since $HG(D^p,\IC)$ is a hyper-tree, thanks to Theorem~\ref{theo:tractableCC},
we can conclude that $D^p$ is consistent iff the following inequalities hold:
$$p_1+p_3+p_4\leq 2;
\hspace*{1cm}
p_2+p_3+p_4\leq 2;
\hspace*{1cm}
p_3+p_4\leq 1.
$$
%
%
\boxx
\end{example}

Note that the condition of Theorem~\ref{theo:tractableCC} is 
a necessary condition for consistency in the presence of conflict hypergraphs of any shape, 
not necessarily hypertrees (in fact, in 
the proof of the necessary condition of Theorem~\ref{theo:tractableCC}, we did not use the
assumption that the conflict hypergraph is a hypertree).
The following example shows that this condition is not sufficient in general, in particular
when the conflict hypergraph contains ``cycles''.
\begin{example}
Consider the hypergraph $HG(D^p, \IC)$ obtained by augmenting the hypertree
in Figure~\ref{fig:ConflictHG} with the hyperedge $e_5=\{t_8, t_9\}$
(whose presence invalidates the acyclicity of the hypergraph).
Let the probabilities of $t_1,\dots,t_9$ be as follows:
\begin{center}
\begin{tabular}{|c||c|c|c|c|c|c|c|c|c|}
\hline
$t_i$	&$t_1$ 				& $t_2$ 	&$t_3$ 				&$t_4$ 				&$t_5$ 	&$t_6$ 				&$t_7$				 	&$t_8$ 				&$t_9$\\
\hline
$p(t_i)$	& \raisebox{0cm}[4mm][2mm]{$\frac{3}{4}$} 	& $1$		&$\frac{3}{4}$ 	&$\frac{3}{4}$ 	& $1$		&$\frac{1}{2}$ 	&$\frac{1}{2}$ 		&$\frac{1}{2}$	&$\frac{1}{2}$ \\ 
\hline
\end{tabular}
\end{center}
Although the condition of Theorem~\ref{theo:tractableCC} holds for every hyperedge 
$e_i$, with $i$ in $[1..5]$,  
there is no model of $D^p$ w.r.t. $\IC$.
In fact,
the overall probability of the possible worlds containing $t_8$ must be $1/2$;
due to hyperedges $e_4$ and $e_5$, these possible worlds can not contain 
neither $t_6$ nor $t_9$, which must appear together in the remaining possible worlds 
(since the marginal probability of $t_6$ and $t_9$ is
equal to the sum of the probabilities of the possible worlds not containing 
$t_8$);
however, as $t_3$ can not co-exist with both $t_6$ and $t_9$ (due to $e_3$), 
it must be in the worlds containing $t_8$;
but, as the overall probability of these worlds is $1/2$, they are not sufficient 
to make the probability of $t_3$ equal to $3/4$.
\boxx
\end{example} 

\vspace*{-1mm}
\subsubsection{``Cyclic'' hypergraphs: cliques and rings}
\label{sec:TractabilityRings}
An interesting tractable case which holds even in the presence of cycles in 
the conflict hypergraph is when the constraints define \emph{buckets} of tuples: 
buckets are disjoint sets of tuples, such that each pair of tuples in the same bucket are
mutually exclusive.
The conflict hypergraph describing a set of buckets is simply a graph consisting of 
disjoint cliques, each one corresponding to a bucket.
It is straightforward to see that, in this case,
the consistency problem can be decided by just verifying that, 
for each clique, the sum of the probabilities of the tuples in it is not greater than $1$.
Observe that the presence of buckets in the conflict hypergraph can be due to 
key constraints.
Thus, what said above implies that \cc\ is tractable in the presence of keys.
However, we will be back on the tractability of key constraints in the next section, where we will generalize this tractability result to the presence of one FD per relation.

We now state a more interesting tractability result holding in the presence of 
``cycles'' in the conflict hypergraph.

\begin{theorem}
\label{theo:ring}
Given an instance $D^p$ of $\D^p$, if $H(D^p, \IC)=\langle N, E\rangle$ is a ring, 
then $D^p\models \IC$ iff both the following hold:\\[4pt]
1)
$\forall e \in E,\  \sum_{t\in e} p(t) \leq |e|-1$;
\hspace*{1mm}
2)
$\sum_{t \in N} p(t) - |N| + \lceil\frac{|E|}{2}\rceil \leq 0$.
\end{theorem}

Interestingly, Theorem~\ref{theo:ring} states that, when deciding the consistency of tuples arranged as a ring in the conflict hypergraph, it is not sufficient to consider the local 
consistency w.r.t. each hyperedge (as happens in the case of conflict hypertrees), as
also a condition involving all the tuples and hyperdges must hold.
As an application of this result, consider the case that $H(D^p, \IC)$ is 
the ring whose nodes are $t_1$, $t_2$, $t_3$, $t_4$ 
(where: $p(t_1)=p(t_2)=p(t_3)=\nicefrac{1}{2}$ and $p(t_4)=1$),
and whose edges~are: 
$e_1=\{t_1,t_2, t_4\}$, $e_2=\{t_1,t_3, t_4\}$, $e_3=\{t_2,t_3\}$.
It is easy to see that property $1)$ of Theorem~\ref{theo:ring} (which is necessary 
for consistency, as already observed) is satisfied, while property
$2)$ is not (in fact, 
$\sum_{t \in N} p(t)\!-\!|N|\!+\!\left\lceil\frac{|E|}{2}\right\rceil=$ 
$\nicefrac{5}{2}-4+2=\nicefrac{1}{2}>0$), which implies inconsistency.
Note that changing $p(t_4)$ to $\nicefrac{1}{2}$ yields consistency.\\[-4pt]

\noindent
\looseness-1
\textbf{Remark 2.}
\emph{Further tractable cases due to the conflict hypergraph.}
The tractability results given so far can be straightforwardly merged into 
a unique more general result: \cc\ is tractable if the conflict hypergraph consists of maximal 
connected components such that each of them is either a hypertree, a clique, 
or a ring.
In fact, it is easy to see that the consistency can be checked by 
considering the connected components separately.

\subsection{Tractability arising from the syntactic form of the denial constraints}
\label{sec:Syntactic-tractable-cases}

We now address the determination of tractable cases from a different perspective.
That is, rather than searching for other properties of the conflict hypergraph guaranteeing 
that the consistency can be checked in polynomial time, we will search for syntactic 
properties of denial constraints which can be detected without looking at the conflict
hypergraph and which yield the tractability of \cc.
We start from the following result.

\begin{theorem}\label{theo:ic-empty-phi}
If $\IC$ consists of a join-free denial constraint, then \cc\ is in \textit{PTIME}.
In particular, $D^p\models \IC$ iff, for each hyperedge $e$ of $HG(D^p, \IC)$, it holds 
that $\sum_{t\in e}\!p(t)\!\leq\!|e|\!-\!1$.
\end{theorem}

\begin{example}
\label{ex:joinfree}
Consider the PDB scheme consisting of the probabilistic relation scheme
\emph{Employee}$^p($\emph{Name}, \emph{Age},  \emph{Team}, \emph{P}$)$.
This scheme is used to represent some (uncertain) personal information about the employees 
of an enterprise.
The uncertain data were obtained starting from anonymized data, and then estimating sensitive information (such as the names of the employees).
Assume that the PDB instance $D^p$ obtained this way consists of the instance 
\emph{employee}$^p$ of \emph{Employee}$^p$ shown in Figure~\ref{fig:join-free}(a).

\begin{figure}[h!]
\centering
\begin{tabular}{cc}
{\small
\begin{tabular}{l||c|c|c|c||}
  \hhline{~|t:====:t|}
& \raisebox{0cm}[3.2mm][1.1mm]{\textbf{\textit{Name}}} &  
\textbf{\textit{Age}} & 
\textbf{\textit{Team}} &
\textbf{\textit{P}}\\
  \hhline{~|:====:|}
$t_1$ & \raisebox{0cm}[3mm][1mm]{P. Jane} & 35 &   A  &$1$ \\
  \hhline{~||----||}
$t_2$ & \raisebox{0cm}[3mm][1mm]{T. Lisbon}  & 25  & B  &$1$\\
  \hhline{~||----||}
$t_3$ & \raisebox{0cm}[3mm][1mm]{W. Rigsby}  & 40 & B &$1/2$  \\
  \hhline{~||----||}
$t_4$ & \raisebox{0cm}[3mm][1mm]{K. Cho}  & 40 & B &$1/2$  \\
  \hhline{~||----||}
$t_5$ & \raisebox{0cm}[3mm][1mm]{G. Van Pelt} & 22  & C  & $1$ \\
  \hhline{~||----||}
$t_6$ & \raisebox{0cm}[3mm][1mm]{G. Bertram}  & 40  & C  &$1/2$\\
  \hhline{~||----||}
$t_7$ & \raisebox{0cm}[3mm][1mm]{R. John}  & 40  & C  &$1/2$\\
\hhline{~|b:====:b|}
\end{tabular}
}
&\hspace*{5mm}
\raisebox{-.5\height}{\includegraphics[scale=0.9]{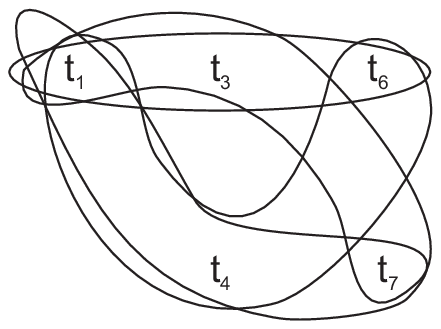}}\\
(a) & (b)\\
\end{tabular}
\caption{(a) PDB instance $D^p$; (b) Conflict hypergraph $HG(D^p,\IC)$}
\label{fig:join-free}
\end{figure}

From some knowledge of the domain, it is known that at least one team among 
`\emph{A}', `\emph{B}', `\emph{C}' consists of only young employees, i.e., employees at most
$30$-year old.
This corresponds to considering $\IC=\{ic\}$ as the set of denial constraints, 
where $ic$ is as follows:

$ic:\ \neg\, \big[$
\emph{Employee}$(x_1, x_2,$ `\emph{A}' $)$
$\wedge$
\emph{Employee}$(x_3, x_4,$ `\emph{B}' $)$
$\wedge$
\emph{Employee}$(x_5, x_6,$ `\emph{C}' $)$
$\wedge$
$x_2\!>\!30 \wedge x_4\!>\!30 \wedge x_6\!>\!30\ \big]$.

It is easy to see that $ic$ is a join-free denial constraint, thus the consistency of 
$D^p$ can be decided using Theorem~\ref{theo:ic-empty-phi}.
In particular, since $HG(D^p,\IC)$ is the hypergraph depicted in Figure~\ref{fig:join-free}(b),
we have that $D^p$ is consistent if and only if the following inequalities hold:
$$
p(t_1)+ p(t_3)+ p(t_6)\leq 2; \hspace*{7mm}
p(t_1)+ p(t_3)+ p(t_7)\leq 2; \hspace*{7mm}
p(t_1)+ p(t_4)+ p(t_6)\leq 2; \hspace*{7mm}
p(t_1)+ p(t_4)+ p(t_7)\leq 2; \hspace*{7mm}
$$
As a matter of fact, all these inequalities are satisfied, thus the considered PDB is consistent.
In fact, there is a unique model $Pr$ for $D^p$ w.r.t. $\IC$.
In particular, $Pr$ assigns probability $1/2$ to each of the possible worlds
$w_1=\{t_1, t_2, t_3, t_4, t_5 \}$ and $w_2=\{t_1, t_2, t_5, t_6, t_7\}$, and probability $0$ 
to all the other possible worlds.
\boxx
\end{example}

The result of Theorem~\ref{theo:ic-empty-phi} strengthens what already observed in the 
previous section:
the arity of constraints is not, \emph{per se}, a source of complexity.
In what follows, we show that the arity can become a source of complexity when combined with the presence of join conditions.

\begin{theorem}\label{theo:ic-arity-3}
There is an $\IC$ consisting of a non-join-free denial constraint of arity $3$ 
such that \cc\ is $N\!P$-hard. 
\end{theorem}

%

Still, one may be interested in what happens to the complexity of \cc\ 
for denial constraints containing joins and having arity strictly lower than $3$. 
In particular, since in the proof of Theorem~\ref{theo:ic-arity-3} we exploit
a ternary EGD to show the $N\!P$-hardness of \cc\ in the presence of ternary constraints with joins (see~\ref{app:567}), it is worth investigating what happens when only binary
EGDs are considered, which are denial constraints with arity $2$ containing joins.
The following theorem addresses this case, and states that \cc\ becomes tractable for any $\IC$
consisting of a binary EGD.


\begin{theorem}\label{theo:EGDs-cc}
If $\IC$ consists of a binary EGD, then \cc\ is in \textit{PTIME}.
\end{theorem}

Differently from the previous theorems on the tractability of \cc, in the statement of 
Theorem~\ref{theo:EGDs-cc}, for the sake of presentation, we have not explicitly reported 
the necessary and sufficient conditions for consistency. 
In fact, in this setting, deciding on the consistency requires reasoning by cases, and then 
checking some conditions which are not easy to be defined compactly. 
However, these conditions can be checked in polynomial time, and the interested reader can 
find their formal definition in the proof of Theorem~\ref{theo:EGDs-cc} (see~\ref{app:567}).

Binary EGDs can be viewed as a generalization of FDs, involving pairs of tuples possibly 
belonging to different relations.
For instance, over the relation schemes 
\emph{Student}$($\emph{Name}, \emph{Address}, \emph{University}$)$ and
\emph{Employee}$($\emph{Name}, \emph{Address}, \emph{Firm}$)$, 
the binary EGD 
$\neg\, \big[$\emph{Student}$(x_1,x_2,x_3)\, \wedge$ \emph{Employee}$(x_1,x_3,x_4)
\wedge\ x_2\!\neq\!x_3\ \big]$
imposes that if a student and an employee are the same person (i.e., they have the same name),
then they must have the same address.
Thus, an immediate consequence of Theorem~\ref{theo:EGDs-cc} is that \cc\ is tractable in the presence of a single FD.

The results presented so far refer to cases where $\IC$ consists of a single denial 
constraint.
We now devote our attention to the case that $\IC$ is not a singleton. 
In particular, the last tractability result makes the following question arise: 
``\emph{Is \cc\ still tractable when $\IC$ contains several binary EGDs?''}.
(Obviously, we do not consider the case of multiple EGDs of any arity, as 
Theorem~\ref{theo:ic-arity-3} states that \cc\ is already hard if $\IC$ merely contains 
one constraint of this form.)
The following theorem provides a negative answer to this question, as it states that \cc\ can be intractable even in the simple case that $\IC$ consists of just two FDs (as recalled above, 
FDs are special cases of binary EGDs).


\begin{theorem}\label{theo:2FDs-cc}
There is an $\IC$ consisting of $2$ FDs over the same relation scheme 
such that \cc\ is $NP$-hard.
\end{theorem}

However, the source of complexity in the case of two FDs is that they are defined over the 
same relation (see the proof of Theorem~\ref{theo:2FDs-cc} in~\ref{app:567}).
As a matter of fact, the following theorem states that all the tractability results stated in this section in the presence of only one denial constraint can be extended to the case of multiple denial constraints defined over disjoint sets of relations.
Intuitively enough, this derives from the fact that, if the denial constraints involve 
disjoint sets of relation, the overall consistency can be checked by considering the constraints separately. 

\begin{theorem}\label{theo:tract-EGDs-and-join-free}
Let each denial constraint in $\IC$ be join-free or a BEGD.
If, for each pair of distinct constraints $ic_1$,$ic_2$ in $\IC$, the relation names occurring in $ic_1$ are distinct from those in $ic_2$, then \cc\ is in \textit{PTIME}.
\end{theorem}

Hence, the above theorem entails that
\cc\ is tractable in the interesting case that $\IC$ consists of one FD per relation.
In the following theorem, we elaborate more on this case, and specify necessary and sufficient conditions which can be checked to decide the consistency.


\begin{theorem}\label{cor:FD-cc}
If $\IC$ consists of one FD per relation, then $HG(D^p,\IC)$ is a graph where
each connected component is either a singleton or a complete multipartite graph.
Moreover, $D^p$ is consistent w.r.t. $\IC$ iff the following property holds:
for each connected component $C$ of $HG(D^p,\IC)$, denoting the maximal independent sets 
of $C$ as $S_1, \dots, S_k$, it is the case that
$\sum_{i\in[1..k]} \tilde{p}_i\leq 1$, where $\tilde{p}_i=\max_{t\in S_i} p(t)$.
\end{theorem}

We recall that a complete multipartite graph is a graph whose nodes can be partitioned into 
sets such that an edge exists if and only if it connects two nodes 
belonging to distinct sets. 
Each of these sets is a maximal independent set of nodes.
For instance, the portion of the graph in Figure~\ref{fig:HG1FDxRelation}(b) containing 
only the nodes $t_1$, $t_2$, $t_3$, $t_4$, $t_5$ is a complete multipartite graph whose
maximal independent sets are  $S_1=\{t_1, t_2\}$, $S_2=\{t_3, t_4\}$, $S_3=\{t_5\}$.
The following example shows an application of Theorem~\ref{cor:FD-cc}.

\begin{example}
\label{ex:FDcons}
Consider the PDB scheme consisting of the probabilistic relation scheme
\emph{Person}$^p($\emph{Name}, \emph{City}, \emph{State}, \emph{P}$)$,
and its instance $D^p$ consisting of the instance 
\emph{person}$^p$ of \emph{Person}$^p$ shown in Figure~\ref{fig:HG1FDxRelation}(a).

\begin{figure}[h!]
\centering
\begin{tabular}{cc}
{\small
\begin{tabular}{l||c|c|c|c||}
\hhline{~|t:====:t|}
& \raisebox{0cm}[3.2mm][1.1mm]{\textbf{\textit{Name}}} &  
\textbf{\textit{City}} & 
\textbf{\textit{State}} &
\textbf{\textit{P}}\\
\hhline{~|:====:|}
$t_1$ & \raisebox{0cm}[3mm][1mm]{B. Van de Kamp} & Sioux City &   IA  &$1/2$ \\
\hhline{~||----||}
$t_2$ & \raisebox{0cm}[3mm][1mm]{S. Delfino}  & Sioux City &  IA &$1/4$\\
\hhline{~||----||}
$t_3$ & \raisebox{0cm}[3mm][1mm]{L. Scavo}  & Sioux City &   NE  &$1/4$  \\
\hhline{~||----||}
$t_4$ & \raisebox{0cm}[3mm][1mm]{G. Solis}  & Sioux City &   NE  &$1/4$  \\
\hhline{~||----||}
$t_5$ & \raisebox{0cm}[3mm][1mm]{E. Britt} & Sioux City &   SD   & $1/4$ \\
\hhline{~||----||}
$t_6$ & \raisebox{0cm}[3mm][1mm]{K. Mayfair}  & Baltimore &   MD   &$3/4$\\
\hhline{~||----||}
$t_7$ & \raisebox{0cm}[3mm][1mm]{R. Perry}  & Fargo  & ND  &$3/4$\\  
\hhline{~||----||}
$t_8$ & \raisebox{0cm}[3mm][1mm]{M. A. Young}  & Fargo  & ND  &$1/4$\\
\hhline{~||----||}
$t_9$ & \raisebox{0cm}[3mm][1mm]{K. McCluskey}  & Fargo  & MN  &$1/4$\\
\hhline{~|b:====:b|}
\end{tabular}
}
&
\raisebox{-.5\height}{\includegraphics[scale=0.8]{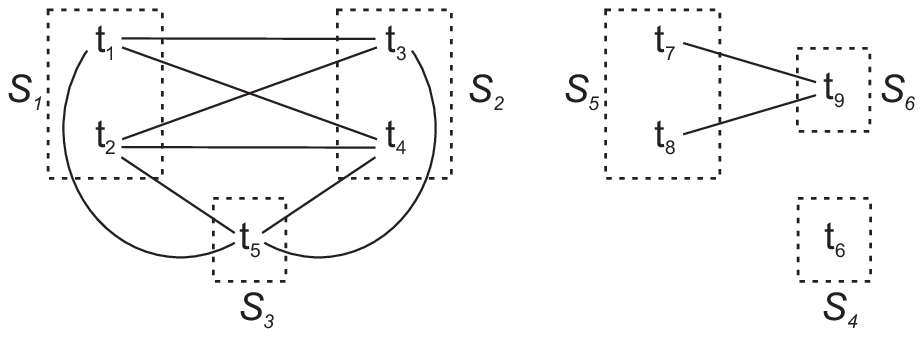}}\\
(a) & (b)\\
\end{tabular}
\caption{(a) PDB instance $D^p$; (b) Conflict hypergraph $HG(D^p,\IC)$}
\label{fig:HG1FDxRelation}
\end{figure}

Consider the FD $ic$: \emph{City} $\rightarrow$ \emph{State}, which can be rewritten as
$\neg \big[$ \emph{Person}$(x_1,x_2,x_3)\,\wedge$ \emph{Person}$(x_4,x_2,x_5)\,\wedge x_3\!\neq\!x_5 \, \big]$.
The conflict hypergraph $HG(D^p,\IC)$ is the graph depicted in Figure~\ref{fig:HG1FDxRelation}(b).
It consists of $3$ connected components: 
one of them is a singleton (and corresponds to the maximal independent set $S_4$), 
and the other two are the complete multipartite graphs over
the maximal independent sets $S_1,S_2,S_3$ and $S_5,S_6$, respectively.
Theorem~\ref{cor:FD-cc} says that $D^p$ is consistent if and only if the 
following three inequalities (one for each connected component of $HG(D^p,\IC)$) hold:
$$
\max\{p(t_1), p(t_2)\} + \max\{p(t_3), p(t_4)\} + p(t_5)\leq 1; \hspace*{10mm}
p(t_6)\leq 1; \hspace*{10mm}
\max\{p(t_7), p(t_8)\} + p(t_9)\leq 1.
$$

As a matter of fact, all these inequalities are satisfied, thus the considered PDB is consistent.
In fact, there is a model $M$ for $D^p$ w.r.t. $\IC$ 
assigning probability $1/4$ to each of the possible worlds
$w_1=\{t_1, t_2, t_6, t_7, t_8 \}$,
$w_2=\{t_1, t_6, t_7\}$, 
$w_3=\{t_3, t_4, t_6, t_7\}$, and
$w_4=\{t_5, t_9\}$,
and probability $0$ to all the other possible worlds.
The reader can easily check that there are models for $D^p$ w.r.t. $\IC$ other than $M$.
\boxx
\end{example}

\subsection{Tractability implied by conflict-hypergraph properties vs. tractability implied by syntactic forms.}
The tractability results stated in sections~\ref{sec:tractablecases} 
and~\ref{sec:Syntactic-tractable-cases} can be viewed as complimentary to each other.
In fact, an instance of \cc\ may turn out to be tractable due the syntactic form of the constraints, even if the shape of the conflict hypergraph is none of those ensuring tractability, and vice versa.
For instance, in the case that $\IC$ consists of a join-free denial constraint or a binary EGD, 
it is easy to see that the conflict hypergraph may not be a hypertree or a ring, but \cc\
is nevertheless tractable due to theorems~\ref{theo:tractableCC} and~\ref{theo:ring}.
Vice versa, if $\IC$ contains two FDs per relation or a ternary denial constraints with joins
(which, potentially, are hard cases, due to theorems~\ref{theo:ic-arity-3} 
and~\ref{theo:2FDs-cc}), \cc\ may turn out to be tractable, if the way the data combine with 
the constraints yields a conflict hypergraph which is a hypertree or a ring (see theorems~\ref{theo:tractableCC} and~\ref{theo:ring}).

On the whole, the tractability results presented in sections~\ref{sec:tractablecases} and~\ref{sec:Syntactic-tractable-cases} can be used conjunctively when addressing \cc :
for instance, one can start by examining the constraints and check whether they conform 
to a tractable syntactic form, and, if this is not the case, one can look at the conflict hypergraph and check whether its structure entails tractability.

\section{Querying PDB\lowercase{s} under constraints}
\label{sec:queryevaluation}
As explained in the previous section, given a PDB $D^p$ 
in the presence of a set $\IC$ of integrity constraints, not all the interpretations of $D^p$
are necessarily models w.r.t. $\IC$.
If $D^p$ is consistent w.r.t. $\IC$, there may be exactly one model (Case 2 of the motivating example), or more (Case 3 of the same example).
In the latter case, given that all the models satisfy all the constraints in $\IC$, there is 
no reason to assume one model more reasonable than the others
(at least in the absence of other knowledge not encoded in the constraints).
Hence, when querying $D^p$, it is ``cautious'' to answer to queries by taking into account
all the possible models for $D^p$ w.r.t. $\IC$.
In this section, we follow this argument and introduce a \emph{cautious querying paradigm}
for conjunctive queries, where query answers consist of tuples associated with probability ranges:
given a query $Q$ posed over $D^p$, the range associated with a tuple $t$ in the answer of $Q$
contains every probability with which $t$ would be returned as an answer of $Q$
if $Q$ were evaluated separately on every model of $D^p$.
In what follows, we first introduce the formal definition of conjunctive query in the 
probabilistic setting, and introduce its semantics according to the above-discussed cautious
paradigm.
Then, we provide our contributions on the characterization of the problem of computing query
answers.

A (conjunctive) query over a PDB schema $\D^p$ is written as a (conjunctive)
query over its deterministic part $det(\D^p)$. 
Thus, it is an expression of the form:\\
$Q(\vec x)=\exists \vec z.\ R_1(\vec{y}_1)\wedge\dots\wedge R_m(\vec{y}_m)\wedge\phi(\vec{y}_1,\dots,\vec{y}_m)$,
where:\vspace*{-1mm}
\begin{list}{--}
     {\setlength{\rightmargin}{0mm}
      \setlength{\leftmargin}{3mm}
      \setlength{\itemindent}{-1mm} 
			\setlength{\itemsep}{-1mm}}
\item
$R_1,\dots,R_m$ are name of relations in $det(\D^p)$;
\item
$\vec{x}$ and $\vec{z}$ are tuples of variables, having no variable common;
\item 
$\vec{y}_1,\!\dots,\!\vec{y}_m$ are tuples of variables and constants such that 
every variable in any $\vec{y}_i$ occurs in either $\vec{x}$ or $\vec{z}$, and vice versa;
\item
$\phi(\vec{y}_1,\dots,\vec{y}_m)$ is a conjunction of built-in predicates, each
of the form 
$\alpha \diamond \beta$, where $\alpha$ and $\beta$ are either variables in 
$\vec{y}_1,\dots,\vec{y}_m$ or constants, and 
$\diamond \in \{=, \neq, \leq, \geq, <, > \}$.
\end{list}

A query $Q$ will be said to be \textit{projection-free} if $\vec{z}$ is empty.

The semantics of a query $Q$ over a PDB $D^p$
in the presence of a set of integrity constraints $\IC$ is given in two steps.
First, we define the answer of $Q$ w.r.t. a single model $M$ of $D^p$.
Then, we define the  answer of $Q$ w.r.t. $D^p$, which summarizes all the answers of $Q$ 
obtained by separately evaluating $Q$ over every model of $D^p$ .
Obviously, we rely on the assumption that $D^p$ is consistent w.r.t. $\IC$, thus $\M(D^p,\IC)$
is not empty.

The answer of $Q$ over a model $M$ of $D^p$ w.r.t. $\IC$ 
is the set \emph{Ans}$^M\!(Q,D^p,\IC)$ of pairs of the form 
$\la\vec{t}, p^M_Q(\vec{t}) \ra$  such that:\vspace*{-1.5mm}
\begin{list}{--}
     {\setlength{\rightmargin}{0mm}
      \setlength{\leftmargin}{3mm}
      \setlength{\itemindent}{-1mm} 
			\setlength{\itemsep}{0mm}}
\item
$\vec{t}$ is a ground tuple such that  $\exists w\!\in\! \mbox{\emph{pwd}}(D^p)\mbox{\emph{ s.t. }} w\models Q(\vec{t})$;
\item
$p^M_Q(\vec{t})=
\sum_{
w\in pwd(D^p) \wedge\\
w\models Q(\vec{t})
} 
M(w)$ is the overall probability of the possible 
worlds where $Q(\vec{t})$ evaluates to true,
\end{list}
\vspace*{-1mm}
\noindent where $w\models Q(\vec{t})$ denotes that $Q(\vec{t})$ evaluates to true in $w$.

\vspace*{1mm}
In general, there may be several models for $D^p$, 
and the same tuple $\vec{t}$ may have different probabilities in the answers evaluated
over different models.
Thus, the overall answer of $Q$ over $D^p$ is defined in what follows as 
a summarization of all the answers of $Q$ over all the models of $D^p$.
\begin{definition}[Query answer]
\label{def:queryanswer}
Let  $Q$ be a query over $\D^p$,
and $D^p$ an instance of $\D^p$.
The answer of $Q$ over $D^p$ is the set \emph{Ans}$(Q,D^p,\IC)$ of pairs 
$\la \vec{t}, [p^{min}, p^{max}] \ra$,
where:
\begin{list}{--}
     {\setlength{\rightmargin}{0mm}
      \setlength{\leftmargin}{3mm}
      \setlength{\itemindent}{-2mm} 
			\setlength{\itemsep}{0mm}}
\item
$\exists M\!\in\!\M(D^p,\IC)$ s.t. $\vec{t}$ is a tuple in \emph{Ans}$^M(Q,D^p,\IC)$;
\item
$p^{\mbox{\scriptsize{min}}}\!=\!\!\!\!\!\begin{array}{cl}
\min\vspace*{-1mm}&\hspace*{-4mm} \left\{ p^M_Q(\vec{t}) \right\},\\ 
\mbox{\scriptsize$M\!\!\in\!\!\M(D^p\!,\!\IC)$}&\\
\end{array}$\ \ \ \ \
$p^{\mbox{\scriptsize{max}}}\!=\!\!\!\!\!\begin{array}{cl}
\max\vspace*{-1mm}&\hspace*{-4mm} \left\{ p^M_Q(\vec{t}) \right\}.\\  
\mbox{\scriptsize$M\!\!\in\!\!\M(D^p\!,\!\IC)$}&\\
\end{array}$
\end{list}
\end{definition}
Hence, each tuple $\vec{t}$ in \emph{Ans}$(Q,D^p,\IC)$ is associated with an interval 
$[p^{\min}, p^{\max}]$, whose extremes are, respectively,
the minimum and maximum probability of $\vec{t}$ in the answers of $Q$ over 
the models of $D^p$.
Examples of answers of a query are reported in the motivating example.
In the following, we say that $\vec{t}$ is an answer of $Q$ 
with minimum and maximum probabilities $p^{\min}$ and $p^{\max}$ if 
$\langle \vec{t}, [p^{\min}, p^{\max}]\rangle \in$\emph{Ans}$(Q,D^p,\IC)$.



The following proposition gives an insight on the semantics of query answers, as it better 
explains the meaning of the probability range associated with
each tuple occurring in the set of answers of a query.
That is, it states that,  taken any pair $\la \vec{t}, [p^{\min}, p^{\max}] \ra$ in 
\emph{Ans}$(Q,D^p,\IC)$, every value $p$ inside the interval $[p^{\min}, p^{\max}]$ is 
``meaningful'', in the sense that there is at least one model for which $\vec{t}$ 
is an answer of $Q$ with probability $p$.
Considering this property along the fact that the boundaries $p^{\min}, p^{\max}$ 
are the minimum and maximum probabilities of $\vec{t}$ as an answer of $Q$ (which follows
from Definition~\ref{def:queryanswer}), we have that $[p^{\min}, p^{\max}]$ is the tightest
interval containing all the probabilities of $\vec{t}$ as an answer of $Q$, and is dense
(every value inside it corresponds to a probability of $\vec{t}$ as an answer of $Q$).

\begin{proposition}
\label{pro:intquery}
Let  $Q$ be a query over $\D^p$, and $D^p$ an instance of $\D^p$.
For each pair $\la \vec{t}, [p^{\min}, p^{\max}] \ra$ in 
\emph{Ans}$(Q,D^p,\IC)$, and each probability value $p\in [p^{\min}, p^{\max}]$, 
there is a model $M$ of $D^p$ w.r.t. $\IC$ such that 
$\la \vec{t}, p \ra \in$ \emph{Ans}$^M(Q,D^p,\IC)$.
\end{proposition}

\begin{proof}
We first introduce a system $S(\D^p,\IC, D^p)$ of linear (in)equalities whose solutions
one-to-one correspond to the models of $D^p$ w.r.t. $\IC$.
For every  $w_i\in pwd(D^P)$, let $v_i$ be a variable ranging over the domain 
of rational numbers.
The variable $v_i$ will be used to represent the probability assigned to $w_i$ by an
interpretation of $D^p$.
The system of linear (in)equalities $S(\D^p,\IC, D^p)$ is as follows:
$$
\left\{
\begin{array}{lr}
\forall t\in D^p,\  \sum_{i | w_i\in pwd(D^p)\wedge t\in w_i} v_i=p(t)& (e1)\\
\sum_{i | w_i\in pwd(D^p)\wedge w_i\not\models \IC} v_i=0 & (e2)\\
\sum_{i | w_i\in pwd(D^p)} v_i=1 & (e3)\\
\forall w_i\in pwd(D^p),\ v_i\geq 0 & (e4)\\
\end{array}\\
\right.
$$
The first $|D^p|$ equalities $(e1)$ in $S(\D^p,\IC, D^p)$ encode the fact that, 
for each tuple $t$ in the PDB instance, the sum of the 
probabilities assigned to the worlds containing the tuple $t$ must be equal to 
the marginal probability of $t$. 
The subsequent two equalities $(e2)$, $(e3)$, along with the inequalities $(e4)$ imposing 
that the probabilities $v_i$ assigned to each possible world are non-negative,
entail that the probability assigned to any world violating $\IC$ is $0$,
as well as that the probabilities assigned to all the possible worlds sum up to $1$.

It is easy to see that every solution $s$ of $S(\D^p,\IC, D^p)$ one-to-one corresponds to 
a model $Pr$ for $D^p$ w.r.t. $\IC$, where $Pr(w_i)$ is equal to $v_i[s]$, i.e.,
the value of $v_i$ in~$s$.

We now consider the system of linear (in)equalities 
$S^*(\D^p,\IC,D^p)$
obtained by augmenting the set of (in)equalities in 
$S(\D^p,\IC,D^p)$
with the following equality:
$$v^* =\sum_{i | w_i\in pwd(D^p)\wedge w_i\models \vec{t}} v_i$$
where $v^*$ is a new variable symbol not appearing in $S(\D^p,\IC,$ $D^p)$.

Obviously, every solution $s$ of $S^*(\D^p,\IC,D^p)$ still one-to-one corresponds to 
a model $Pr$ for $D^p$ w.r.t. $\IC$ such that, for each possible world $w_i\in pwd(D^p)$,
$Pr(w_i)$ is equal to $v_i[s]$, and $v^*[s]$ (the value of $v^*$ in $s$) is equal 
to the sum of the probabilities assigned by $Pr$ to the possible worlds 
where $\vec{t}$ is an answer of $Q$.
Therefore, $p^{\min}$ (resp. $p^{\max}$) is the solution of the following
linear programming problem $LP(S^*)$:
$$
\begin{array}{l}
\mbox{minimize (resp. maximize) } v^*\\
\mbox{subject to\ } 
S^*(\D^p,\IC,D^p)
\end{array}$$

Since the feasible region shared by the min- and max- variants of $LP(S^*)$ is defined by linear inequalities only, it follows that it is a convex polyhedron.
Hence, the following well-known result~\cite{PapOpt} can be exploited:
``\emph{given two linear programming problem $LP_1$ and $LP_2$
minimizing and maximizing the same objective function $f$ over the same convex feasible region $S$, respectively,
it is the case that 
for any value $v$ belonging to the interval $[v^{\min}, v^{\max}]$, whose extreme values 
are the optimal solutions of $LP_1$ and $LP_2$, respectively,
there is a solution $s$ of $S$ such that $v$ is the value taken by $f$ when evaluated over 
$s$}''.   
This result entails that, for every probability value 
$p\in [p^{\min}, p^{\max}]$ taken by the objective function $v^*$ of $LP(S^*)$,
there is a feasible solution $s$ of $S^*(\D^p,\IC,D^p)$ such that $p=v^*[s]$.
Hence, the statement follows from the fact that every solution of $S^*(\D^p,\IC, D^p)$ one-to-one corresponds to a model for $D^p$ w.r.t. $\IC$.
\end{proof}

The definition of query answers with associated ranges is reminiscent of the
treatment of aggregate queries in inconsistent databases~\cite{ArenasBCHRS03}.
In that framework, the consistent answer of an aggregate query \emph{Agg} is a range 
$[v_1, v_2]$, whose boundaries represent the minimum and maximum answer which would be 
obtained by evaluating \emph{Agg} on at least one repair of the database.
However, the consistent answer is not, in general, a dense interval: for instance,
it can happen that there are only two repairs, one corresponding to $v_1$ and one to $v_2$,
while the values between $v_1$ and $v_2$ can not be obtained as answers on any repair.

In the rest of this section, 
we address the evaluation of queries from two standpoints:
we first consider a decision version of the query answering problem,
and then we investigate the query evaluation as a search problem.
In the following, besides assuming that a database schema $\D^p$ and a set of constraints
$\IC$ of fixed size are given, we also assume that queries over $\D^p$ are of fixed 
size.
Thus, all the complexity results refer to data complexity. 

\subsection{Querying as a decision problem}
In the classical ``deterministic'' relational setting, the decision version of the 
query answering problem is commonly defined as the \emph{membership problem} of 
deciding whether a given tuple belongs to the answer of a given query.
In our scenario, tuples belong to query answers with some probability range, thus it 
is natural to extend this definition to our probabilistic setting in the following way.
\begin{definition}[Membership Problem (\mp)]
\label{def:mp}
Given a 
query $Q$ over $\D^p$, an instance $D^p$ of $\D^p$, 
a ground tuple $\vec{t}$, and the constants $k_1$ and $k_2$ (with $0\!\!\leq\!\!k_1\!\!\leq\!\! k_2\!\!\leq\!\! 1$), 
the membership problem is deciding whether $\vec{t}$ is an answer of $Q$ with 
minimum and maximum probabilities $p^{\mbox{\scriptsize{min}}}$ and $p^{\mbox{\scriptsize{max}}}$ 
such that 
$p^{\mbox{\scriptsize{min}}}\!\geq\!k_1$ and $p^{\mbox{\scriptsize{max}}}\!\leq\!k_2$.
\end{definition}


Hence, solving \mp\ can be used to decide whether a given tuple is an answer with a 
probability which is at least $k_1$ and not greater than $k_2$.
Observe that Definition~\ref{def:mp} collapses to the classical definition of 
membership problem when data are deterministic: 
in fact, asking whether a tuple belongs to the answer of a query posed over a deterministic
database corresponds to solving \mp\ over the same database with $k_1\!=\!k_2\!=\!1$.


From the results in \cite{Luk01}, where an entailment problem more general than 
\mp\ was shown to be in co$N\!P$ (see Section \ref{sec:relatedwork}), 
it can be easily derived that \mp\ is in co$N\!P$ as well. 
The next theorems (which are preceded by a preliminary lemma) determine two cases when this 
upper bound on the complexity is tight.

\begin{lemma}
\label{boundedpmin-pmax}
Let $Q$ be a conjunctive query over $\D^p$, $D^p$ an instance of $\D^p$, and 
$\vec{t}$ an answer of $Q$ having minimum probability $p^{\min}$ and maximum probability 
$p^{\max}$.
Let $m$ be the number of tuples in $D^p$ plus $3$ and $a$ be the maximum among the 
numerators and denominators of the probabilities of the tuples in $D^p$.
Then 
$p^{\min}$ and $p^{\max}$ are expressible as fractions of the form $\frac{\eta}{\delta}$,
with $0\leq \eta \leq (m a)^m$ and $0< \delta \leq (m a)^m$.
\end{lemma}
\begin{theorem}[Lower bound of \mp]
\label{theo:mphardness1}
There is at least one conjunctive query containing projection for which \mp\ is co$NP$-hard, 
even if $\IC$ is empty.
\end{theorem}
\begin{proof}
We show a LOGSPACE reduction from the consistency checking problem (\cc)
in the presence of binary denial constraints, 
which is $NP$-hard (see Theorem~\ref{theo:2FDs-cc}), to the 
complement of the membership problem ($\nmp$).

Let $\la \D^p_\cc,\IC_\cc,D^p_\cc \ra$ be an instance of \cc.
We construct an equivalent instance 
$\la \D^p_{\nmp},\IC_{\nmp},D^p_{\nmp}, Q, t_\emptyset, k_1, k_2 \ra$
of $\nmp$ as follows.
\begin{list}{--}
     {\setlength{\rightmargin}{0mm}
      \setlength{\leftmargin}{3mm}
      \setlength{\itemindent}{-1mm} 
			\setlength{\itemsep}{0mm}}
\item
$\D^p_{\nmp}$ consists of relation schemas $R^p(tid,P)$ and 
$S^p(tid_1,$ $tid_2,P)$;
\item
$\IC_{\nmp}=\emptyset$, that is, no constraint is assumed on $\D^p_{\nmp}$;
\item
$D^p_{\nmp}$ is the instance of $\D^p_{\nmp}$ which contains, for each tuple
$t\in D^p_\cc$, the tuple $R^p(id(t), p(t))$, where $id(t)$ is a unique identifier 
associated to the tuple $t$. 
Moreover, $D^p_{\nmp}$ contains, for each pair of 
tuples $t_1,t_2$ in $D^p_\cc$ which are conflicting w.r.t. $\IC_\cc$, the tuple 
$S^p(id(t_1),id(t_2), 1)$.
\item
$Q = \exists x,y\ R(x)\wedge R(y)\wedge S(x,y)$;
\item
$t_\emptyset$ is the empty tuple;
\item
the lower bound $k_1$ of the minimum probability of $t_\emptyset$ as answer of $Q$ is 
set equal to $k_1=\frac{1}{(m a)^m}$, where $m$ is the number of tuples in $D^p_{\nmp}$ 
plus $3$, and $a$ the maximum among the numerators and denominators of the probabilities 
of the tuples in $D^p_{\nmp}$;
\item
the upper bound $k_2$ of the maximum probability of $t_\emptyset$ as answer of $Q$ is 
set equal to $1$.
\end{list}

Obviously, the $\nmp$ instance returns true iff the minimum probability that 
$t_\emptyset$ is an answer to $Q$ over $D^p_{\nmp}$ is (strictly) less than $k_1$.

It is easy to see that every interpretation of $D^p_\cc$ (the database in the \cc\ instance)
corresponds to a unique interpretation of $D^p_{\nmp}$ (the database in the $\nmp$ instance),
and vice versa.
Observe that $D^p_{\nmp}$ is consistent, since the set of constraints considered in the $\nmp$
instance is empty.

We show now that the above-considered \cc\ and $\nmp$ instances are equivalent, that is,
the \cc\ instance is true iff the $\nmp$ instance is true.
On the one hand, if the \cc\ instance is true, then there is at least is one model
$Pr_{\cc}$ for $D^p_\cc$ w.r.t. $\IC_\cc$ (that is, $Pr_{\cc}$ assigns probability $0$ 
to every possible world $w$ which contains tuples which are conflicting according to 
$\IC_\cc$).
It is easy to see that evaluating $Q$ on the corresponding interpretation $Pr_{\nmp}$ of 
$\nmp$ yields probability $0$ for the empty tuple $t_\emptyset$. 
Hence, the $\nmp$ instance is true in this case.

On the other hand, if the $\nmp$ instance is true, then the minimum probability that
$t_\emptyset$ is an answer of $Q$ must be less than $\frac{1}{(m a)^m}$.
Since $\frac{1}{(m a)^m}$ is the smallest non-zero value that can be assumed by
the minimum probability of $t_\emptyset$ (see Lemma~\ref{boundedpmin-pmax}), this implies 
that the minimum probability that $t_\emptyset$ is an answer of $Q$ is $0$.
This means that there is a model $Pr_{\nmp}$ that assigns probability $0$ to every possible world
$w$ which contains three tuples $R(x_1)$, $R(y_1)$ and $S(x_2,y_2)$ with $x_1=x_2$ and $y_1=y_2$.
It is easy to see that the corresponding interpretation $Pr_{\cc}$ 
is a model for $D^p_\cc$ w.r.t. $\IC_\cc$, as
it assigns probability $0$ to every possible world which contains conflicting tuples.
Hence the \cc\ instance is true in this case.
\end{proof}

The above theorem establishes that the type of the query, 
and in particular that fact that it contains projection, 
is an important source of complexity making \mp\ hard, 
irrespectively of the constraints considered.
For projection-free queries, the next theorem states that \mp\ remains hard even if only binary constraints are considered.
\begin{theorem}[Lower bound of \mp]
\label{theo:mphardness2}
There is at least one projection-free conjunctive query and a set $\IC$ consisting of 
only binary constraints for which \mp\ is co$NP$-hard.
\end{theorem}

We recall that, when addressing \mp, we assume that the database 
is consistent w.r.t. the constraints.
Thus, the hardness results for \mp\ do not derive from any source 
of complexity inherited by \mp\ from \cc.
On the whole, theorems \ref{theo:mphardness1} and \ref{theo:mphardness2} 
suggest that \mp\ has at least two sources of complexity: the type of query 
(the fact that the query contains projection or not),  and the form 
of the constraints.

Once some sources of complexity of \mp\ have been identified, the problem is  worth addressing of determining tractable cases.
We defer this issue after the characterization of the query evaluation as a search problem, 
since, as it will be clearer in what follows, the conditions yielding tractability of the 
latter problem also ensure the tractability of \mp.

\subsection{Querying as a search problem}
\label{sec:queryassearch}
Viewed a search problem, the query answering problem (\qa) is the problem of computing the
set \emph{Ans}$(Q,D^p,\IC)$.
The complexity of this problem is characterized as follows.

\vspace*{-1mm}\begin{theorem}
\label{theo:fp-at-np}
\qa\ is in $F\!P^{N\!P}$ and is $F\!P^{N\!P [\log n]}$-hard.
\end{theorem}

\vspace*{-1mm}The fact that \qa\ is in $F\!P^{N\!P}$ means that our ``cautious'' query evaluation paradigm 
is not more complex than the query evaluation based on the independence assumption, 
which has been shown in~\cite{Suciu04} to be complete for $\#P$ (which strictly 
contains $F\!P^{N\!P}$, assuming $P\!\neq\!N\!P$).
On the other hand, the hardness for $F\!P^{N\!P[\log n]}$ is interesting also because it tightens 
the characterization given in~\cite{Luk01} of the more general entailment problem 
for probabilistic logic programs containing a general form of probabilistic rules
(\emph{conditional rules}).
Specifically, in~\cite{Luk01}, the above-mentioned entailment problem was shown to be in 
$F\!P^{N\!P}$, but no lower bound on its data complexity was stated.
Thus, our result enriches the characterization in~\cite{Luk01}, as it implies that 
$F\!P^{N\!P[\log n]}$ is a lower bound for the entailment problem for probabilistic logic 
programs under data complexity even in the presence of rules much simpler than conditional 
rules. 
More details are given in Section~\ref{sec:relatedwork}, where we provide a more 
thorough comparison with~\cite{Luk01}.
However, finding the tightest characterization for \qa\ remains an open problem, as it might 
be the case that 
\qa\ is complete for either $F\!P^{N\!P[\log n]}$ or $F\!P^{N\!P}$.
We conjecture that none of these cases holds (thus a characterization of \qa\ tighter 
than ours can not be provided), thus \qa\ is likely to be
in the ``limbo'' containing the problems in  $F\!P^{N\!P}$ but not in $F\!P^{N\!P[\log n]}$, 
without being hard for the former (this limbo is non-empty if 
$P\!\neq\!N\!P$~\cite{Krentel88}).
\subsection{Tractability results}
In this section, we show some sufficient conditions for the tractability of the query evaluation problem, which hold for both its decision and search versions.
When stating our results, we refer to \qa\ only, as its tractability implies that of \mp\
(as \mp\ is straightforwardly reducible to \qa).

Again, we address the tractability from two standpoints: we will show sufficient 
conditions which regard either 
$a)$ the shape of the conflict hypergraph, or
$b)$ the syntactic form of the constraints.
Specifically, we focus on finding islands of tractability when queries are projection-free and either the conflict hypergraph collapses to a graph -- as for direction $a)$, or the 
constraints are binary -- as for direction $b)$.
These are interesting contexts,
since Theorem~\ref{theo:mphardness2} entails that \mp\ (and, thus, also \qa) is, 
in general, hard in these cases (indeed, Theorem~\ref{theo:mphardness2}
implicitly shows the hardness for the case of conflict hypergraphs collapsing to graphs, as,
in the presence of binary constraints, the conflict hypergraph is a graph).

The next result goes into direction $a)$, as it states that, for projection-free queries,
\qa\ is tractable if the conflict hypergraph is a graph satisfying some structural properties.

\begin{theorem}
\label{theo:qatractable}
For projection-free conjunctive queries, \qa\ is in \emph{PTIME} if $HG(D^p, \IC)$ is a 
graph where each maximal connected component is either a tree or a clique.
\end{theorem}

The polynomiality result stated above is rather straightforward in the case that each connected component is a clique, but is far from being straightforward in the presence of 
connected components which are trees.
Basically, when the conflict hypergraph is a tree, the tractability derives from the
fact that, for any conjunction of tuples, its minimum (or, equivalently, maximum) probability
can be evaluated as the solution of an instance of a linear programming problem.
In particular, differently from the ``general'' system of inequalities used in the proof of
Proposition~\ref{pro:intquery} (where the variables corresponds to the possible worlds, thus their number is exponential in the number of tuples),
here we can define a system of inequalities where both the number of inequalities and
variables depend only on the arity of the query (which is constant, as we address data complexity).
We do not provide an example of the form of this system of inequalities, as explaining the correctness of the approach on a specific case is not easier than proving its validity in the 
general case.
Thus, the interested reader is referred to the proof of Theorem~\ref{theo:qatractable} reported in~\ref{app:tractabilityqa} for more details.

The following result goes into direction of locating tractability scenarios arising from 
the syntactic form of the constraints, as it states that, if $\IC$ consists of 
one FD for each relation scheme, the evaluation of projection-free queries is tractable.

\begin{theorem}
\label{theo:qa-tractable-binary-EGDs}
For projection-free conjunctive queries, \qa\ is in \emph{PTIME} if 
$\IC$ consists of at most one FD per relation scheme.
\end{theorem}
\begin{proof}
We consider the case that $\IC$ contains one relation and one FD only, as the general case 
(more relations, and one FD per relation) follows straightforwardly.
Let the denial constraint $ic$ in $\IC$ be the following FD over relation scheme $R$:
$X\rightarrow Y$, where $X$, $Y$ are disjoint sets of attributes of $R$.
We denote as $r$ the instance of $R$ in the instance of \qa.
Constraint $ic$ implies a partition of $r$ into disjoint relations, each corresponding to a different combination of the values of the attributes in $X$ in the tuples of $r$.
Taken one of this combinations $\vec{x}$ (i.e., $\vec{x} \in \Pi_X(r)$), we denote the corresponding set of tuples in this partition as $r(\vec{x})$.
That is, $r(\vec{x})= \{t \in r | \Pi_X(t)=\vec{x}\}$.
In turn, for each $r(\vec{x})$, $ic$ partitions it into disjoint relations, each corresponding
to a different combinations of the values of the attributes in $Y$.
Taken one of this combinations $\vec{y}$ (i.e., $\vec{y} \in \Pi_Y(r(\vec{x}))$), we denote the corresponding set of tuples in this partition as $r(\vec{x}, \vec{y})$.

Given this, constraint $ic$ entails that the conflict hypergraph is a graph with the
following structure:
there is an edge $(t_1,t_2)$ iff $\exists \vec{x}, \vec{y_1},\vec{y_2}$, 
with $\vec{y_1}\neq\vec{y_2}$, such that $t_1\in r(\vec{x}, \vec{y_1})$ and 
$t_2\in r(\vec{x}, \vec{y_2})$.

Now, consider any conjunction of tuples $T=t_1, \dots, t_n$.
The probability of $T$ as an answer of the query $q$ specified in the instance of \qa\ can be computed as follows.
First, we partition $\{t_1, \dots, t_n\}$ according to the maximal connected components of the conflict hypergraph.
This way we obtain the disjoint subsets $T_1, \dots, T_k$ of  $\{t_1, \dots, t_n\}$, where 
each $T_i$ corresponds to a maximal connected component of the conflict hypergraph, and 
contains all the tuples of  $\{t_1, \dots, t_n\}$ which are in this component. 
The minimum and maximum probabilities of $T$ as answer of $q$ can be obtained by 
computing the minimum and maximum probability of each set $T_i$, and then combining them
using the well known Frechet-Hoeffding formulas (reported also in the appendix as Fact~\ref{fac:boole}), which give the minimum and maximum probabilities of a conjunction of events 
among which no correlation is known (in fact, since $T_1, \dots, T_k$ correspond to distinct 
connected components, they can be viewed as pairwise uncorrelated events).

Then, it remains to show how the minimum and maximum probabilities of a single $T_i$ can be computed.
We consider the case that $T_i$ contains at least two tuples (otherwise, the minimum and the maximum probabilities of $T_i$ coincide with the marginal probability of the unique tuple in $T_i$).
If $\exists t_\alpha,t_\beta\in T_i$ $\exists \vec{x}, \vec{y_1},\vec{y_2}$ such that $t_\alpha\neq t_\beta$ and
$\vec{y_1}\neq\vec{y_2}$ and $t_\alpha\in r(\vec{x}, \vec{y_1})$, while $t_\beta\in r(\vec{x}, \vec{y_2})$, then the minimum and maximum probabilities of $T_i$ are both $0$ (since 
$\{t_\alpha,t_\beta\}$ is a conflicting set).
Otherwise, it is the case that all the tuples in $T_i$ share all the values $\vec{x}$ 
for the attributes $X$, and the same values $\vec{y}$ for the attributes $Y$.
Due to the structure of the conflict hypergraph, it is easy to see that this implies that
the tuples in $T_i$ can be distributed in any way in the portion of the probability space 
which is not invested to represent the tuples having the same values  $\vec{x}$ for $X$, but combinations for $Y$ other than $\vec{y}$.
The size of this probability space is 
$S=1-\sum_{\vec{y}^*\neq \vec{y}}\max \{p(t)|t\in r(\vec{x},\vec{y}^*)\}$.
Hence, the minimum and maximum probabilities of $T_i$ are:
\begin{center}
$p^{\min}= \max\left\{0,  \sum_{t\in T_i} p(t)- |T_i| +S\right\}$; \hspace*{1cm}
$p^{\max}= \min\left\{p(t)\, |\, t\in T_i\right\}$.
\end{center}
The first formula is an easy generalization of the corresponding formula for the minimum probability given in Lemma~\ref{lem:pmin1/2} to the case of a probability space of a generic
size less than $1$.
The second formula derives from the above-recalled Frechet-Hoeffding formulas, and from the
fact that the database is consistent (we recall that we rely on this assumption when addressing the query evaluation problem).
\end{proof}

Again, observe that the last two results are somehow complementary: it is easy to see that
there are FDs yielding conflict hypergraphs not satisfying 
the sufficient condition of Theorem~\ref{theo:qatractable}, as well as conflict hypergraphs 
which are trees generated by some ``more general'' denial constraint, not 
expressible as a set of FDs over distinct relations.

\section{Extensions of our framework}
\label{app:extensions}

Some extensions of our framework are discussed in what follows.
In particular, for each extension, we show its impact on our characterization of the fundamental problems addressed in the paper.\\

\subsection{Tuples with uncertain probabilities}
All the results stated in this paper can be trivially extended to the case that 
tuples are associated with ranges of probabilities, rather than single probabilities
(as happens in several probabilistic data models, such as~\cite{LakshmananLRS97,Luk01}).

Obviously, all the hardness results for \cc, \mp, \qa\ hold also for this variant,
since considering tuples with single probabilities is a special case of 
allowing tuples associated with range of probabilities.

As regards \cc, both the membership in $NP$ and the extendability of the tractable cases
straightforwardly derive from the fact that, as only denial constraints are considered,
deciding on the consistency of an assignment of ranges of probabilities can be accomplished
by looking only at the minimum probabilities of each range.

As regards \mp\ and \qa, the fact that the complexity upper-bounds do not change follows from the results in~\cite{Luk01}.
Finally, it can be shown, with minor changes to the proof of 
Theorem~\ref{theo:qatractable}, that \mp\ and \qa\ are still tractable under the hypotheses on the shape of the conflict hypergraph stated in this theorem.
We refer the interested reader to~\ref{app:estensionequery}, where a hint is given on how the proof of Theorem~\ref{theo:qatractable} can be extended to deal with 
tuples with uncertain probabilities.
The extension of the tractability results for \mp\ and \qa\ regarding the syntactic forms  of 
the constraints is even simpler, and can be easily understood after reading the proofs of these results.

\subsection{Associating constraints with probabilities.}
Another interesting extension consists in allowing constraints to be assigned probabilities.
In our vision, constraints should encode some certain knowledge on the data domain, 
thus they should be interpreted as deterministic.
However, this extension can be interesting at least from a theoretical point of view, or
when constraints are derived from some elaboration on historical data~\cite{FassettiF07}.
Thus, the point becomes that of giving a semantics to the probability assigned to 
the constraints. 
The semantics which seems to be the most intuitive is as follows: 
``\emph{A constraint with probability $p$ forbidding the co-existence of some tuples is 
satisfied if there is an interpretation where the overall probability of the possible 
worlds satisfying the constraint is at least $p$}''. 
This means that the condition imposed by the constraint must hold in a portion of size $p$ of 
the probability space, while nothing is imposed on the remaining portion of the probability space.

Starting from this, we first discuss the impact of associating constraints with 
probabilities on our results about \cc.
First of all, it is easy to see that there is a reduction from any instance 
\emph{Prob-cc} of the variant of \cc\ with probabilistic constraints to an equivalent 
instance \emph{Std-cc} of the standard version of \cc.
Basically, this reduction constructs the conflict hypergraph $H($\emph{Std-cc}$)$ of
\emph{Std-cc} as follows:
denoting the conflict hypergraph of \emph{Prob-cc} as $H($\emph{Prob-cc}$)$,
each hyperedge $e \in H($\emph{Prob-cc}$)$ (with probability $p(e)$) is transformed 
into a hyperedge $e'$ of $H($\emph{Std-cc}$)$ which consists of the same nodes in $e$ 
plus a new node with probability $p(e)$.
On the one hand, the existence of this reduction suffices to state that also the probabilistic version of \cc\ is $NP$-complete. 
On the other hand, it is worth noting that applying this reduction yields a conflict hypergraph $H($\emph{Std-cc}$)$ with the same ``shape'' as $H($\emph{Prob-cc}$)$, except that
each hyperedge has one new node, belonging to no other hyperedge: 
hence, if $H($\emph{Prob-cc}$)$ is a hypertree (resp., a ring), 
then $H($\emph{Std-cc}$)$ is a hypertree (resp., a ring) too. 
This means that all the tractability results given for \cc\ concerning the shapes of the 
conflict hypergraph hold also when stated directly on its probabilistic version.
However, this does not suffice to extend the tractability results for \cc\ regarding  the syntactic forms of the constraints, as in the considered cases the conflict hypergraph may not be a hypertree or a ring.
Thus, the extension of the tractability results on the syntactic forms is deferred to future work.

As regards \mp\ and \qa, the  arguments used in the discussion of the previous
extension can be used to show that our lower and upper bounds still hold for the 
variants of these problems allowing probabilistic constraints.
As for the tractability results, in~\ref{app:estensionequery}, a more detailed discussion is provided explaining how
the proof of Theorem~\ref{theo:qatractable} (which deal with conflict hypergraphs 
where each maximal connected componenent is either a clique or a tree) can be extended to deal with probabilistic constraints.
The extension of the tractability result for FDs stated in Theorem~\ref{theo:qa-tractable-binary-EGDs} 
is deferred to future work.\\

\subsection{Assuming pairs of tuples as independent unless this contradicts the constraints}
As observed in the introduction, in some cases, rejecting the assumption of independence for some groups of tuples may be somehow ``overcautious''.
For instance, if we consider further tuples pertaining to a different  hotel in the introductory example (where constraints involve tuples over the same hotel), it may be reasonable to assume that these tuples encode events independent from those pertaining 
hotel $1$.

A naive way of extending our framework in this direction is that of assuming every pair 
of tuples which are not explicitly ``correlated'' by some constraint as independent from 
one another.
This means considering as independent any two tuples $t_1$, $t_2$ such that there is no 
hyperedge in the conflict hypergraph containing both of them.
However, this strategy can lead to wrong interpretations of the data.
For instance, consider the case of Example~\ref{ex:probabilisticIC}, where each of the 
three tuples $t_1$, $t_2$, $t_3$ has probability $1/2$, and two (ground) constraints are
defined over them:
one forbidding the co-existence of $t_1$ with $t_2$, 
and the other forbidding the co-existence of $t_2$ with $t_3$.
As observed in Example~\ref{ex:probabilisticIC}, the combination of these two constraints 
implicitly enforces the co-existence of $t_1$ with $t_3$.
Hence, the fact that $t_1$ and $t_3$ are not involved in the same (ground) constraint
does not imply that these two tuples can be considered as independent from one another.

However, it is easy to see that if two tuples are not connected through any path in the 
conflict hypergraph, assuming independence among them does not contradict the 
constraints in any way.
Hence, a cautious way of incorporating the independence assumption in our framework is 
the following:
any two tuples are independent from one another iff they belong to distinct maximal 
connected components of the conflict hypergraph.

If this model is adopted, nothing changes in our characterization of the consistency 
checking problem.
In fact, it is easy to see that an instance of \cc\ is equivalent to an instance of the variant of \cc\ where independence is assumed among maximal connected components 
of the conflict hypergraph.
This trivially follows from the fact that, if a PDB $D^p$ is consistent according to the original framework, all the possible interpretations combining the models of the maximal connected components are themselves models of $D^p$, and the set of these interpretations 
contains also the interpretation corresponding to assuming independence among the maximal 
connected components.

As regards the query evaluation problem, adopting this variant of the framework makes 
\qa\ \#P-hard (as \qa\ becomes more general than the problem of evaluating queries under the independence assumption~\cite{Suciu04}).
However, all our tractability results for projection-free queries still hold.
In fact, the probability of $t_1,\dots,t_n$ as an answer of a query can be obtained as
follows.
First, the set $T=\{t_1,\dots,t_n\}$ is partitioned into the (non-empty) sets 
$S_1, \dots, S_k$ which correspond to distinct maximal connected components of the conflict hypergraph, and
where each $S_i$ consists of all the tuples in $T$ belonging to the 
connected component corresponding to $S_i$.
Then, the minimum and maximum probabilities of each $S_i$ are computed (in \emph{PTIME}, when 
our sufficient conditions for tractability hold), by considering each $S_i$ separately.
Finally, the independence assumption among the tuples belonging to distinct maximal components is exploited, so that the minimum (resp., maximum) probability of $t_1,\dots,t_n$ is evaluated as the product of the so obtained minimum (resp., maximum) probabilities
of $S_1, \dots, S_k$.


\section{Related work}
\label{sec:relatedwork}
\noindent
We separately discuss the related work in the AI and DB literature.\\[1pt]
\emph{AI setting}.
The works in the AI literature related to ours are mainly those dealing with
probabilistic logic.
The problem of integrating probabilities into logic was first addressed 
(though pretty informally) in~\cite{Nilsson86}.
Then, in~\cite{Pap88} the PSAT problem was formalized as the satisfiability problem 
in a propositional fragment of the logic discussed in~\cite{Nilsson86}, and shown to be
$N\!P$-complete.
In~\cite{FaginHM90}, a more general probabilistic propositional logic than that 
in~\cite{Pap88} was defined, which enables algebraic relations to be specified
among the probabilities of propositional formulas 
(such as ``\emph{the probability of $\phi_1 \wedge \phi_2$ is twice that of 
$\phi_3 \vee \phi_4$}).
\cite{FaginHM90} mainly focuses on the satisfiability problem, showing that it is 
$N\!P$-complete (thus generalizing the result on PSAT of~\cite{Pap88}).
However, it provides no tractability result (whose investigation is our main contribution in 
the study of the corresponding consistency problem).
Up to our knowledge, most of the works devising techniques for efficiently solving the
satisfiability problem (such as~\cite{KavvadiasP90,Luk99}) rely on translating it into a
Linear Programming instance and using some heuristics, which do not guarantee 
polynomial-bounded complexity. 
Thus, the only works determining provable polynomial cases of probabilistic 
satisfiability are~\cite{Pre01,Pap88}.
As for \cite{Pap88}, we refer the reader to the discussions in Section~\ref{sec:consistencychecking} (right after Definition~\ref{def:cc}) and 
at the begininning of Section~\ref{sec:tractablecases}.
As regards~\cite{Pre01}, it is related to our work in that it showed 
that PSAT is tractable if the hypergraph of the formula (which corresponds to our conflict
hypergraph) is a hypertree.
However, the notion of hypertree in~\cite{Pre01} is very restrictive, as it relies on a 
notion of acyclicity much less general than the $\gamma$-acyclicity used here.
In fact, even the simple hypergraph consisting of 
$e_1=\{t_1, t_2, t_3\}$, $e_2=\{t_2, t_3, t_4\}$ is not viewed in~\cite{Pre01} as a 
hypertree, since it contains at least one cycle, such as $t_1,e_1,t_2,e_2,t_3,e_1,t_1$ 
(note that, in our framework, this would not be a cycle). 
Basically, hypertrees in~\cite{Pre01} are special cases of our hypertrees, 
as they require distinct hyperedges to have at most one node in common.
Hence, our result strongly generalizes the forms of conflict hypergraphs over which \cc\ 
turns out to be tractable according to the result of~\cite{Pre01} on PSAT.

The entailment problem (which corresponds to our query answering problem) was studied both
in the propositional~\cite{Luk99} and in the (probabilistic-)logic-programming 
setting~\cite{Luk01,NgS92,Ng97}.
The relationship between these works and ours is in the fact that they deal with knowledge
bases where rules and facts can be associated with probabilities.
Intuitively, imposing constraints over a PDB might be simulated by a 
probabilistic logic program, where tuples are encoded by (probabilistic) facts and constraints
by (probabilistic) rules with probability $1$.
However, not all the above-cited probabilistic-logic-programming frameworks can be used
to simulate our framework: for instance, \cite{NgS92,Ng97} use rules 
which can not express our constraints.
On the contrary, the framework in~\cite{Luk01} enables pretty general rules to be 
specified, that is \emph{conditional rules} of the form $(H|B)[p_1,p_2]$, where $H$ and $B$
are classical open formulas, stating that the probability of the formula $H\wedge B$ is 
between $p_1$ and $p_2$ times the probability of $B$.
Obviously, any denial constraint $ic$ can be written as a conditional rule of the form 
$(H|$\emph{true}$)[1,1]$, where $H$ is the open formula in $ic$.  
In the presence of conditional rules, \cite{Luk01} characterizes the complexity of the 
satisfiability and the entailment problems.
The novelty of our contribution w.r.t. that of~\cite{Luk01} derives from the specific 
database-oriented setting considered in our work.
In particular, as regards the consistency problem, our tractable 
cases are definitely a new contribution, as~\cite{Luk01} does not determine 
polynomially-solvable instances.
As regards the query answering problem, our contribution is relevant 
from several standpoints.
First, we provide a lower bound of the membership problem by assuming that the database
is consistent: this is a strong difference with~\cite{Luk01}, where 
the decisional version of the entailment problem has been addressed without assuming the 
satisfiability of the knowledge base, thus the satisfiability checking is used as a source of
complexity when deciding the entailment. 
Second, we have characterized the lower bound of the membership problem 
w.r.t. two specific aspects, which make sense in a database-perspective and were 
not considered in~\cite{Luk01}:
the presence of projection in the query (Theorem~\ref{theo:mphardness1}) and the type of 
denial constraints (Theorem~\ref{theo:mphardness2}).
Third,~\cite{Luk01} did not prove any lower bound for the data complexity of the search 
version of the entailment problem.
Indeed, it provided an $F\!P^{N\!P}$-hardness result only under combined complexity
(assuming all the knowledge base as part of the input, while we consider constraints
of fixed size) and exploiting the strong expressiveness of conditional rules, which enable 
also constraints not expressible by denial constraints to be specified.
Hence, in brief, our Theorem~\ref{theo:fp-at-np} shows that constraints simpler than
conditional constraints suffice to get an $F\!P^{N\!P[\log n]}$-hardness of the entailment for 
probabilistic logic programs, even under data complexity.
Finally, our  tractable cases of the query evaluation problem, up to our knowledge, are not subsumed by any result in the literature, and depict islands of 
tractability also for the more general entailment problem studied in~\cite{Luk01}. 
\\[2pt]
\emph{DB setting}.
The database research literature contains several works addressing various aspects related to
probabilistic data, and a number of models have been proposed for their representation and 
querying.
In this section, we first summarize the most important results on probabilistic databases
relying on the independence assumption (which, obviously, is somehow in contrast with 
allowing integrity constraints to be specified over the data, thus making these works 
marginally related to ours).
Then, we focus our attention on other works, which are more related to ours as they 
allow some forms of correlations among data to be taken into account when representing and 
querying data.

As regards the works relying on the independence assumption, the problem of efficiently evaluating (conjunctive) queries was first studied in~\cite{Suciu04}, where it was shown that
this problem is \#P-hard in the general case of queries without self-joins, but can be solved
in polynomial time for queries admitting a particular evaluation plan (namely, \emph{safe plan}).
Basically, a safe plan is obtained by suitably pushing the projection in the query 
expression, in order to extend the validity of the independence assumption also to the partial 
results of the query. 
The results of~\cite{Suciu04} were extended 
in~\cite{Dalvi10,Suciu07, Suciu07b,SuciuEDBT2010,KochICDE09}.
Specifically, in~\cite{Suciu07}, a technique was presented for computing safe plans on 
disjoint-independent databases (where only tuples belonging to different buckets are 
considered as independent). 
In~\cite{Suciu07b} and~\cite{Dalvi10}, the dichotomy theorem of~\cite{Suciu04} was extended 
to deal with conjunctive queries with self-joins and unions of conjunctive queries, respectively.
In~\cite{KochICDE09}, it was shown that a polynomial-time evaluation can be accomplished also
with query plans with any join ordering (not only those orderings required by safe plans).
Finally, in~\cite{SuciuEDBT2010}, a technique was presented enabling the determination of 
efficient query plans even for queries admitting no safe plan (this is 
allowed by looking at the database instance to decide the most suitable query 
plan, rather than looking only at the database schema).

The problem of dealing with probabilistic data when correlations are not known (and 
independence may not be assumed) was addressed in~\cite{LakshmananLRS97}.
Here, an algebra for querying probabilistic data was introduced, as well as a system called 
\emph{ProbView}, which supports the evaluation of algebraic expressions by returning answers 
associated with probability intervals.  
However, the query evaluation is based on an extensional semantics and no integrity constraints encoding domain knowledge were considered. 

One of the first works investigating a suitable model for representing correlations among
probabilistic data is~\cite{Green*2006}, where \emph{probabilistic c-tables} were introduced.
In this framework, whose rationale is also at the basis of the PDB 
\emph{MayBMS}~\cite{MayBMS}, correlations are expressed by associating tuples with boolean
formulas on random variables, whose probability functions are represented in a table.
However, in this approach, only one interpretation for the database is considered (the one
deriving from assuming the random variables independent from one another), and it is not 
suitable for simulating the presence of integrity constraints on the data when the marginal 
probabilities of the tuples are known.
Similar differences, such as that of assuming only one interpretation,
hold between our framework and that at the basis of 
\emph{Trio}~\cite{TRIO,Agrawal2010}, where incomplete and probabilistic data are modeled by 
combining the possibility of specifying buckets of tuples with the association of each tuple
with its lineage (expressed as the set of tuples from which each tuple derived).
In particular, in~\cite{Agrawal2010} an extension of Trio is proposed which aims at 
better managing the epistemic uncertainty (i.e., the information about uncertainty is 
itself incomplete).
Here, the semantics of \textit{generalized uncertain databases} is given in terms of a Dempster-Shafer mass distribution over the powerset of the possible worlds (this collapses 
to the case of a PDB with one probability distribution, if the mass distribution is defined 
over every single possible world).
Further approaches to representing rich correlations and querying the data are those 
in~\cite{Deshpande07,DeshpandePods09,DeshpandeSigmod10}, where 
correlations
among data are represented according to some graphical models (such as PGMs, junction trees, 
AND/XOR trees). 
In these approaches, correlations are detected while data are generated and, in some sense, 
they are data themselves: the database consists of a graph representing correlations among
events, so that the marginal distributions of tuples are not explicitly represented, but 
derive from the correlations encoded in the graph. 
This is a strong difference with our framework, where a PDB is 
a set of tuples associated with their marginal probabilities, and constraints can be imposed 
by domain experts with no need of taking part to the 
data-acquisition process.
Moreover, in~\cite{Deshpande07,DeshpandePods09,DeshpandeSigmod10}, independence is assumed between
tuples for which a correlation is not represented in the graph of correlations.
On the contrary, our query evaluation model relies on a ``cautious'' paradigm, where no 
assumption is made between tuples not explicitly correlated by the constraints.
In~\cite{DalviSuciu05}, the problem of evaluating queries over probabilistic views under integrity 
constraints (functional and inclusion dependencies) and in the presence of statistics on the
cardinality of the source relations was considered. 
In this setting, when evaluating query answers and their probabilities, all the possible values 
of the attribute values of the original relations must be taken into account, and this backs the
use of the Open World Assumption (as the original relations may contain attribute values which 
do not occur in the views). 
Under this assumption, queries are evaluated over the interpretation of the data having the 
maximum entropy among all the possible models.

All the above-cited works assume that the correlations represented among the data are
consistent.
In~\cite{KochPVLDB2O08}, the problem was addressed of querying a PDB 
when integrity constraints are considered a posteriori, thus some possible worlds 
having non-zero probability under the independence assumption may turn out to be inconsistent.
In this scenario, queries are still evaluated on the unique interpretation 
entailed by the independence assumption, but the possible worlds are assigned the probabilities 
\emph{conditioned} to the fact that what entailed by the constraint is true.
That is, in the presence of a constraint $\Gamma$, the probability $P(Q)$ of a query $Q$ is evaluated as $P(Q|\Gamma)$, which is the probability of $Q$ assuming that $\Gamma$ holds.
This corresponds to evaluating queries by augmenting them with the constraints, thus
it is a different way of interpreting the constraints and queries from the
semantics adopted in our paper, where constraints are applied on the database.
The same spirit as this approach is at the basis of~\cite{CohenKS09}, where specific forms 
of integrity constraints in the special case of probabilistic XML data are taken into account
by considering a single interpretation, conditioned on the constraints.
 
%

\section{Conclusions and Future work}
\label{sec:conclusions}
We have addressed two fundamental problems dealing with PDBs 
in the presence of denial constraints:
the \emph{consistency checking} and the \emph{query evaluation} problem.
We have thoroughly studied the complexity of these problems,
characterizing the general cases and pointing out several tractable cases.

There exist a number of interesting directions for future work.
First of all, the cautious querying paradigm will be extended to deal with further forms of 
constraints.
This will allow for enriching the types of correlations which can be expressed 
among the data, and this may narrow the probability ranges associated with the answers (in fact, 
for queries involving tuples which are not involved in any denial constraint, the obtained 
probability ranges may be pretty large, and of limited interest for data analysis).

Another interesting  direction for future work is the identification of other tractable 
cases of the consistency checking and the query evaluation problems.
As regards the consistency checking problem, 
we conjecture that polynomial-time strategies can be devised when the conflict 
hypergraph exhibits a limited degree of cyclicity (as a matter of fact, 
we have shown that this problem is feasible in linear time not only for hypertrees, but
also for rings, which have limited cyclicity as well).
A possible starting point is investigating the connection between the consistency
checking problem (viewed as evaluating the (dual) lineage of the constraint query - see Remark~$1$) 
and the model checking problem of Boolean formulas.
The connection between lineage evaluation and model checking has been well established mainly 
for the cases of tuple-independent PDBs~\cite{OlteanuH08,Suciu11}.
In fact, in this setting, it has been shown that, as it happens for checking 
Boolean formulas, the probability of a lineage can be evaluated by compiling it 
into a \emph{Binary Decision Diagram} - BDD~\cite{OBDDbook}, and then suitably processing 
the diagram.
Specifically, if the lineage (or, equivalently, the Boolean formula to be checked) $L$ 
can be compiled into a particular case of BDDs (such as \emph{Read-Once} or \emph{Ordered} 
BDD), the lineage evaluation (as well as the formula verification) can be accomplished  
as the result of a traversal of the BDD, in time linear w.r.t. the diagram size.
Hence, in all the cases where $L$ can be compiled into a \emph{Read-Once} or an \emph{Ordered}
BDD of polynomial size, $L$ can be evaluated in polynomial time.
One of the most general result about the compilability of Boolean formulas into  
\emph{Ordered} BDDs was stated in~\cite{Vardi05}, where it was shown that any CNF expression
over $n$ variables whose hypergraph of clauses has bounded treewidth ($<k$) admits
an equivalent ordered BDD of size $O(n^{k+1})$.
Then, the point becomes devising a mechanism for exploiting an Ordered BDD
equivalent to a Boolean formula $f$ to evaluate the probability of $f$, when neither 
independence nor precise correlations can be assumed among the terms of $f$.
Up to our knowledge, this topic has not been investigated yet, and we plan to address it in
future work.
If it turned out that, under no assumption on the way terms are correlated, the probability 
of formulas can be evaluated by traversing their equivalent Ordered BDDs,
then the above-cited result of~\cite{Vardi05} would imply other tractable cases of our
consistency checking problem.
However, our results on hypertrees and rings would be still of
definite interest, as we have found that in these cases the consistency checking problem 
can be solved in linear time, while the construction of the ordered BDD is $O(n^{k+1})$.
Moreover, our results show that the consistency checking problem over hypertrees and rings 
is still polynomially solvable (actually, in quadratic time) in the case that the 
cardinality of hyperedges is not known to be bounded by constants (see the discussion 
right after Theorem~\ref{theo:tractableCC}), which does not always correspond to
structures having bounded treewidth.

Finally, our framework can be exploited to
address the problem of repairing data and extracting reliable information from inconsistent
PDBs.
This research direction is somehow related to~\cite{And06}, where the evaluation of clean answers over \emph{deterministic} databases which are inconsistent due to the presence of duplicates is accomplished by encoding the inconsistent database into a PDB adopting the
bucket independent model.
Basically, in this PDB, probabilities are assigned to tuples representing variants of the same tuple, and these variants are grouped in buckets.
However, the so obtained PDB is consistent, thus this approach is not a repairing framework for inconsistent PDBs, but is a technique for getting clean answers over inconsistent
deterministic databases after rewriting queries into ``equivalent'' queries over the corresponding consistent PDBs.
A more general repairing problem in the probabilistic setting has been recently addressed in~\cite{Lian10},
where a strategy based on deleting tuples has been proposed, ``inspired'' by the 
common approaches for inconsistent deterministic databases~\cite{Ber11}.
We envision a different repairing paradigm, which addresses a source of inconsistency which 
is typical of the probabilistic setting: inconsistencies may arise from wrong assignments 
to the marginal probabilities of tuples, due to limitations of the model adopted for 
encoding uncertain data into probabilistic tuples.
In this perspective, a repairing strategy based on properly updating the probabilities
of the tuples (possibly by adapting frameworks for data repairing in the deterministic setting based on attribute updates~\cite{FlescaFP07,FlescaFP10,Wijsen05}) seems to be the most suitable choice.\\

%
%
%
 %
%
%
\noindent
{\em\textbf{\emph{Acknowledgements.}}
We are grateful to the anonymous reviewers of an earlier conference submission 
of a previous version of this 
paper for their fruitful suggestions (one especially for 
pointing out the reduction of \cc\ to PSAT), as 
well as Thomas Lukasiewicz,
for insightful discussions about his work \cite{Luk01}, and Francesco Scarcello, 
for valuable comments~about~our~work.}


\bibliographystyle{model1b-num-names}
\bibliography{bib}

\newpage
\appendix

\section{Proofs}
\label{app:proofs}
In this appendix we report the proofs of the theorems whose statement have been provided and
commented in the main body of the paper.
Furthermore, the appendix contains some new lemmas which are exploited in these proofs.

%

\subsection{Proofs of Theorem~\ref{theo:CCisNPcomplete}, Proposition~\ref{prop:esistenzamatrioska}, and Lemma~\ref{lem:pmin1/2}}
\label{app:Pro1Lemma1}
\noindent{\textbf{Theorem~\ref{theo:CCisNPcomplete}.}}
(Complexity of \cc)
\cc\  \textit{is} $NP$-\textit{complete}.
\begin{proof}
The membership of \cc\ in $NP$ has been already proved in the core of the 
paper, where a reduction from \cc\ to \emph{PSAT} has been described.
As regards the hardness, it follows from Theorem~\ref{theo:2FDs-cc} (or, equivalently, from 
Theorem~\ref{theo:ic-arity-3}), whose proof is given in Section~\ref{app:567}.
\end{proof}

We now report a property of $\gamma$-acyclic hypergraphs from~\cite{DAtri84},
which will be used in the proof
of Proposition~\ref{prop:esistenzamatrioska}.

\begin{fact}\label{fact:dm}
\emph{\cite{DAtri84}}~Let $H=\langle N,E\rangle$ be a hypertree.
There exists at least one hyperedge $e\in E$ such that
at least one of the following conditions hold:
\begin{enumerate}
\item
$e\cap \mbox{\emph{N}}(H^{-\{e\}})$ is a set of edge equivalent nodes;
\item
there exists $e'\in E$ such that $e'\neq e$ and
$e\cap \mbox{\emph{N}}(H^{-\{e, e'\}})=
e'\cap \mbox{\emph{N}}(H^{-\{e, e'\}})$.
\end{enumerate}
Moreover, $H^{-\{e\}}$ is still a hypertree. 
\end{fact}

\noindent{\textbf{Proposition~\ref{prop:esistenzamatrioska}.}}
\textit{
Let $H=\langle N,E\rangle$ be a hypertree.
Then, there is at least one hyperedge $e\in E$ such that 
$Int(e, H)$ is a matryoshka.
Moreover, $H^{-\{e\}}$ is still a hypertree.
}
\begin{proof} 
Reasoning by induction on the number of hyperedges in $E$, we prove that
there is a total ordering $e_1, \cdots, e_n$ of the edges in $E$
such that all the following conditions hold for each $i\in [1..n]$:
\begin{enumerate}
\item
either $e_i\cap \mbox{\emph{N}}(H^{-\{e_1, \cdots, e_{i-1}\}})$ is a set 
of edge equivalent nodes, or there
exists $e'\in \mbox{\emph{E}}(H^{-\{e_1, \cdots, e_{i-1}\}})$ such that $e'\neq e$ and 
$e\cap \mbox{\emph{N}}(H^{-\{e, e'\}})=
e'\cap \mbox{\emph{N}}(H^{-\{e, e'\}})$;
\item
$H^{-\{e_1, \cdots, e_{i}\}}$ is a hypertree;
\item
$Int(e_i, H^{-\{e1, \cdots, e_{i-1}\}})$ is a matryoshka.
\end{enumerate}

The base case ($|E|=1$) is straightforward.
In order to prove the induction step, we reason as follows.
Since $H$ is a hypertree, Fact~\ref{fact:dm} implies that there 
is a node $e$ such that 
$1)$ either $e\cap \mbox{\emph{N}}(H^{-\{e\}})$ is a set of edge equivalent 
nodes, or there exists $e'\in E$ such that $e'\neq e$ and
$e\cap \mbox{\emph{N}}(H^{-\{e, e'\}})=
e'\cap \mbox{\emph{N}}(H^{-\{e, e'\}})$, and
$2)$ $H^{-\{e\}}$ is a hypertree. 

From the inductive hypothesis, since $H^{-\{e\}}$ is a hypertree,
there exists a total ordering $e_1, \cdots, e_{n-1}$ of the nodes in
$E-\{e\}$ such that for each $i\in [1..n-1]$ conditions $1,2$ and $3$ 
are satisfied w.r.t. $H^{-\{e\}}$.

If $Int(e, H)$ is a matryoshka, then the total ordering
$e, e_1, \cdots,$ $e_{n-1}$ of the nodes in $E$ satisfies
conditions $1,2$ and $3$ for every edge in the sequence thus the
statement is proved in this case.

Otherwise, since $Int(e, H)$ is not a matryoshka then
$e\cap \mbox{\emph{N}}(H^{-\{e\}})$ is not a set of 
edge equivalent nodes. Hence, since $e$ satisfies the conditions
of Fact~\ref{fact:dm} then
there exists $e_j\in \{e_1, \cdots, e_{n-1}\}$ such that
$e\cap \mbox{\emph{N}}(H^{-\{e, e_j\}})=
e_j\cap \mbox{\emph{N}}(H^{-\{e, e_j\}})$.

We now consider separately the following two cases:
\begin{list}{--}
     {\setlength{\rightmargin}{0mm}
      \setlength{\leftmargin}{5mm}
      \setlength{\itemindent}{5mm} 
			\setlength{\itemsep}{0mm}}
\item[Case $1)$:] 
there is $k\in [1..j-1]$ such that 
$e_k\cap \mbox{\emph{N}}(H^{-\{e,e_1,\cdots, e_k, e_j\}})=
e_j\cap \mbox{\emph{N}}(H^{-\{e,e_1,\cdots, e_k, e_j\}})$.
\item[Case $2)$:] 
there is no $k\in [1..j-1]$ such that 
$e_k\cap \mbox{\emph{N}}(H^{-\{e,e_1,\cdots, e_k, e_j\}})=
e_j\cap \mbox{\emph{N}}(H^{-\{e,e_1,\cdots, e_k, e_j\}})$.
\end{list}

We first prove Case $1)$. Let $k\in [1..j-1]$ be the smallest index
such that $e_k\cap \mbox{\emph{N}}(H^{-\{e,e_1,\cdots, e_k, e_j\}})=$
$e_j\cap$ $ \mbox{\emph{N}}(H^{-\{e,e_1,\cdots, e_k, e_j\}})$.
We consider the total ordering 
of the edges of $E$ obtained by inserting 
$e$ immediately before $e_k$ in $e_1,\cdots, e_{n-1}$, i.e., 
$e_1,\cdots, e_{k-1}, e, $ $ e_k,\cdots, e_j,\cdots, e_{n-1}$.

We first prove that for each $i\in[1..k-1]$ conditions
$1,2$ and $3$ still hold.
For each $i\in[1..k-1]$ one of the following cases occur:
\begin{itemize}
\item
$e_i\cap e_j = \emptyset$. In this case since 
$e\cap \mbox{\emph{N}}(H^{-\{e, e_j\}})=
e_j\cap \mbox{\emph{N}}(H^{-\{e, e_j\}})$, it is straightforward to see that conditions $1,2$ and $3$ hold.
\item
$e_i\cap e_j \neq \emptyset$ and 
$e_i\cap \mbox{\emph{N}}(H^{-\{e,e_1,\cdots, e_{i-1}\}})$ is a set of edge equivalent nodes.
Since $e\cap \mbox{\emph{N}}(H^{-\{e, e_j\}})=
e_j\cap \mbox{\emph{N}}(H^{-\{e, e_j\}})$, 
$e_i\cap e_j \neq \emptyset$
and $e_j$ is an edge of $H^{-\{e,e_1,\cdots, e_{i-1}\}}$
then
$e_i\cap \mbox{\emph{N}}(H^{-\{e,e_1,\cdots, e_{i-1}\}})=
e_i\cap \mbox{\emph{N}}(H^{-\{e_1,\cdots, e_{i-1}\}}).$
Therefore, the nodes in $e_i\cap \mbox{\emph{N}}(H^{-\{e_1,\cdots, e_{i-1}\}})$
are edge equivalent w.r.t $H^{-\{e_1,\cdots, e_{i-1}\}}$ too.
Hence, conditions $1,2$ and $3$ hold.
\item
$e_i\cap e_j \neq \emptyset$ and there is an $h\in [i+1..n-1]$, 
with $h\neq j$, such that
$e_i\cap \mbox{\emph{N}}(H^{-\{e,e_1,\cdots, e_{i}, e_h\}})=
e_h\cap \mbox{\emph{N}}(H^{-\{e,e_1,\cdots, e_{i}, e_h\}}).$
Since, $e_j$ is and edge of $H^{-\{e,e_1,\cdots, e_{i}, e_h\}}$
and
$e\cap \mbox{\emph{N}}(H^{-\{e, e_j\}})=
e_j\cap \mbox{\emph{N}}(H^{-\{e, e_j\}})$
it holds that
$e_i\cap \mbox{\emph{N}}(H^{-\{e,e_1,\cdots, e_{i}, e_h\}})=e_i\cap \mbox{\emph{N}}(H^{-\{e_1,\cdots, e_{i}, e_h\}})=
e_h\cap \mbox{\emph{N}}(H^{-\{e_1,\cdots, e_{i}, e_h\}}).$
Hence conditions $1,2$ and $3$ hold in this case too.
\end{itemize}
Observe that, in the last two cases mentioned above
the fact that $Int(e_i, H^{-\{e1, \cdots, e_{i-1}\}})$ is a matryoshka
follows from the fact that 
$e_i\cap N(H^{-\{e1, \cdots, e_{i-1}\}}) = 
e_i\cap N(H^{-\{e,e1, \cdots, e_{i-1}\}})$ and
$e_i\cap e = e_i\cap e_j$.
Moreover, 
conditions $1,2$ and $3$ still hold for each $i\in[k..n-1]$
since they are not changed w.r.t. the inductive hypothesis.

As regards the edge $e$, it is easy to see that
conditions $1$ and $2$ are
satisfied since $e_j$ appears after $e$ in the total ordering
$e_1,\cdots, e_{k-1}, e, $ $ e_k,\cdots, e_j,\cdots, e_{n-1}$.

We now prove that condition $3$ holds for $e$.
We know from the induction hypothesis that 
$Int(e_k, H^{\{e,e_1,\cdots, e_{k-1}\}})$ is a matryoshka.
However, since $e\cap \mbox{\emph{N}}(H^{-\{e, e_j\}})=
e_j\cap \mbox{\emph{N}}(H^{-\{e, e_j\}})$ and $j>k$ then
$Int(e_k, H^{\{e,e_1,\cdots, e_{k-1}\}}) =$
$Int(e_k, H^{\{e_1,\cdots, e_{k-1}\}})$.
Since, $e_k\cap \mbox{\emph{N}}(H^{-\{e,e_1,\cdots, e_k, e_j\}})=$
$e_j\cap \mbox{\emph{N}}(H^{-\{e,e_1,\cdots, e_k, e_j\}})$
and
$e\cap \mbox{\emph{N}}(H^{-\{e, e_j\}})=
e_j\cap \mbox{\emph{N}}(H^{-\{e, e_j\}})$
it holds that 
$$e_k\cap \mbox{\emph{N}}(H^{-\{e,e_1,\cdots, e_k, e_j\}})=e_j\cap \mbox{\emph{N}}(H^{-\{e,e_1,\cdots, e_k, e_j\}})=
e\cap \mbox{\emph{N}}(H^{-\{e,e_1,\cdots, e_k, e_j\}}).$$
Therefore the set of nodes in
$e\cap \mbox{\emph{N}}(H^{-\{e_1,\cdots, e_{k-1}\}})$
can be partitioned in three sets $N, N', N''$ such that:
\begin{list}{--}
     {\setlength{\rightmargin}{0mm}
      \setlength{\leftmargin}{1mm}
      \setlength{\itemindent}{2mm} 
			\setlength{\itemsep}{0mm}}
\item
$N= e_k\cap \mbox{\emph{N}}(H^{-\{e,e_1,\cdots, e_k, e_j\}}) 
=\bigcup_{S\in Int(e_k, H^{\{e,e_1,\cdots, e_{k-1}\}})} S$,
\item
$N'= e_k\cap e_j - N$, and
\item
$N'' =  e\cap e_j -N'-N$.
\end{list}
Hence, it is easy to see that
$Int(e, H^{\{e_1,\cdots, e_{k-1}\}}) =
Int(e_k, H^{\{e,e_1,\cdots, e_{k-1}\}}) \cup \{N \cup N'\} \cup \{N \cup N'\cup N''\}.$
Therefore, $Int(e, H^{\{e_1,\cdots, e_{k-1}\}})$ is a matryoshka.
Hence, the proof for Case $1)$ is completed.

We now prove Case $2)$.
We consider the total ordering 
of the edges of $E$ obtained by inserting 
$e$ immediately before $e_j$ in $e_1,\cdots,$ $e_{n-1}$, i.e., 
$e_1,\cdots, e_{j-1}, $ $ e, e_j,\cdots, e_{n-1}$.
It is easy to see that we can prove that
for each $i\in[1..j-1]$ conditions
$1,2$ and $3$ are satisfied applying the same reasoning applied 
in Case $1)$ in order to prove that
for each $i\in[1..k-1]$ conditions
$1,2$ and $3$ hold.
Analogously to the proof of Case $1)$ it is straightforward to see that
conditions $1,2$ and $3$ still hold for each $i\in[j..n-1]$
since they are not changed w.r.t. the inductive hypothesis.

As regards the edge $e$, it is easy to see that
conditions $1$ and $2$ are
satisfied since $e_j$ appears after $e$ in the total ordering
$e_1,\cdots, e_{j-1}, $ $ e, e_j,\cdots, e_{n-1}$.

To complete the proof we show that condition $3$ holds for $e$ in this case.
From the induction hypothesis, we know that it is the case that 
$Int(e_j, H^{\{e,e_1,\cdots, e_{j-1}\}})$ is a matryoshka.
However, since $e\cap \mbox{\emph{N}}(H^{-\{e, e_j\}})=
e_j\cap \mbox{\emph{N}}(H^{-\{e, e_j\}})$ 
then $e\cap \mbox{\emph{N}}(H^{-\{e,e_1,\cdots, e_j\}})=$
$e_j\cap \mbox{\emph{N}}(H^{-\{e,e_1,\cdots, e_j\}})$, and 
it holds that the set of nodes in 
$e\cap \mbox{\emph{N}}(H^{-\{e_1,\cdots, e_{k-1}\}})$
can be partitioned in the sets $N$ and $N'$ such that:
\begin{list}{--}
     {\setlength{\rightmargin}{0mm}
      \setlength{\leftmargin}{1mm}
      \setlength{\itemindent}{2mm} 
			\setlength{\itemsep}{0mm}}
\item
$N= e_j\cap \mbox{\emph{N}}(H^{-\{e,e_1,\cdots, e_j\}}) = \bigcup_{S\in Int(e_j, H^{\{e,e_1,\cdots, e_{k-1}\}})} S$,
\item
$N'= e_k\cap e_j - N$.
\end{list}
It is easy to see that the following holds
$Int(e, H^{\{e_1,\cdots, e_{j-1}\}}) =Int(e_j, H^{\{e,e_1,\cdots, e_{k-1}\}}) \cup \{N \cup N'\}.$
Therefore, $Int(e, H^{\{e_1,\cdots, e_{j-1}\}})$ is a matryoshka,
which completes the proof for Case $2)$ and the proof of the proposition.
\end{proof}

%
%

Before providing the proof of Lemma~\ref{lem:pmin1/2},
we report a well-known result 
on the minimum and maximum probability of the conjunction 
of events among which no correlation is known, taken from~\cite{Boole1854}. 
\begin{fact}
\label{fac:boole}
Let $E_1, E_2$ be a pair of events such that their marginal probabilities  
$p(E_1)$, $p(E_2)$ are known, while no correlation among them is known.
Then, the minimum and maximum probabilities of the event $E_1\wedge E_2$ are as follows:\\
$p^{min}(E_1\wedge E_2)= \max\left\{0,  p(E_1)+p(E_2)-1\right\}$; and
$p^{max}(E_1\wedge E_2)= \min\left\{p(E_1),p(E_2)\right\}$.
\end{fact}

The formulas reported above are also known as Frechet-Hoeffding formulas.
In Lemma~\ref{lem:pmin1/2}, we generalize the formula for the minimum probability, and adapt it to our database setting.\\


\noindent{\textbf{Lemma~\ref{lem:pmin1/2}.}}
\textit{Let $D^p$ be an instance of $\D^p$ consistent w.r.t. $\IC$, 
$T$ a set of tuples of $D^p$, and
$H=HG(D^p, \IC)$.
If either 
$i)$ the tuples in $T$ are pairwise disconnected in $H$, 
or 
$ii)$ $Int(T,H)$ is a matryoshka,
then
$p^{min}(T)= \max\left\{0, \sum_{t\in T} p(t)\!-|T|\!+\! 1\right\}$.
Otherwise, this formula provides a lower bound for $p^{min}(T)$.
}
\begin{proof}
Case $i)$: 
In the case that $t_1,\dots, t_n$ are pairwise disconnected in the conflict hypergraph, 
the formula for $p^{min}(t_1,\dots,t_n)$ 
can be proved by induction on $n$, considering as base case
the formula for the minimum probability of a pair of events reported in Fact~\ref{fac:boole}.\\[4pt]
\noindent
Case $ii)$: 
We prove an equivalent formulation of the statement over the same instance of $D^p$:
{\em``Let $T$ be a set of nodes of $H=HG(D^p, \IC)$ such that $Int(T,H)$ is a 
matryoshka.
Let $T^n= t_1,\dots,t_n$ be a sequence consisting of the nodes of $T$ ordered as follows:
$i>j \implies s(t_i)\supseteq s(t_j)$, where $s(t_i)$ is the maximal set in $Int(T,H)$
containing $t_i$.
Then, 
$p^{min}(t_1,\dots,t_n)= \max\left\{0, \sum_{i=1}^n p(t_i)-n+1\right\}$''.}
That is, we consider the nodes in $T$ suitably ordered, as this will help us to reason 
inductively.

We reason by induction on the length of the sequence $T^n$.
The base case ($n=1$)  trivially holds, as, for any tuple $t$, $p^{min}(t)=p(t)$.
We now prove the induction step: we assume that the property holds for any sequence 
of the considered form of length $n-1$, and prove that this implies that the property holds for 
sequences of $n$ nodes.

From induction hypothesis, we have that the property holds for the subsequence 
$T^{n-1}= t_1, \dots, t_{n-1}$ of $T^n$.
That is, there is a model $M$ for $D^p$ w.r.t. $\IC$ such that
$\sum_{w\supseteq \{t_1,\dots, t_{n-1}\}} M(w)=$ 
$\max \{0,$ $\sum_{i=1}^{n-1}p(t_i)-(n-1)+1\}$.
We show how, starting from $M$, a model $M'$ can be constructed such that
$\sum_{w\supseteq \{t_1,\dots, t_{n-1}, t_n\}}=$ 
$\max \{0, \sum_{i=1}^{n}p(t_i)-n+1\}$, which is the formula reported in the statement for
$p^{min}(t_1,\dots,t_{n})$.
According to $M$, the set of possible worlds of $D^p$ can be partitioned into:
\begin{itemize}
\item
$W\left(t_1\wedge \dots\wedge t_{n-1}\wedge t_n\right)$: the set of possible worlds containing all the tuples $t_1,\dots, t_{n-1}, t_n$;
\item
$W\left(\neg(t_1\wedge\dots\wedge t_{n-1}) \wedge t_n\right)$: the set of possible worlds containing $t_n$, but not containing at least one among $t_1, \dots,t_{n-1}$;
\item 
$W\left(t_1\wedge \dots\wedge t_{n-1},\neg t_n\right)$: the set of possible worlds containing 
all the tuples $t_1,\dots, t_{n-1}$, but not containing $t_n$;
\item
$W\left(\neg(t_1\wedge\dots\wedge t_{n-1})\wedge \neg t_n \right)$: the set of possible worlds 
not containing $t_n$ and not containing at least one tuple among $t_1, \dots,t_{n-1}$.
\end{itemize}

For the sake of brevity, the set of worlds defined above will be denoted as 
$W$, $W'$, $W''$, $W'''$, respectively.
In the following, given a set of possible worlds $\mathcal{W}$, we denote as $M(\mathcal{W})$ the overall probability
assigned by $M$ to the worlds in $\mathcal{W}$, i.e., 
$M(\mathcal{W})=\sum_{w\in \mathcal{W}} M(w)$.
Thus, if $M(W)= \max \{0,$ $\sum_{i=1}^{n}p(t_i)-n+1\}$, then we are done, since the right-hand side of this formula is the expression for 
$p^{min}(t_1,\dots,t_n)$ given in the statement, and it is in every case a lower bound for 
$p^{min}(t_1,\dots,t_n)$ (in fact, $p^{min}(t_1,\dots,t_n)$ can not be less than the case that the tuples are pairwise disconnected in $H$).
Otherwise, it must be the case that 
$M(W)>\max \{0, \sum_{i=1}^{n}p(t_i)-n+1\}$.
Assume that $\sum_{i=1}^{n}p(t_i)-n+1>0$ (the case that 
$\max \{0, \sum_{i=1}^{n}p(t_i)-n+1\}=0$ can be proved similarly).
Hence, we are in the case that $M(W)=\sum_{i=1}^{n}p(t_i)-n+1+\epsilon>0$, with $\epsilon>0$.
Since $M(W')=p(t_n)-M(W)$, this means that
$M(W')= p(t_n)-\left( \sum_{i=1}^{n}p(t_i)-n+1 +\epsilon \right)=$
$-\sum_{i=1}^{n-1}p(t_i)+ (n-1)-\epsilon$.
From the induction hypothesis, the term $-\sum_{i=1}^{n-1}p(t_i)+ (n-1)$ is equal to 
$1- p^{min}(t_1,\dots, t_{n-1})$, thus 
we have:
$M(W')=1-p^{min}(t_1,\dots, t_{n-1})-\epsilon$.
Since $p^{min}(t_1,\dots, t_{n-1})$ is exactly the overall probability, according to $M$, of the
possible worlds containing all the tuples $t_1,\dots, t_{n-1}$, we have that
$1-p^{min}(t_1,\dots, t_{n-1})=M(W')+M(W''')$, thus we obtain:
$M(W')=M(W')+M(W''')-\epsilon$.
This means that $M(W''')=\epsilon$, where 
$\epsilon>0$.
That is, the overall probability of the possible worlds in $W'''$ is equal to the difference 
$\epsilon$ between $M(W)$ and the value $\sum_{i=1}^{n}p(t_i)-n+1$ that we want to obtain for the cumulative probability of the worlds in $W$.
We now show how $M$ can be modified in order to obtain a model $M'$ such that $M'(W)$ is exactly this value.
We construct $M'$ as follows.
Let $w'''_1, \dots,w'''_k$ be the possible worlds in $W'''$ such that $M(w'''_i)>0$, for each 
$i\in[1..k]$.
Take $k$ values $\epsilon_1,\dots,  \epsilon_k$, where each $\epsilon_i$ is equal to 
$M(w'''_i)$.  
Hence $\sum_{i=1}^k \epsilon_i=\epsilon$.
Then, for each $i\in[1..k]$, let $M'(w'''_i)=M(w'''_i)-\epsilon_i=0$, and, for each $w'''\in W'''\setminus\{w'''_1,\dots,w'''_k\}$, $M'(w''')=0$.
This way, $M'(W''')=\sum_{w'''\in W'''} M'(w''')= M(W''')-\epsilon=0$.
For each $w'''_i$ (with $i\in[1..k])$, let $w'_i$ be the possible world in $W'$ 
``corresponding'' to $w'''_i$: that is, $w'_i$ is the possible world $w'''_i\cup\{t_n\}$.
The, for each $i\in[1..k]$, let $M'(w'_i)=M(w'_i)+\epsilon_i$, and, 
for each $w'\in W'\setminus\{w'_1,\dots,w'_k\}$, $M'(w')=M(w')$.
This way, $M'(W')=\sum_{w'\in W'} M'(w')= M(W')+\epsilon$.
Basically, we are constructing the model $M'$ by ``moving'' $\epsilon$ from the overall 
probability assigned by $M$ to the worlds of $W'''$ towards the worlds of $W'$.
Observe that every world $w'_i\in W'$ such that $M'(w'_i)>0$ is consistent w.r.t. $\IC$, for the following reason.
If $M'(w'_i)=M(w'_i)$ the property derives from the fact that $M$ is a model.
Otherwise, we are in the case that $w'_i=w'''_i\cup\{t_n\}$, where $M(w'''_i)>0$. 
Since $M$ is a model,  $M(w'''_i)>0$ implies that $w'''_i$ is consistent w.r.t. $\IC$. 
Then, adding $t_n$ to $w'''_i$ to obtain $w'_i$ has no impact on the consistency:
$w'_i$ does not contain at least one tuple among $t_1, \dots, t_{n-1}$, and from the fact that 
any hyperedge of $HG(D^p, \IC)$ containing $t_n$ contains all the tuples $t_1, \dots, t_{n-1}$
no constraint encoded by the hyperedges containing $t_n$ is fired in $w'_i$.

It is easy to see that the strategy that we used to move $\epsilon$ from the overall probability
of $W'''$ to $W'$ does not change the overall probabilities assigned to the tuples 
different from $t_n$ in the worlds in $W'\cup W''$, but it changes the overall probability
assigned to tuple $t_n$ in the same worlds, as it is increased by $\epsilon$.
Hence, to adjust this, we perform an analogous reasoning to ``move'' $\epsilon$ from the 
overall probability $M(W)$ (which is at least $\epsilon$ and whose worlds contain $t_n$) to the overall probability assigned to $W''$ (which contains the same worlds of $W$ deprived of 
$t_n$).
Thus, we define $M'$ by ``moving" portions of $\epsilon$ from the worlds of $W$ to the 
corresponding worlds of $W''$ (where the corresponding worlds are those having the same tuples
except from $t_n$), analogously to what done before from the worlds of $W'''$ to those of $W'$.
This way, we obtain that $M'(W)=M(W)-\epsilon$ and $M'(W'')=M(W)+\epsilon$. 
Also in this case, $M'$ does not assign a non-zero probability to inconsistent worlds of $W''$:
for any $w''_i$ such that $M'(w''_i)>M(w''_i)$, it is the case that $M(w_i)>0$ 
(where $w_i=w''_i
\cup \{t_n\}$, which means that $w_i$ is consistent, and thus $w''_i$ (which results from 
removing a tuple from $w_i$) must be consistent as well (removing a tuple cannot fire any denial constraint).
Finally, observe that this strategy for moving $\epsilon$ from the cumulative probability of
$W$ to $W''$ does not alter the marginal probabilities of the tuples different from $t_n$ in these worlds.

Therefore, $M'$ is a model for $D^p$ w.r.t. $\IC$ which assigns to $W$ a cumulative probability
equal to 
$M'(W)=M(W)-\epsilon=\sum_{i=1}^{n}p(t_i)-n+1$, which ends the proof.
\end{proof}

\subsection{Proof of Theorem~\ref{theo:ring}}\label{sec:AppendixRing}
In order to prove Theorem~\ref{theo:ring}, we exploit a property that holds for particular 
conflict hypergraphs, called \textit{chains}.
Basically, a chain is the hypergraph resulting from removing a hyperedge from a ring.
Thus, a chain consists of a sequence of hyperedges $e_1, \dots, e_n$
where all and only the pairs of
consecutive hyperedges have non-empty intersection (differently from the ring, 
$e_1\cap e_n=\emptyset$).

Given a chain $\mathcal{C}=e_1, \dots, e_n$, 
we say that $n$ is its length, and denote it with \emph{length}$(\mathcal{C})$.
Moreover, for each $i\in[1..n-1]$, we will use the symbol $\alpha_i$ to denote 
the intersection $e_i\cap e_{i+1}$ of consecutive hyperedges,
and, for each $i\in[1..n]$, we will use the symbol $\beta_i$ to denote 
\emph{ears}$(e_i)$, and  $\tilde{\beta}_i$ to denote a subset of \emph{ears}$(e_i)$.
Finally, \emph{sub}$(\mathcal{C})$ will denote the subsequence $e_2,\dots, e_{n-1}$ of the 
hyperedges in $\mathcal{C}$.

In the following, given a set of tuples $X$, we will use the term ``\emph{event $X$}'' 
to denote the event that all the tuples in the set $X$ co-exist.
Furthermore, $p^{min}_{H}(E)$ will denote the minimum probability of the event 
$E$ involving the tuples of the database $D^p$ when the conflict hypergraph contains only the
hyperedges in $H$.

\begin{lemma}
\label{lem:pmin-chain-hyper}
Let $D^p$ be a PDB instance of $\D^p$ such that $D^p\models \IC$.
Assume that $HG(D^p, \IC)$ is the chain $\mathcal{C}=e_1,\dots,e_n$ (with $n>1$). 
Moreover, let $\tilde{\beta}_1$, $\tilde{\beta}_n$ be subsets of the ears $\beta_1$, 
$\beta_n$ of $e_1$ and $e_n$, respectively.
Then:

\noindent

$$
p^{min}_{\mathcal{C}}(\tilde{\beta}_1\cup \tilde{\beta}_n)=
\max
\left\{ 
0, \ 
p^{min}_{\emptyset}(\tilde{\beta}_1)+ 
p^{min}_{\emptyset}(\tilde{\beta}_n)
-
\left[
1-
p^{min}_{\mbox{\scriptsize{sub}}(\mathcal{C})}
(\alpha_1\cup(\beta_1\!\setminus\! \tilde{\beta}_1)\cup \alpha_{n-1}\cup (\beta_n\!\setminus\!\tilde{\beta}_n))
\right]
\right\}
$$

\noindent
where:
\noindent
$p^{min}_{\mbox{\scriptsize{sub}}(\mathcal{C})}
(\alpha_1\cup(\beta_1\!\setminus\! \tilde{\beta}_1)\cup \alpha_{n-1}\cup (\beta_n\!\setminus\!\tilde{\beta}_n))=
\max
\left\{ 
0, \ 
p^{min}_{\mbox{\scriptsize{sub}}(\mathcal{C})}(\alpha_1\cup\alpha_{n-1})\!+\! 
p^{min}_{\emptyset}\!\left((\beta_1\!\setminus\!\tilde{\beta}_1)\!\cup\!(\beta_n\!\setminus\!\tilde{\beta}_n)\right)\!-\!1\right\}
$
and,
for any set of tuples $\gamma$, 
$p^{min}_{\emptyset}\!(\gamma)\!=\!
\max\!\left\{ 0, \sum_{t\in \gamma} p(t)\!-\!|\gamma|\!+\!1 \right\}.$
\end{lemma}
\begin{proof}
$p(\tilde{\beta}_1\cup \tilde{\beta}_n)$ can be minimized as follows.\\
$1)$ We start from any model $M$ of $D^p$ minimizing the portion of the probability space 
where neither the event $\tilde{\beta}_1$ nor the event $\tilde{\beta}_n$ can occur.
That is, $M$ is any model minimizing the probability of the event 
$E=\alpha_1\cup(\beta_1\!\setminus\! \tilde{\beta}_1)\cup \alpha_{n-1}\cup (\beta_n\!\setminus\!\tilde{\beta}_n)$ (this event is mutually exclusive with both $\tilde{\beta}_1$ 
and $\tilde{\beta}_n$ due to hyperedges $e_1$ and $e_n$).
It is easy to see that $M$ is also a model for $D^p$ w.r.t. the conflict hypergraph 
\emph{sub}$(\mathcal{C})$, and that the minimum probability 
$p^{min}_{\mbox{\em\scriptsize{sub}}(\mathcal{C})}(E)$ of $E$ w.r.t. \emph{sub}$(\mathcal{C})$ 
is equal to the minimum probability $p^{min}_{\mathcal{C}}(E)$ of $E$ w.r.t. $\mathcal{C}$.
We denote this probability as $Y$.\\
$2)$ We re-distribute the tuples in $\tilde{\beta_1}\cup\tilde{\beta}_n$ over the portion of size $1\!-\!Y$ of the probability space not assigned to $E$, so that 
$p(\tilde{\beta_1})=p^{min}_{\emptyset}(\tilde{\beta_1})$  and 
$p(\tilde{\beta_2})=p^{min}_{\emptyset}(\tilde{\beta_n})$, and with the aim of minimizing the intersection of the events $\tilde{\beta}_1$ and $\tilde{\beta}_n$.
The fact that the events $\tilde{\beta_1}$ and $\tilde{\beta_n}$ can be simultaneously
assigned their minimum probabilities  $p^{min}_{\emptyset}(\tilde{\beta_1})$ and 
$p^{min}_{\emptyset}(\tilde{\beta_n})$, respectively, derives from Lemma~\ref{lem:pmin1/2} and 
from the consistency of $D^p$ w.r.t. $\mathcal{C}$.
This yields a (possibly) new model $M'$ for $D^p$ w.r.t. the ``original'' chain $\mathcal{C}$
where 
$p(\tilde{\beta}_1\cup \tilde{\beta}_n)=
\max
\left\{ 
0, \ p^{min}(\tilde{\beta}_1)+ p^{min}(\tilde{\beta}_n)-[1-Y]\right\}.$
In fact, viewing the available probability space as a segment of length $1-Y$, this corresponds
to assigning the left-most part of the segment of length $p^{min}(\tilde{\beta}_1)$ to event 
$\tilde{\beta}_1$, and the right-most part of length $p^{min}(\tilde{\beta}_n)$ to event 
$\tilde{\beta}_n$.
This way, the probability of the intersection is the length of the segment portion (if any)
assigned to both $\tilde{\beta}_1$ and $\tilde{\beta}_n$.
In brief, we obtain the formula reported in the statement for $p^{min}(\tilde{\beta}_1\cup \tilde{\beta}_n)$.

The formula for 
$p^{min}(\alpha_1\cup(\beta_1\!\setminus\! \tilde{\beta}_1)\cup \alpha_{n-1}\cup (\beta_n\!\setminus\!\tilde{\beta}_n))$ can be proved with an analogous reasoning, while the formula for $p^{min}_{\emptyset}(\gamma)$ follows from Lemma~\ref{lem:pmin1/2}.
\end{proof}

\noindent{\textbf{Theorem~\ref{theo:ring}.} }
\textit{Given an instance $D^p$ of $\D^p$, if $HG(D^p, \IC)=\langle N, E\rangle$ is a ring, then $D^p\models \IC$
iff both the following hold:
1)
$\forall e \in E,\  \sum_{t\in e} p(t) \leq |e|-1$;
\hspace*{1mm}
2)
$\sum_{t \in N} p(t) - |N| + \lceil \frac{|E|}{2}\rceil \leq 0$.
}
\begin{proof}
In the following, we will denote the ring $HG(D^p,$ $\IC)$ as 
$\mathcal{R}=e_1, \dots, e_n, e_{n+1}$,
and, for each $i\in[1..n+1]$, the ears of $e_i$ as $\varepsilon_i$, and,
for each $i\in[1..n]$, the intersection $e_i\cap e_{i+1}$ as $\gamma_i$, and 
$e_1\cap e_{n+1}$ as $\gamma_0$.
Moreover, we will denote as $\mathcal{C}=e_1, \dots, e_n$ the chain obtained from ring 
$\mathcal{R}$ by removing the edge $e_{n+1}$.
We now prove the left-to-right and right-to-left implications separately.\\[-4pt]

\noindent
$(\Rightarrow)$:
We first show that, if $D^p\models \IC$ and $HG(D^p, \IC)$ is a ring, then both 
Condition $1.$ and $2.$ hold.
Condition $1.$ trivially follows from the fact that the proof of the left-to-right 
implication of Theorem~\ref{theo:tractableCC} holds for general conflict hypergraphs.

We now focus on Condition $2.$
As $D^p$ is consistent w.r.t. $\mathcal{R}$, the presence of hyperedge $e_{n+1}$ in 
$HG(D^p, \IC)$ implies that the minimum probability that the tuples in $e_{n+1}$ co-exist is equal to $0$.
That is, $p^{min}_{\mathcal{R}}((\gamma_0\cup \gamma_n)\cup \varepsilon_{n+1})=0$.
On the other hand, $p^{min}_{\mathcal{C}}((\gamma_0\cup \gamma_n)\cup \varepsilon_{n+1})\leq 
p^{min}_{\mathcal{R}}((\gamma_0\cup \gamma_n)\cup \varepsilon_{n+1})$, thus it must hold that
$p^{min}_{\mathcal{C}}((\gamma_0\cup \gamma_n)\cup \varepsilon_{n+1})=0$.
Since, according to the conflict hypergraph $\mathcal{C}$, no correlation is imposed between
the events $(\gamma_0\cup \gamma_n)$ and $\varepsilon_{n+1}$, we also have that
$p^{min}_{\mathcal{C}}((\gamma_0\cup \gamma_n)\cup \varepsilon_{n+1})=
\max\{0, p^{min}_{\mathcal{C}}(\gamma_0\cup \gamma_n)+p^{min}_\emptyset(\varepsilon_{n+1})-1\}$ 
(see Fact~\ref{fac:boole}).
Hence, the following inequality must hold:
\begin{equation}
p^{min}_{\mathcal{C}}(\gamma_0\cup \gamma_n)+p^{min}_\emptyset(\varepsilon_{n+1})-1\leq 0.
\label{eq:provaring}
\end{equation}

We now show that inequality (\ref{eq:provaring}) entails that Condition $2.$ holds.
First, observe that $\gamma_0$ and $\gamma_n$ are subsets of the ears of $e_1$ and $e_n$, 
respectively, w.r.t. the hypergraph $\mathcal{C}$. 
Hence, since $\mathcal{C}$ is a chain, we can apply Lemma~\ref{lem:pmin-chain-hyper} to obtain
$p^{min}_{\mathcal{C}}(\gamma_0\cup \gamma_n)$ in function of $p^{min}_{sub(\mathcal{C})}(\gamma_1\cup \gamma_{n-1})$.
Thus, by recursively applying ($\lfloor\frac{n}{2}\rfloor$ times) Lemma~\ref{lem:pmin-chain-hyper}, we obtain the following expression for
$p^{min}_{\mathcal{C}}(\gamma_0\cup \gamma_n)$ 
(where $x=\lfloor\frac{n}{2}\rfloor-1$ and $y=\lceil\frac{n}{2}\rceil+1$):

$$
\begin{array}{l}
\max
\left\{ 
0, 
\max\left\{ 0, 
\sum_{t\in \gamma_0} p(t) - |\gamma_0|+1 \right\}+
\max\left\{ 0, 
\sum_{t\in \gamma_n} p(t) - |\gamma_n|+1 \right\}-1+\right. \\ 
\hspace*{20mm}\max \left\{ 0,
\max \left\{  0,
\max\left\{ 0, 
\sum_{t\in \gamma_1 } p(t) - |\gamma_1|+1 \right\}+
\max\left\{ 0, 
\sum_{t\in \gamma_{n-1} } p(t) - |\gamma_{n-1}|+1 \right\}-1+\right.\right.\\
\hspace*{40mm}\dots\\
\hspace*{45mm}\max
\left\{
0, 
\max\left\{ 0, 
\sum_{t\in \gamma_{x}} p(t) - |\gamma_x|+1 \right\}+
\max\left\{ 0, 
\sum_{t\in \gamma_y} p(t) -|\gamma_y|+1 \right\}
-1+P\right\}+\\
\hspace*{40mm}\dots\\
\hspace*{31mm}\max\left\{ 0, 
\sum_{t\in (\varepsilon_2\cup\varepsilon_{n-1})} p(t)\!-\! 
|\varepsilon_2\cup\varepsilon_{n-1}|\!+1\! \right\}-1 \\
\hspace*{27mm}
\left. \right\}+\\
\hspace*{11mm}\max\left\{ 0, 
\sum_{t\in (\varepsilon_1\cup\varepsilon_n)} p(t) - 
|\varepsilon_1\cup\varepsilon_n|\!+1\! \right\}-1 \\
\hspace*{7mm}
\left. \right\}\\
\end{array}
$$

\noindent
where:
$$P=\left\{
\begin{array}{ll}
p^{min}_{\emptyset}(\gamma_{x+1}) & \mbox{if $n$ is even;}\\
p^{min}_{e_{y-1}}(\gamma_{x+1}\cup\gamma_{y-1}) & \mbox{otherwise}.
\end{array}
\right.$$

In this formula, 
$p^{min}_\emptyset(\gamma_{x+1})=\max\{0, \sum_{t\in \gamma_{x+1}} p(t) - |\gamma_{x+1}| +1 \}$, and
$p^{min}_{e_{y-1}}(\gamma_{x+1}\cup\gamma_{y-1})=
\max\{0,\\
\sum_{t\in (\gamma_{x+1}\cup\gamma_{y-1})} p(t) - |(\gamma_{x+1}\cup\gamma_{y-1})| +1
\}$
(the latter follows from applying Lemma~\ref{lem:pmin1/2}).


The value of $p^{min}_{\mathcal{C}}(\gamma_0\cup \gamma_n)$ is greater than or equal to
the sum $S$ of the non-zero terms that occur in the expression obtained so far, that is:
$$
S=\left\{
\begin{array}{l}
\sum_{t \in (N\setminus\varepsilon_{n+1})} p(t) - (|N|-|\varepsilon_{n+1}|) 
+ \frac{n}{2}+1,	\\
\mbox{if the length $n$ of the chain $\mathcal{C}$ is even;}\\ \\
\sum_{t \in (N\setminus\varepsilon_{n+1})} p(t) - 
(|N|\!-\!|\varepsilon_{n+1}|\! -\!|\varepsilon_{x+2}|)\! 
+\! \lfloor \frac{n}{2}\rfloor\! +1, \\ 
\mbox{if the length $n$ of $\mathcal{C}$ is odd.}
\end{array}
\right.$$

The fact that $p^{min}_{\mathcal{C}}(\gamma_0\cup \gamma_n)\geq S$ straightforwardly 
follows from that $S$ is obtained by summing also 
possibly negative contributions of terms of the form 
$p^{min}_\emptyset(Z)=\sum_{t\in Z} p(t) - |Z|+1$,
which are not considered
when evaluating $p^{min}_{\mathcal{C}}(\gamma_0\cup \gamma_n)$, 
since invocations of the $\max$ function return non-negative values only.

As the number of edges in the ring $\mathcal{R}$ is $|E|=n+1$,
the value of $S$ is in every case greater than or equal to
$$S'= \sum_{t \in (N\setminus\varepsilon_{n+1})} p(t) - (|N|-|\varepsilon_{n+1}|) + \left\lceil \frac{|E|}{2}\right\rceil
$$

In brief, we have obtained $S'\leq S \leq p^{min}_{\mathcal{C}}(\gamma_0\cup \gamma_n)$.


Since $D^p\models\IC$ implies that 
$p^{min}_{\mathcal{C}}(\gamma_0\cup \gamma_n)+p^{min}_\emptyset(\varepsilon_{n+1})-1\leq 0$ (equation (\ref{eq:provaring})),
we obtain $S'+ p^{min}_\emptyset(\varepsilon_{n+1})-1\leq 0$.
By replacing $S'$ and $p^{min}_\emptyset(\varepsilon_{n+1})$ with the corresponding formulas, we obtain
$$
\sum_{t \in (N\setminus\varepsilon_{n+1})} p(t) -  (|N|-|\varepsilon_{n+1}|) + \lceil \frac{|E|}{2}\rceil +\sum_{t\in \varepsilon_{n+1}} p(t)-|\varepsilon_{n+1}| \leq 0
$$
that is,
$\sum_{t \in N} p(t) - |N| + \left\lceil \frac{|E|}{2}\right\rceil\leq 0$.\\[-4pt]

\noindent
$(\Leftarrow)$:
We now prove the right-to-left implication, reasoning by contradiction.
Assume that both Condition $1.$ and $2.$ hold, but $D^p$ is not consistent w.r.t. the 
conflict hypergraph $\mathcal{R}$.
However, since $\mathcal{C}$ is a hypertree and Condition $1.$ holds, from Theorem~\ref{theo:tractableCC} we have that $D^p$ is consistent w.r.t. the 
conflict hypergraph $\mathcal{C}$.
In particular, it must be the case that 
$p^{min}_{\mathcal{C}}(e_{n+1})= p^{min}_{\mathcal{C}}((\gamma_0\cup \gamma_n)\cup \varepsilon_{n+1})>0$: otherwise, any model of $D^p$ w.r.t. $\mathcal{C}$ assigning probability $0$ to 
the event $(\gamma_0\cup \gamma_n)\cup \varepsilon_{n+1}$ would be also a model 
for $D^p$ w.r.t. $\mathcal{R}$, which is in contrast with the contradiction hypothesis.

Since, according to the conflict hypergraph $\mathcal{C}$, no correlation is imposed between
the events $(\gamma_0\cup \gamma_n)$ and $\varepsilon_{n+1}$, we also have that
$p^{min}_{\mathcal{C}}((\gamma_0\cup \gamma_n)\cup \varepsilon_{n+1})=
\max\{0, p^{min}_{\mathcal{C}}(\gamma_0\cup \gamma_n)+p^{min}_\emptyset(\varepsilon_{n+1})-1\}$ 
(see Fact~\ref{fac:boole}).
Hence, the following inequality must hold:
\begin{equation}\label{eq:provaRing2}
p^{min}_{\mathcal{C}}(\gamma_0\cup \gamma_n)+p^{min}_\emptyset(\varepsilon_{n+1})-1>0
\end{equation}
which also implies both 
$p^{min}_{\mathcal{C}}(\gamma_0\cup \gamma_n)>0$ and
$p^{min}_{\emptyset}(\varepsilon_{n+1})> 0$
(as probabilities values are bounded by $1$).

By applying Lemma~\ref{lem:pmin-pmax-chain}, we obtain that 
$p^{min}_{\mathcal{C}}(\gamma_0\cup\gamma_n)$ is equal to 
$$
\begin{array}{l}
\max\{0, 
p^{min}_{\emptyset}(\gamma_0)+
p^{min}_{\emptyset}(\gamma_n)
-1
+\max\{0, p^{min}_{sub(\mathcal{C})}(\gamma_1\cup \gamma_{n-1})
+
p^{min}_{\emptyset}(\varepsilon_1\cup\varepsilon_n)
-1
\}
\}
\end{array}
$$

As shown above, $p^{min}_{\mathcal{C}}(\gamma_0\cup\gamma_n)>0$, thus
the expression for $p^{min}_{\mathcal{C}}(\gamma_0\cup\gamma_n)$ can be 
simplified into:
$$
\begin{array}{l}
p^{min}_{\emptyset}(\gamma_0)+
p^{min}_{\emptyset}(\gamma_n)
-1
+\max\{0, p^{min}_{sub(\mathcal{C})}(\gamma_1\cup \gamma_{n-1})
+
p^{min}_{\emptyset}(\varepsilon_1\cup\varepsilon_n)
-1
\}
\end{array}
$$

By replacing $p^{min}_{\mathcal{C}}(\gamma_0\cup\gamma_n)$ with this  formula in 
equation (\ref{eq:provaRing2}), we obtain
\begin{equation}\label{eq:provaRing3}
\begin{array}{l}
p^{min}_{\emptyset}(\gamma_0)+
p^{min}_{\emptyset}(\gamma_n)+
p^{min}_\emptyset(\varepsilon_{n+1})-2+
\max\{0, p^{min}_{sub(\mathcal{C})}(\gamma_1\cup \gamma_{n-1})
+
p^{min}_{\emptyset}(\varepsilon_1\cup\varepsilon_n)
-1
\}
>0
\end{array}
\end{equation}

Since $p^{min}_{\emptyset}(\gamma_0)+ p^{min}_{\emptyset}(\gamma_n)+ 
p^{min}_\emptyset(\varepsilon_{n+1})-2
\leq 
p^{min}_{\emptyset}(\gamma_0\cup \gamma_n \cup \varepsilon_{n+1})$
(which follows from applying twice Fact~\ref{fac:boole}), and 
$p^{min}_{\emptyset}(\gamma_0\cup \gamma_n \cup \varepsilon_{n+1})=\max
\{0, \sum_{t\in (\gamma_0\cup \gamma_n\cup\varepsilon_{n+1})} p(t)-|(\gamma_0\cup \gamma_n\cup\varepsilon_{n+1})|+1\}$, and
$\sum_{t\in (\gamma_0\cup \gamma_n\cup\varepsilon_{n+1})} p(t)-|(\gamma_0\cup \gamma_n\cup\varepsilon_{n+1})|+1\leq 0$ (Condition $1.$ over hyperedge $e_{n+1}$),
we obtain that 
$p^{min}_{\emptyset}(\gamma_0)+ p^{min}_{\emptyset}(\gamma_n)+ 
p^{min}_\emptyset(\varepsilon_{n+1})-2\leq 0$.
Hence, the second argument of $\max$ in equation (\ref{eq:provaRing3}) must be strictly positive, thus equation (\ref{eq:provaRing3}) can be rewritten as:
\begin{equation}\label{eq:provaRing4}
\begin{array}{l}
p^{min}_{\emptyset}(\gamma_0)+
p^{min}_{\emptyset}(\gamma_n)+
p^{min}_\emptyset(\varepsilon_{n+1})-2
+p^{min}_{sub(\mathcal{C})}(\gamma_1\cup \gamma_{n-1})
+
p^{min}_{\emptyset}(\varepsilon_1\cup\varepsilon_n)
-1
>0
\end{array}
\end{equation}
where $p^{min}_{sub(\mathcal{C})}(\gamma_1\cup \gamma_{n-1})>0$ and 
$p^{min}_{\emptyset}(\varepsilon_1\cup\varepsilon_n)>0$ (otherwise, the second argument of 
$\max$ in equation (\ref{eq:provaRing3}) could not be strictly positive, being probability 
values bounded by $1$).

Observe that all the terms of the form $p^{min}$ occurring in (\ref{eq:provaRing4}) are strictly positive.
In fact, we have already shown that this holds for $p^{min}_\emptyset(\varepsilon_{n+1})$,
$p^{min}_{sub(\mathcal{C})}(\gamma_1\cup \gamma_{n-1})$, and $p^{min}_{\emptyset}(\varepsilon_1\cup\varepsilon_n)$.
As regards $p^{min}_{\emptyset}(\gamma_0)$, the fact that it is strictly greater than $0$ 
derives from the $p^{min}_{\emptyset}(\gamma_0)=p^{min}_{\mathcal{C}}(\gamma_0)$
(which is due to Lemma~\ref{lem:pmin1/2}, as $\gamma_0$ is a matryoshka w.r.t. 
$\mathcal{C}$), and $p^{min}_{\mathcal{C}}(\gamma_0)\geq p^{min}_{\mathcal{C}}(\gamma_0\cup \gamma_n)$, where $p^{min}_{\mathcal{C}}(\gamma_0\cup \gamma_n)>0$, as shown before.
The same reasoning suffices to prove that $p^{min}_{\emptyset}(\gamma_n)>0$.

The fact that all the terms of the form $p^{min}_\emptyset$  in (\ref{eq:provaRing4}) are strictly positive implies that we can replace them with the corresponding formulas given in
Lemma~\ref{lem:pmin1/2}, simplified by eliminating the $\max$ operator.
Therefore, we obtain:
\begin{equation}\label{eq:provaRing5}
\begin{array}{l}
\left(\sum_{t\in \gamma_0} p(t) - |\gamma_0| +1\right)+
\left(\sum_{t\in \gamma_n} p(t) - |\gamma_n| +1\right)+
\left(\sum_{t\in \varepsilon_{n+1}} p(t) - |\varepsilon_{n+1}| +1\right)+\
p^{min}_{sub(\mathcal{C})}(\gamma_1\cup \gamma_{n-1})+\\
+
\left(\sum_{t\in \varepsilon_{1}} p(t) - |\varepsilon_{1}| +1\right) +
\left(\sum_{t\in \varepsilon_{n}} p(t) - |\varepsilon_{n}| +1\right) -1
-3
>0
\end{array}
\end{equation}

By recursively applying the same reasoning on $p^{min}_{sub(\mathcal{C})}(\gamma_1\cup \gamma_{n-1})$ a number of times equal to $\lfloor \frac{n}{2}\rfloor$, 
the term on the left-hand side of equation (\ref{eq:provaRing5}) can be shown to be less than
or equal to $\sum_{t\in N} p(t) -|N|+\left\lceil\frac{|E|}{2}\right\rceil$ (depending on 
whether $n$ is even or not, analogously to the proof of the inverse implication).
Thus, we obtain $\sum_{t\in N} p(t) -|N|+\left\lceil\frac{|E|}{2}\right\rceil>0$, which contradicts Condition $2$.
\end{proof}

\subsection{Proofs of theorems~\ref{theo:ic-empty-phi},~\ref{theo:ic-arity-3},~\ref{theo:EGDs-cc}, ~\ref{theo:2FDs-cc}, and~\ref{theo:tract-EGDs-and-join-free}}
\label{app:567}


\ \\
\noindent{\textbf{Theorem~\ref{theo:ic-empty-phi}.} }
\textit{If $\IC$ consists of a join-free denial constraint, then \cc\ is in \textit{PTIME}.
In particular, $D^p\models \IC$ iff, for each hyperedge $e$ of $HG(D^p, \IC)$, it holds 
that $\sum_{t\in e}\!p(t)\!\leq\!|e|\!-\!1$.}
\begin{proof}

Let $\IC$ consist of the denial constraint $ic$ having the form:  
$\neg [R_1(\vec{x}_1)\wedge\dots\wedge R_m(\vec{x}_m)\wedge \phi_1(\vec{x}_1)\wedge\dots\wedge \phi_m(\vec{x}_m)]$, 
where no variable occurs in two distinct relation atoms of $ic$,
and, for each built-in predicate occurring in 
$\phi_1(\vec{x}_1)\wedge\dots\wedge \phi_m(\vec{x}_m)$ 
at least one term is a constant.
Given an instance $D^p$ of $\D^p$, 
we show that $D^p\models \IC$
iff 
for each hyperedge $e$ of $HG(D^p, \IC)$, it holds 
that $\sum_{t\in e}\!p(t)\!\leq\!|e|\!-\!1$.

$(\Rightarrow)$: It straightforwardly follows for the fact that, as pointed out in the core of the paper after Theorem~\ref{theo:tractableCC},
the condition that, for each hyperedge $e$ of $HG(D^p, \IC)$, $\sum_{t\in e}\!p(t)\!\leq\!|e|\!-\!1$ is a necessary condition for the 
consistency in the presence of any conflict hypergraph.

$(\Leftarrow)$:
For each $i\in [1..m]$, let $R_{\phi_i}$ be the maximal set of tuples in 
the instance of $R_i$ such that every tuple $t_i\in R_{\phi_i}$ satisfies 
$R_i(\vec{x}_i)\wedge\phi_i(\vec{x}_i)$.

It is easy to see that $HG(D^p, \IC)$ consists of the set of hyperedges 
$\left\{ \{t_1,\dots,t_m\}\,|\, \forall i\in[1..m]\, t_i\in R_{\phi_i}\right\}$.
Observe that not all the hyperdeges in $HG(D^p, \IC)$ have size $m$, as the same
relation scheme may appear several times in $ic$. That is, in the case that
there are $i,j\in [1..m]$ with $i< j$ such that
$R_{\phi_i}\cap R_{\phi_j}\neq \emptyset$, the tuples $t_i$ and $t_j$ occurring in the
same hyperedge $\{t_1,\dots,t_i,\dots,t_j,\dots, t_m\}$ may coincide, thus this hyperedge 
has size less than $m$.

From the hypothesis, it holds that, for every hyperedge $e$ of 
$HG(D^p, \IC)$, it must be the case that $\sum_{t\in e} p(t)\leq |e|-1$.
Let $e^*$ be the hyperedge in $HG(D^p, \IC)$ such that
$|e|-1-\sum_{t\in e} p(t)$ is the minimum, that is, 
$$e^* = argmin_{e\in HG(D^p, \IC)}\left(|e|-1-\sum_{t\in e} p(t)\right).$$

For the sake of simplicity of presentation we consider the case that
$e^*$ has size $m$, and denote its tuples as $t_1,\dots,t_m$. 
The generalization to the case that the
size of $e^*$ is less than $m$ is straightforward.

Let $S$ be a subset of $D^p$. 
We denote with $D^p_S$ the subset of $D^p$ containing only the tuples in $S$.
Let $Pr_{e^*}$ be a model in $\M(D^p_{e^*},\IC)$.
Moreover, let $t'_1,\dots,t'_n$ be the tuples in $D^p/e^*$.

In the following, we will define a sequence of interpretations $Pr_0,Pr_1,\dots,Pr_n$
such that, for each $i\in[0..n]$, 
$Pr_i$ is a model in $\M(D^p_{e^*\cup \{t'_j| j\leq i\}},\IC)$.

We start by taking $Pr_0$ equal to $Pr_{e^*}$.
At the $i^{th}$ step we consider tuple $t'_i$ and define $Pr_i$ as follows:
\begin{enumerate}
\item
In the case that, for each $j\in[1..m]$, it holds that $t'_i\not\in R_{\phi_j}$,
we define, 
for each possible world $w$ in \emph{pwd}$(e^*\cup \{t'_j| j\leq i\})$,
$Pr_i(w) = Pr_{i-1}(w\setminus\{t'_i\})\cdot p(t'_i)$, if $t'_i\in w$,
and $Pr_i(w) = Pr_{i-1}(w\setminus\{t'_i\})\cdot (1-p(t'_i))$, otherwise.
\item
Otherwise, if there is $j\in[1..m]$ such that $t'_i\in R_{\phi_j}$,
we consider the set $J$ of all the indexes $j\in[1..m]$ such that $t'_i\in R_{\phi_j}$.
Moreover, we denote with $p_J$ the sum of the probabilities  (computed according to 
$Pr_{i-1}$) of
all the possible worlds $w\in$ \emph{pwd}$(e^*\cup \{t'_j| j\leq i-1\})$ 
such that, for each $j\in J$, the corresponding tuple $t_j$ appearing in $e^*$ belongs also
to $w$, i.e.,
$p_J = \sum_{w\in pwd(e^*\cup \{t'_j| j\leq i-1\}), 
s.t. \forall j\in J\, t_j\in w}Pr_{i-1}(w)$.

Then, for each possible world $w$ in \emph{pwd}$(e^*\cup \{t'_j| j\leq i\})$,
we define $Pr_i$ as follows:
\begin{itemize}
\item
$Pr_i(w) = Pr_{i-1}(w-\{t'_i\})\cdot \frac{p(t'_i)}{p_J}$, if $t'_i\in w$ and for each $j\in J$ it holds that $t_j\in w$,
\item
$Pr_i(w) = Pr_{i-1}(w-\{t'_i\})\cdot \frac{max(0,p_J- p(t'_i))}{p_J}$, if $t'_i\not\in w$ and for each $j\in J$ it holds that $t_j\in w$,
\item
$Pr_i(w) = Pr_{i-1}(w)$, if $t'_i\not\in w$ and there is a $j\in J$ such that
$t_j\not\in w$,
\item
$Pr_i(w) = 0$, otherwise.
\end{itemize}
\end{enumerate}

We prove that for each $i\in[0..n]$ it holds that
$Pr_i$ is a model in $\M(D^p_{e^*\cup \{t'_j| j\leq i\}},\IC)$
reasoning by induction on $i$.
The proof is straightforward for $i=0$.
We now prove the induction step, that is, we assume that
$Pr_{i-1}$ is a model in $\M(D^p_{e^*\cup \{t'_j| j\leq i-1\}},\IC)$
and prove that $Pr_i$ is a model in $\M(D^p_{e^*\cup \{t'_j| j\leq i\}},\IC)$.

As regards the first case of the definition of $Pr_i$ from $Pr_{i-1}$,
it is easy to see that
$Pr_i$ is a model in $\M(D^p_{e^*\cup \{t'_j| j\leq i\}},\IC)$
since $Pr_i$ consists in a trivial extension of $Pr_{i-1}$
which takes into account a tuple not correlated with the other
tuples in the database.

As regards the second case of the definition of $Pr_i$ from $Pr_{i-1}$, 
it is easy to see that, if 
$p_J\geq p(t'_i)$
than $Pr_i$ guarantees that the condition about the marginal probabilities of all
the tuples in $e^*\cup \{t'_j| j\leq i\}$ holds.
Moreover, $Pr_j$ assigns zero probability to each possible world
$w$ such that $w\not\models \IC$, since, for each
possible world $w$ in \emph{pwd}$(e^*\cup \{t'_j| j\leq i\})$,
there is no subset $S$ of $w$ such that 
for each $i\in[1..m]$ there is a tuple $t\in S$ such that $t\in R_{\phi_i}$.
The latter follows from the induction hypothesis, which ensures that $Pr_{i-1}$
is a model in $\M(D^p_{e^*\cup \{t'_j| j\leq i-1\}},\IC)$, and from the fact that
$Pr_i$ assigns non-zero probability to a possible 
world $w$ in \emph{pwd}$(e^*\cup \{t'_j| j\leq i\})$ containing $t'_i$ iff
for each $j\in J$ it holds that $t_j\in w$.
Specifically, it can not be the case that $w$ contains, for each
$x\in [1..m]$ such that $x\not\in J$ a tuple $t_x\in R_{\phi_i}$,
as otherwise $w - \{t'_i\}$ would satisfy all the conditions
expressed in $ic$, and $w - \{t'_i\}$ would be assigned a non-zero probability by
$Pr_{i-1}$, thus contradicting the induction hypothesis that
$Pr_{i-1}$ is a model in $\M(D^p_{e^*\cup \{t'_j| j\leq i-1\}},\IC)$.

We now prove that $p_J\geq p(t'_i)$. 
Reasoning by contradiction, assume that $p_J< p(t'_i)$.
From the definition of $p_J$ it
follows that $p_J\geq p^{min}(\wedge_{j\in J} t_j)$.
Therefore, since $p^{min}(\wedge_{j\in J} t_j)$ is equal to
$\max\left\{0, \sum_{j\in J} p(t_j)-|J|+1\right\}$
it follows that $p(t'_i)> \sum_{j\in J} p(t_j)-|J|+1$.
Consider the hyperdege 
$e = \{t_x| t_x\in e^* \wedge x\not\in J\}\cup \{t'_i\}$.
From the definition of $e^*$ it follows that
$|e|-1-\sum_{t\in e} p(t)\geq |e^*|-1-\sum_{t\in e^*} p(t)$.
The latter implies that
$1-p(t'_i) \geq |J|-\sum_{j\in J} p(t_j)$, from
which it follows that
$\sum_{j\in J} p(t_j) -|J|+1 \geq p(t'_i)$
which is a contradiction.
Hence, we can conclude that, in this case
$Pr_i$ is a model in $\M(D^p_{e^*\cup \{t'_j| j\leq i\}},\IC)$.

This conclude the proof, as 
$Pr_n$ is a model in $\M(D^p,\IC)$
and then $D^p\models \IC$.
\end{proof}

\ \\
\noindent{\textbf{Theorem~\ref{theo:ic-arity-3}.} }
\textit{
There is an $\IC$ consisting of a non-join-free denial constraint of arity $3$ 
such that \cc\ is $NP$-hard.}
\begin{proof}
\emph{The reader is kindly requested to read this proof after that of Theorem~\ref{theo:2FDs-cc}, as the construction used there will be exploited in the reasoning used below}.

We show that the reduction from \textsc{3-coloring} 
to \cc\ presented in the hardness proof of Theorem~\ref{theo:2FDs-cc} can
be rewritten to obtain an instance of \cc\  where 
$\IC$ contains only a denial constraints having arity equal to $3$. 

Let $G=\la N,E \ra$ be a \textsc{3-coloring} instance.
We construct an equivalent instance  $\la \D^p,\IC,D^p \ra$ of \cc\ as follows:
\begin{list}{--}
     {\setlength{\rightmargin}{0mm}
      \setlength{\leftmargin}{3mm}
      \setlength{\itemindent}{0mm} 
			\setlength{\itemsep}{0mm}}
\item
$\D^p$ consists of the probabilistic relation schemas
$R^p_1($\emph{Node, Color, P}$)$ and
$R^p_2($\emph{Node1, Node2, Color1, Color2, P}$)$;
\item 
$D^p$ is the instance of $\D^p$ consisting of the instances $r^p_1$ of $R^p_1$, 
and $r^p_2$ of $R^p_2$, defined as follows:
\begin{itemize}
\item
for each node $n\in N$, and for each color $c\in\{$\emph{Red, Green, Blue}$\}$,
$r^p_1$ contains the tuple $(n, c, \frac{1}{3})$;
\item
for each edge $\{n_1,n_2\}\in E$, and for each color $c\in\{$\emph{Red, Green, Blue}$\}$,
$r^p_2$ contains the tuple $(n_1, n_2, c, c, 1)$;\\
moreover, 
for each node $n\in N$, and for each pair of distinct colors 
$c_1, c_2\in\{$\emph{Red, Green, Blue}$\}$,
$r^p_2$ contains the tuple $(n, n, c_1, c_2, 1)$;
\end{itemize}
\item
$\IC$ is the set of denial constraints over $\D^p$ consisting of the 
constraint:
$\neg [R_1(x_1,x_2)\wedge R_1(x_3,x_4)\wedge R_2(x_1,x_3,x_2,x_4)]$.
\end{list}

Basically, the constraint in $\IC$ imposes that adjacent nodes can not be
assigned the same color, and the same node can not be assigned more than one color.

Let $\la \overline{\D}^p,\overline{\IC},\overline{D}^p \ra$ be 
the instance of \cc\ defined in the hardness proof of Theorem~\ref{theo:CCisNPcomplete},
where it was shown that an instance $G$ of \textsc{3-coloring} is 3-colorable iff $\overline{D}^p\models \overline{\IC}$. 
It is easy to see that $D^p\models\IC$ iff $\overline{D}^p\models \overline{\IC}$,
which completes the proof.
\end{proof}

\ \\
\noindent{\textbf{Theorem~\ref{theo:EGDs-cc}.}}
\hspace*{-1mm}\textit{If $\IC$ consists of a BEGD, then \cc\ is in \textit{PTIME}.}
\begin{proof}
Let the BEGD in $\IC$ be:
$$ic= \neg [R_1(\vec{x},\vec{y}_1)\wedge R_2(\vec{x},\vec{y}_2)\wedge z_1\neq z_2 ],$$
where each $z_i$ (with $i\in\{1,2\}$) is a variable in $\vec{y}_1$ or $\vec{y}_2$.
That is, for the sake of presentation, we assume that the conjunction of built-in predicates
in $ic$ consists of one conjunct only (this yields no loss of generality, as it is easy to 
see that the reasoning used in the proof is still valid in the presence of more conjuncts).
We consider two cases separately.\\[2pt]
\noindent
\textit{\underline{Case 1}}: 
$R_1=R_2$, 
that is, only one relation name occurs in $ic$.
Let $\overline{X}$ be the set of attributes in \emph{Attr}$(R_1)$ corresponding to the 
variables in $\vec{x}$, and 
let $Z_1$ and $Z_2$ be the attributes in \emph{Attr}$(R_1)$ corresponding to the variables $z_1$ and $z_2$,
respectively.
Let $r$ be an instance of $R_1$.

It is easy to see that the conflict hypergraph $HG(r, \IC)$ is a graph having the following 
structure: for any pair of tuples $t_1$, $t_2$,
there is the edge $(t_1,t_2)$ in $HG(r, ic)$ iff: 
$1)$ $\forall X\in\overline{X}$, $t_1[X]=t_2[X]$,
and
$2)$ $t_1[Z_1]\neq t_2[Z_2]$.
%

This structure of the conflict hypergraph implies a partition of the tuples of $r$, where the
tuples in each set of the partition share the same values of the attributes in $\overline{X}$.
Obviously, \cc\ can be decided by considering these sets separately.
%

For each set $G$ of this partition, we reason as follows.
Let $\mathcal{P}_G$ be the set of pairs of values
$\langle c_1,c_2 \rangle$ occurring as values of attributes $Z_1$ and $Z_2$ in at least one tuple of $r$ (that is, $\mathcal{P}_G$ is the projection of $r$ over $Z_1$ and $Z_2$).
For each pair $\langle c_1,c_2 \rangle\in \mathcal{P}_G$,
let $T[c_1,c_2]$ be the set of tuples in $G$ such that, $\forall t\in  T[c_1,c_2]$,
$t[Z_1]=c_1$ and $t[Z_2]=c_2$.
A first necessary condition for consistency is that there is no pair $\langle c_1,c_2 \rangle\in \mathcal{P}_G$ such that $c_1\neq c_2$:
otherwise, any tuple in $T[c_1,c_2]$ would not satisfy the constraint, thus it would not be possible to put it in any possible world with non-zero probability\footnote{Obviously, we assume that there is no tuple with zero probability, as tuples with zero probability can be discarded from the database instance.}.
Straightforwardly, this condition is also sufficient if $z_1$ and $z_2$ belong to the same relation atom.
Thus, in this case, the proof ends, as checking this condition can be done in polynomial time.

Otherwise, if $z_1$ and $z_2$ belong to different relation atoms and if the above-introduced
necessary condition holds, we proceed as follows.
From what said above, it must be the case that $\mathcal{P}_G$ contains only pairs of the form
$\langle c,c \rangle$, and, correspondingly, all the sets $T[c_1,c_2]$ are of the form 
$T[c,c]$.
For each $T[c,c]$, let $\widetilde{p}(T[c,c])$ be the maximum probability of the tuples in 
$T[c,c]$ (i.e., $\widetilde{p}(T[c,c])=\max_{t\in T[c,c]}\{p(t)\}$.
Moreover, for each $\langle c,c \rangle\in \mathcal{P}_G$, take the tuple $t_{c}$ in $G$ such that
$p(t_{c})=\widetilde{p}(T[c,c])$, and let 
$\mathcal{T}_G$ be the set of these tuples.
We show that \cc\ is \textit{true} iff, 
for each $G$,
the following inequality (which can be checked in polynomial time) holds:
\begin{equation}\label{eq:caso1-cc-BEGD}
\sum_{\langle c,c \rangle\in \mathcal{P}_G} \widetilde{p}(T[c,c]) \leq 1
\end{equation}

$(\Rightarrow)$:
Reasoning by contradiction, assume that, for a group $G$, 
inequality (\ref{eq:caso1-cc-BEGD}) does not hold, but there is a model for the PDB w.r.t. 
$\IC$.

The constraint entails that,
for each pair of distinct tuples $t_1,t_2\in \mathcal{T}_G$,
there is the edge $(t_1,t_2)$ in $HG(r, \IC)$.
Hence, there is a clique in $HG(r, ic)$ consisting of the tuples in $\mathcal{T}_G$.
Since the sum of the probabilities of the tuples in $\mathcal{T}_G$ is greater than $1$ 
(by contradiction hypothesis), and since \cc\ is \textit{true} only if, for each clique 
in the conflict hypergraph, the sum of the probabilities in the clique does not exceed $1$,
it follows that \cc\ is \textit{false}.

$(\Leftarrow)$:
It is straightforward to see that there is model for $\mathcal{T}_G$ w.r.t. $ic$,
since the sum of the probabilities of the tuples in $\mathcal{T}_G$ is less than or equal to $1$,
and since the tuples in $\mathcal{T}_G$ describe a clique in $HG(\mathcal{T}_G, ic)$.
Since, for each $\langle c,c \rangle\in \mathcal{P}_G$, the tuple 
$t_{c}$ in $\mathcal{T}_G$ is such that its probability is not less than
the probability of every other tuple in $T[c,c]$, it is easy to see that a model $M$ 
for $G$ w.r.t. $ic$ can be obtained by putting the tuples in $T[c,c]$ other than $t_{c}$ 
in the portion of the probability space occupied by the worlds containing $t_c$.\\[2pt]
\noindent
\textit{\underline{Case 2}}: 
$R_1\neq R_2$.
We assume that $z_1\in\vec{y_1}$ and $z_2\in\vec{y_2}$, 
that is, two distinct relation names occur in $ic$, and the variables of the inequality predicate belongs to different relation atoms.
In fact, the case that $z_1$ and $z_2$ belong to the same relation atom can be proved by
reasoning analogously.

Let $\overline{X}_1$ and $\overline{X}_2$ be, respectively, the set of attributes in 
\emph{Attr}$(R_1)$ and \emph{Attr}$(R_2)$
corresponding to the variables in $\vec{x}$, and 
let $Z_1$ and $Z_2$ be the attributes in \emph{Attr}$(R_1)$ and \emph{Attr}$(R_2)$ 
corresponding to the variables $z_1$ and $z_2$, respectively.
Let $r_1$ be the instance of $R_1$, and $r_2$ be the instance of $R_2$.

Observe that $ic$ does not impose any condition between pairs of tuples 
$t_1\in r_1$ and $t_2\in r_2$ such that 
there are attributes $X_1\in\overline{X}_1$ and $X_2\in\overline{X}_2$ 
such that $t_1[X_1]\neq t_2[X_2]$.
This entails that  \cc\ can be decided by considering the consistency of the
tuples of $r_1$ and $r_2$ sharing the same combination of values for the attributes 
corresponding to the variables in $\vec{x}$ separately from the tuples sharing
different combinations of values for the same attributes.
For each combination $\vec{v}=v_1, \dots, v_k$ of values for these attributes (i.e., 
$\forall \vec{v}\in \Pi_{\overline{X}_1}(r_1) \cap \Pi_{\overline{X}_2}(r_2)$),
let $G_1(\vec{v})$ and $G_2(\vec{v})$ be the sets of tuples of $r_1$ and $r_2$, respectively,
where the attributes corresponding to the variables in $\vec{x}$ have values
$v_1, \dots, v_k$.
Let $\mathcal{V}(G_1(\vec{v}))=\{t[Z_1] \ |\ t\in G_1(\vec{v}) \}$ and 
$\mathcal{V}(G_2(\vec{v}))=\{t[Z_2] \ |\ t\in G_2(\vec{v})\}$.
For each $c_1\in \mathcal{V}(G_1(\vec{v}))$ (resp., $c_2\in \mathcal{V}(G_2(\vec{v}))$),
let $T_1[c_1]$ (resp., $T_2[c_2]$) be the set of tuples $t$ of $G_1(\vec{v})$ 
(resp., $G_2(\vec{v})$)
such that $t[Z_1]=c_1$ (resp., $t[Z_2]=c_2$).
Moreover,
for each $c_1\in \mathcal{V}(G_1(\vec{v}))$ (resp., $c_2\in \mathcal{V}(G_2(\vec{v}))$),
let $\widetilde{p}(T_1[c_1])$ (resp., $\widetilde{p}(T_2[c_2])$) 
be the maximum probability of the tuples in $T_1[c_1]$ (resp., $T_2[c_2]$).

We show that \cc\ is \textit{true} iff, 
$\forall \vec{v}\in \Pi_{\overline{X}_1}(r_1) \cap \Pi_{\overline{X}_2}(r_2)$,
it is the case that:
\begin{equation}\label{eq:caso3-cc-BEGD}
\forall c_1\in \mathcal{V}(G_1(\vec{v}))\ \forall c_2\in \mathcal{V}(G_2(\vec{v}))
\mbox{ s.t. } c_1\neq c_2, \mbox{ it holds that } \widetilde{p}(T_1[c_1])+\widetilde{p}(T_2[c_2]) \leq 1
\end{equation}

$(\Rightarrow)$:
Reasoning by contradiction, assume that the database is consistent but there are 
$c_1\in \mathcal{V}(G_1(\vec{v}))$ and $c_2\in \mathcal{V}(G_2(\vec{v}))$, 
with $c_1\neq c_2$, such that 
$\widetilde{p}(T_1[c_1])+\widetilde{p}(T_2[c_2]) >1$.
Hence, there are tuples $t_1\in T_1[c_1]$ and $t_2\in T_2[c_2]$ such that $p(t_1)+p(t_2)>1$.
As these tuples form a conflicting set, the conflict hypergraph $HG(D^p, \IC)$ contains the edge $(t_1,t_2)$.
It follows that the condition of Theorem~\ref{theo:tractableCC}, that is a necessary condition for the 
consistency in the presence of any hypergraph (as pointed out in the core of the paper after Theorem~\ref{theo:tractableCC}),
is not satisfied, thus contradicting the hypothesis.

$(\Leftarrow)$:
It suffices to separately consider each $\vec{v}\in \Pi_{\overline{X}_1}(r_1) \cap \Pi_{\overline{X}_2}(r_2)$, and to show that the fact that (\ref{eq:caso3-cc-BEGD}) holds for this $\vec{v}$ implies 
the consistency of the tuples in $G_1(\vec{v})\cup G_2(\vec{v})$ (as explained above, 
the consistency can be checked by separately considering the various combinations in  
$\Pi_{\overline{X}_1}(r_1) \cap \Pi_{\overline{X}_2}(r_2)$).

Let  $\widetilde{t}_{1}\in G_1(\vec{v})$ and $\widetilde{t}_{2}\in G_2(\vec{v})$ be 
such that 
\begin{itemize}
\item[(i)]
$\widetilde{t}_{1}\in T_1[c_1]$ and  $\widetilde{t}_{2}\in T_2[c_2]$, with $c_1\neq c_2$; and
\item[(ii)]
among the pair of tuples satisfying the above conditions, 
$\widetilde{t}_{1}$ and $\widetilde{t}_{2}$ have maximum probability  w.r.t. the tuples in $G_1(\vec{v})$ and $G_2(\vec{v})$, respectively.
\end{itemize}

If these two tuples do not exist, it means that the set of tuples 
$G_1(\vec{v})\cup G_2(\vec{v})$ is consistent, as there are
no tuples coinciding in the values of the attributes corresponding to $\vec{x}$, but not in 
the attributes corresponding to $z_1$ and $z_2$.
It remains to be proved that, if these two tuples exist, then the tuples in 
$G_1(\vec{v})\cup G_2(\vec{v})$ are consistent w.r.t. $\IC$.
%
In fact, equation~(\ref{eq:caso3-cc-BEGD}) ensures that
$p(\tilde{t}_1)+p(\tilde{t}_2)\leq 1$, which in turn entails that a model for 
$\{\widetilde{t}_{1},\widetilde{t}_{2} \}$ w.r.t. $\IC$ exists.
Starting from this model, a model $M$ for $G_1(\vec{v})\cup G_2(\vec{v})$ w.r.t. $\IC$ 
can be obtained as follows.
The tuples in $G_1(\vec{v})$ other than $\tilde{t}_{1}$ which are conflicting with at least
one tuple $G_2(\vec{v})$ are put in the portion of the probability space occupied by the worlds containing $\widetilde{t}_{1}$.
This can be done since the fact that $\tilde{t}_1$ has maximum probability among the tuples 
in $G_1(\vec{v})$ conflicting with at least one tuple in $G_2(\vec{v})$ makes any other tuple
in $G_1(\vec{v})$ conflicting with at least one tuple in $G_2(\vec{v})$ have a probability
which fits the portion of the probability space occupied by $\tilde{t}_1$.
Similarly, the tuples in $G_2(\vec{v})$ other than $\tilde{t}_{2}$ which are conflicting with 
at least one tuple $G_1(\vec{v})$ are put in the portion of the probability space occupied by
the worlds containing $\widetilde{t}_{2}$.
Also in this case, this can be done since $\tilde{t}_2$ has maximum probability among the 
tuples in $G_2(\vec{v})$ conflicting with at least one tuple in $G_1(\vec{v})$.
Finally, any tuple in $G_1(\vec{v})$ (resp., $G_2(\vec{v})$) which is conflicting with no tuple in $G_2(\vec{v})$ (resp., $G_1(\vec{v})$) can be put in any portion of the probability 
space, since its co-occurrence with any other tuple makes no constraint violated.
\end{proof}

\ \\
\noindent{\textbf{Theorem~\ref{theo:2FDs-cc}.} }
\textit{There is an $\IC$ consisting of $2$ FDs over the same relation scheme 
such that \cc\ is $NP$-hard.}
\begin{proof}
We show a LOGSPACE reduction  
from \textsc{3-coloring} to \cc\ which yields \cc\ instances where 
$\IC$ contains only functional dependencies. 
The rationale of the proof is similar to the proof in~\cite{Pap88}
of the $NP$-hardness of PSAT.

We briefly recall the definition of \textsc{3-coloring}.
An instance of  \textsc{3-coloring} consists of a graph $G=\la N,E \ra$, where 
$N$ is a set of node identifiers and $E$ is a set of edges (pairs of node identifiers).
The answer of  a \textsc{3-coloring} instance is \emph{true} iff there is a total 
function 
$f:N\rightarrow \{$\emph{Red, Green,} \emph{Blue}$\}$ such that $f(n_i)\neq f(n_j)$ 
whenever $\{n_i,n_j\}\in E$ ($f$ is said to be a $3$-coloring function over $G$).

Let $G=\la N,E \ra$ be a \textsc{3-coloring} instance.
We construct an equivalent instance  $\la \D^p,\IC,D^p \ra$ of \cc\ as follows:
\begin{list}{--}
     {\setlength{\rightmargin}{0mm}
      \setlength{\leftmargin}{3mm}
      \setlength{\itemindent}{0mm} 
			\setlength{\itemsep}{0mm}}
\item
$\D^p$ consists of the probabilistic relation schema
$R^p($\emph{Node, Color, IdEdge, P}$)$;
\item 
$D^p$ is the instance of $\D^p$ consisting of the instance $r^p$ of $R^p$ 
defined as follows:
for each node $n\in N$, for each edge $e\in E$ such that $n\in e$,
and for each color $c\in\{$\emph{Red,Green,Blue}$\}$, 
$r^p$ contains the tuple $(n, c, e, \frac{1}{3})$.
\item
$\IC$ is the set of denial constraints over $\D^p$ consisting of the 
following two functional dependencies:
\begin{list}{--}
     {\setlength{\rightmargin}{0mm}
      \setlength{\leftmargin}{6mm}
      \setlength{\itemindent}{0mm} 
			\setlength{\itemsep}{0mm}}
\item[$ic_1:$] 
$
\neg [R(x_1,x_2,x_3)\wedge R(x_1,x_4,x_5)\wedge x_2\neq x_4]$ 
\item[$ic_2:$] 
$
\neg [R(x_1,x_2,x_3)\wedge R(x_4,x_2,x_3)\wedge x_1\neq x_4]$ 
\end{list}

\end{list}


We first show that, if $G$ is $3$-colorable, then $D^p\models\IC$.
In fact, given a $3$-coloring function $f$ over $G$,
the interpretation $Pr$ defined below is a model of $D^p$ w.r.t. $\IC$.
$Pr$ assigns non zero probability to the following three possible worlds only:

\noindent
$w_1 = \{R(n, f(n), e) \ | \ n\in N, e\in E \wedge n\in e \}$;\\
$w_2 = \{R(n,$ \emph{Next}$(f(n)), e) \ | \ n\in N, e\in E \wedge n\in e \}$;\\
$w_3 = \{R(n,$ \emph{Next}$($\emph{Next}$(f(n))), e) \ | \ n\in N, e\in E \wedge  n\in e \}$,\\
where \emph{Next} is a function which receives a color $c\in \{$\emph{Red, Green, Blue}$\}$ and returns 
the next color in the sequence $[$\emph{Red, Green, Blue}$]$ 
(where \emph{Next}$($\emph{Blue}$)$ returns \emph{Red}).
Specifically, $Pr$ assigns probability $\frac{1}{3}$ to all the three possible worlds 
$w_1,w_2,$ $w_3$.
It is easy to see that each possible world $w_1,w_2,w_3$ satisfies $\IC$ and that
every tuple in $D^p$ appears exactly in one possible world in $\{w_1,w_2, w_3\}$.
Therefore $Pr$ is a model of $D^p$.

We now show that, if $D^p\models\IC$, then $G$ is $3$-colorable.
It is easy to see that $G$ is $3$-colorable if 
there is a model $Pr$ for $D^p$ w.r.t. $\IC$ having the following property $\Pi$:
$Pr$ assigns non-zero probability only to 3-\emph{coloring possible worlds}, i.e., possible worlds containing, for each edge $e=(n_i,n_j)\in E$, two tuples 
$t^e_i=R(n_i,c_i, e)$ and $t^e_j=R(n_j,c_j, e)$, where $c_i\neq c_j$. 
In fact, starting from $Pr$ and a 3-coloring possible world $w$ with $Pr(w)>0$, a
function $f^w$ can be defined which assigns to each node $n\in N$ the color $c$ if there is 
a tuple $R(n,c, e)\in w$ ($f^w$ is a function since it is injective, as $w$ cannot 
contain tuples assigning different colors to the same node).
Clearly, $f^w$ 
is a $3$-coloring function, as 
it associates every node $n$ with a unique color
and assigns different colors to pairs of nodes connected by an edge. 
Hence, it remains to be shown that at least one model satisfying $\Pi$ exists.
In fact, we prove that any model for $D^p$ w.r.t. $\IC$ satisfies $\Pi$. 
Reasoning by contradiction, assume that, for a model $Pr$, there is a non-3-coloring possible world $w^*$ such that $Pr(w^*) = \epsilon >0$.
That is, there is at least a pair $n,e$, with $n\in N$ and $e\in E$  such that
for each $c\in\{$\emph{Red,Green,Blue}$\}$, $R(n,c,e)\not\in w^*$.
Now, consider the tuples 
$t_1= R(n,$\emph{Red}$,e)$, 
$t_2= R(n,$\emph{Green}$,e)$, 
$t_3= R(n,$\emph{Blue}$,e)$
and the sets\\ 
$S_1 = \{w\in pwd(D^p)\ | \ t_1\in w \wedge Pr(w)>0\}$,\\
$S_2 = \{w\in pwd(D^p)\ | \ t_2\in w \wedge Pr(w)>0\}$,\\
$S_3 = \{w\in pwd(D^p)\ | \ t_3\in w \wedge Pr(w)>0\}$.\\
Since $ic_1$ is satisfied by every possible world $w\in pwd(D^p)$
such that $Pr(w)>0$, this means that
for each possible world $w$ there is at most one color $c\in\{$\emph{Red,Green,Blue}$\}$ 
such that the tuple
$R(n,c,e)$ belongs to $w$.
Therefore, it must be the case that, $\forall i,j\in\{1,2,3\},i\neq j$, 
$S_i\cap S_j=\emptyset$.
Since $Pr$ is an interpretation, the following equalities must hold:
\begin{itemize}
\item
$\frac{1}{3} = p(t_1)=\sum_{w\in S_1} Pr(w)$;
\item
$\frac{1}{3} = p(t_2)=\sum_{w\in S_2} Pr(w)$;
\item
$\frac{1}{3} = p(t_3)=\sum_{w\in S_3} Pr(w)$.
\end{itemize}
This implies that 
$$\sum_{w\in S_1} Pr(w)+ \sum_{w\in S_2} Pr(w)+ \sum_{w\in S_3} Pr(w) =1$$
However, since $Pr(w^*)=\epsilon>0$ and $Pr$ is an interpretation,
$\sum_{w\in pwd(D^p)\setminus\{w^*\}} Pr(w)<1$. 
The latter, since $w^*\not\in S_i$ for each $i\in\{1,2,3\}$,
implies that $pwd(D^p)\setminus\{w^*\}\supseteq S_1\cup S_2\cup S_3$, and then
$\sum_{w\in (S_1\cup S_2\cup S_3)} Pr(w)<1$ which is a contradiction.
\end{proof}

\noindent{\textbf{Theorem~\ref{theo:tract-EGDs-and-join-free}.}}
\textit{Let each denial constraint in $\IC$ be join-free or a BEGD.
If, for each pair of distinct constraints $ic_1$,$ic_2$ in $\IC$, the relation names occurring in $ic_1$ are distinct from those in $ic_2$, then \cc\ is in \textit{PTIME}.}
\begin{proof}
Trivially follows from theorems~\ref{theo:EGDs-cc},~\ref{theo:ic-empty-phi}, and from the fact that the consistency can be checked 
by considering the maximal connected components of the conflict hypergraph separately.
\end{proof}

\noindent{\textbf{Theorem~\ref{cor:FD-cc}.}}
\textit{If $\IC$ consists of one FD per relation, then $HG(D^p,\IC)$ is a graph where
each connected component is either a singleton or a complete multipartite graph.
Moreover, $D^p$ is consistent w.r.t. $\IC$ iff the following property holds:
for each connected component $C$ of $HG(D^p,\IC)$, denoting the maximal independent sets 
of $C$ as $S_1, \dots, S_k$, it is the case that
$\sum_{i\in[1..k]} \tilde{p}_i\leq 1$, where $\tilde{p}_i=\max_{t\in S_i} p(t)$.}
\begin{proof}
It is easy to see that multiple FDs over distinct relations involve disjoint sets of tuples.
Thus, it is straightforward to see that the conflict hypergraph has the structural property described in the statement iff, for each relation, the conflict hypergraph over the set of
tuples of this relation is a graph having the same structural property.
Moreover, as observed in the proof of Theorem~\ref{theo:tract-EGDs-and-join-free}, the consistency can be checked by considering the maximal connected components of the conflict hypergraph separately.

This implies that, in order to prove the statement, it suffices to consider the case that
that $\IC$ consists of a unique FD $ic$ over a relation $R$, and $D^p$ consists of an 
instance $r$ of $R$.
In particular, we assume that $ic$ is of the form:
$$\neg [R(\vec{x},\vec{y}_1)\wedge R(\vec{x},\vec{y}_2)\wedge z_1\neq z_2 ],$$
where $z_1$ and $z_2$ are variables in $\vec{y}_1$ and $\vec{y}_2$, respectively,
corresponding to the same attribute $Z$ of $R$.
That is, we are assuming that the FD $ic$ is in canonical form (i.e., its right-hand side consists of a unique attribute).
This yields no loss of generality, as it is easy to 
see that the reasoning used in the proof is still valid in the presence of FDs whose 
right-hand sides contain more than one attribute.

The relation instance $r$ can be partitioned into the two relations $r'$, $r''$, containing the tuples connected to at least another tuple in $HG(D^p,\IC)$ (that is, tuples belonging to some conflicting set) and the isolated tuples (that is, tuples belonging to no conflicting set), respectively.
Obviously, the subgraph of $HG(D^p, \IC)$ containing only the tuples in $r''$ contains no edge, and it is such that each of its connected component is a singleton.
Therefore, in order to complete the proof of the first part of the statement, it remains to be
proved that the subgraph $G$ of $HG(D^p, \IC)$ containing only the tuples in $r'$ is such
that each of its connected component is a complete multipartite graph.

Let $\overline{X}$ be the set of attributes in \emph{Attr}$(R)$ corresponding to the 
variables in $\vec{x}$.
The form of $ic$ implies that $G$ is a graph having the following structural property
$\mathcal{S}$: 
for any pair of tuples $t_1$, $t_2$, there is the edge $(t_1,t_2)$ in $G$ iff: 
$1)$ $\forall X\in\overline{X}$, $t_1[X]=t_2[X]$,
and
$2)$ $t_1[Z]\neq t_2[Z]$.

This implies that $G$ has as many connected components as the cardinality of 
$\Pi_{\overline{X}} r'$.
Specifically, each connected component of $G$ corresponds to a tuple $\vec{v}$ in 
$\Pi_{\overline{X}} r'$, as it contains every tuple of $r'$ whose projection over 
$\overline{X}$ coincides with $\vec{v}$.
In fact, property $\mathcal{S}$ implies that:
\begin{enumerate}
\item[$A.$]
there is no path in $G$ between tuples differing in at least one attribute in 
$\overline{X}$;
\item[$B.$]
any two tuples $t'$, $t''$ coinciding in all the attributes in $\overline{X}$ are either
directly connected to one another (in the case that they do not coincide in attribute $Z$),
or there is a third tuple $t'''$ to which they are both connected.
In fact, $t'$ and $t''$ are not isolated (otherwise they would not belong to $r'$), and any 
tuple conflicting with $t'$ is also conflicting with $t''$, as we are in the case that 
$t'$ and $t''$ coincide in $Z$.
\end{enumerate}

To complete the proof of the first part of the statement, we now show that, taken any connected component $C$ of $G$, $C$ is a complete multipartite graph.
This straightforwardly follows from the following facts:
\begin{enumerate}
\item[$a.$]
the nodes of $C$ can be partitioned into the maximal independent sets $S_1, \dots, S_k$, 
where $k$ is the number of distinct values of attribute $Z$ occurring in the tuples in $C$.
In particular, each $S_i$ corresponds to one of these values $v$ of $Z$, and contains all 
the tuples of $C$ having $v$ as value of attribute $Z$.
The fact that every $S_i$ is a maximal independent set trivially follows from property 
$\mathcal{S}$.
\item[$b.$]
for every pair of tuples $t_i$ and $t_j$ belonging to $S_i$ and $S_j$ (with $i,j\in [1..k]$ 
and $i\neq j$), there is an edge connecting $t_i$ to $t_j$ (this also trivially follows 
from property $\mathcal{S}$).
\end{enumerate}

We now prove the second part of the statement.

$(\Rightarrow)$:
Reasoning by contradiction, assume that $D^p$ is consistent w.r.t. $\IC$ but, for some connected component $C$ of $HG(D^p,\IC)$, it does not hold that 
$\sum_{i\in[1..k]} \tilde{p}_i\leq 1$, where $\tilde{p}_i=\max_{t\in S_i} p(t)$ and $S_1, \dots, S_k$ are the maximal independent sets of $C$.
Obviously, $C$ can not be a singleton (otherwise the inequality would hold), thus it
must be the case that $C$ is a complete multipartite graph. 

For each $i\in[1..k]$, let $\tilde{t}_i$ be a tuple of $S_i$ such that 
$p(\tilde{t}_i)=\tilde{p}_i$.
Since $C$ is a complete multipartite graph, and since the so obtained tuples 
$\tilde{t}_1, \dots, \tilde{t}_k$ belong to distinct independent sets, it must be the case
that, for each $i,j\in[1..k]$ with $i\neq j$, there is an edge in $C$ between $\tilde{t}_i$ 
and $\tilde{t}_j$.
This means that, in every model $M$ for $D^p$ w.r.t. $\IC$,
for each $i,j\in[1..k]$ with $i\neq j$,
the tuples $\tilde{t}_i, \tilde{t}_j$ can not co-exist in a non-zero probability
possible world.
That is, every non-zero probability possible world contains at most one tuple
among those in $\{\tilde{t}_1, \dots, \tilde{t}_k\}$.
This entails that the sum of the probabilities of the possible worlds containing 
the tuples in $\{\tilde{t}_1, \dots, \tilde{t}_k\}$ is equal to 
the sum of the marginal probabilities of the tuples in $\{\tilde{t}_1, \dots, \tilde{t}_k\}$, which, by contradiction hypothesis, is greater than $1$.
This contradicts the fact that $M$ is a model.


$(\Leftarrow)$:
We now show that $D^p$ is consistent w.r.t. $\IC$ if the inequality $\sum_{i\in[1..k]} \tilde{p}_i\leq 1$ holds, 
where $\tilde{p}_i=\max_{t\in S_i} p(t)$ and $S_1, \dots, S_k$ are the maximal independent sets of $C$.
Consider the database instance $\tilde{D}^p$ consisting of the tuples 
$\tilde{t}_1, \dots, \tilde{t}_k$ where $\tilde{t}_i$ (with  $i\in[1..k]$) is a tuple of $S_i$ such that 
$p(\tilde{t}_i)=\tilde{p}_i$.
It is easy to see that there is a model for $\tilde{D}^p$ w.r.t. $IC$:
since $C$ is a complete multipartite graph, and $\tilde{t}_1, \dots, \tilde{t}_k$ 
belong to distinct independent sets of $C$, it follows that, 
for each $i,j\in[1..k]$ with $i\neq j$, there is exactly one edge in $C$ between $\tilde{t}_i$ 
and $\tilde{t}_j$.
That is, the conflict graph of $\tilde{D}^p$ w.r.t. $\IC$ is a clique.
Hence, the fact that inequality $\sum_{i\in[1..k]} \tilde{p}_i\leq 1$ holds is sufficient
to ensure the existence of a model $\tilde{M}$ for $\tilde{D}^p$ w.r.t. $\IC$.
Starting from $\tilde{M}$, a model for $D^p$ w.r.t. $\IC$ can be obtained by reasoning as follows.
Since, for each maximal independent set $S_i$ of $C$ (with  $i\in[1..k]$),
the tuples in $S_i$ other than $\tilde{t}_i$ are such that their probability is less than or equal to 
$p(\tilde{t}_i)$, a model $M$ for $D^p$ w.r.t. $\IC$ can be obtained by putting the tuples in 
$S_i$ other than $\tilde{t}_i$ in the portion of the probability space corresponding to 
that occupied by the worlds containing $\tilde{t}_i$ according the model $\tilde{M}$.
The fact that $M$ is a model follows from the fact that, for each $i\in[1..k]$, the tuples in 
$S_i$ other than $\tilde{t}_i$ are conflicting only with the same tuples which are conflicting
with $\tilde{t}_i$.
\end{proof}

\subsection{Proofs of Lemma~\ref{boundedpmin-pmax} and Theorem~\ref{theo:mphardness2}}

\ \\
\noindent{\textbf{Lemma~\ref{boundedpmin-pmax}.} }
\textit{Let $Q$ be a conjunctive query over $\D^p$, $D^p$ an instance of $\D^p$, and 
$\vec{t}$ an answer  of $Q$ having minimum probability $p^{min}$ and maximum probability 
$p^{max}$.
Let $m$ be the number of tuples in $D^p$ plus $3$ and $a$ be the maximum among the 
numerators and denominators of the probabilities of the tuples in $D^p$.
Then 
$p^{min}$ and $p^{max}$ are expressible as fractions of the form $\frac{\eta}{\delta}$,
with $0\leq \eta \leq (m a)^m$ and $0< \delta \leq (m a)^m$.
}
\begin{proof}
Consider the equivalent form of the
linear programming problem $LP(S^*)$ described in the proof of
Proposition~\ref{pro:intquery}, where equalities 
$(e1)$ of $S^*(\D^p,\IC,D^p)$ are rewritten as:\\
$\forall t\in D^p,\  d(p(t)) \times \sum_{i | w_i\in pwd(D^p)\wedge t\in w_i} v_i= d(p(t)) \times p(t)$,\\
where $p(t) = \frac{n(p(t))}{d(p(t))}$ (i.e., $n(p(t))$ and $d(p(t))$ are the numerator and denominator of $p(t)$, respectively).
This way, we have that all the coefficients of $S^*(\D^p,\IC,D^p)$ are integers, where each coefficient can be either $0$, or $1$,  or the numerator or the denominator of the marginal probability of a tuple of $D^p$.

In~\cite{PapOpt}, it was shown that the solution of any instance of the
linear programming problem with integer coefficients is expressible as a fraction of the form 
$\frac{\eta}{\delta}$,
where both $\eta$ and $\delta$ are naturals bounded by $(m a)^m$, 
where $m$ is the number of (in)equalities and 
$a$ the greatest integer coefficient occurring in the instance. 
By applying this result to $LP(S^*)$, we get the statement:
in fact, it is easy to see that 
$i)$ $S^*(\D^p,\IC,D^p)$ contains integer coefficients only, 
$ii)$ the number $m$ of (in)equalities in $S^*(\D^p,\IC,D^p)$ is equal to 
the number of tuples in $D^p$ plus $3$, and 
$iii)$ the greatest integer constant $a$ in $S^*(\D^p,\IC,D^p)$ is the maximum among the numerators and denominators of the probabilities of the tuples in $D^p$.
\end{proof}

\ \\
\noindent{\textbf{Theorem~\ref{theo:mphardness2}.}}
(Lower bound of \mp)
\textit{There is at least one conjunctive query without projection for which \mp\ is co$NP$-hard, even if $\IC$ consists of binary constraints only.}
\begin{proof}
We show a reduction from the planar \textsc{3-coloring} problem to the complement 
of the membership problem ($\nmp$).
An instance of planar \textsc{3-coloring} consists of a planar graph $G=\la N,E \ra$, 
where $N$ is a set of node identifiers and $E$ is a set of edges (pairs of node identifiers).
The answer of a planar \textsc{3-coloring} instance $G$ is true iff 
there is a $3$-coloring function over $G$, i.e., a total function  $f:N\rightarrow \{R,G,B\}$ 
such that $f(n_i)\neq f(n_j)$ whenever $\{n_i,n_j\}\in E$.
Observe that every planar graph $G=\la N,E \ra$ is $4$-colorable, that is,
there exists a function $f:N\rightarrow \{R,G,B,C\}$ 
such that $f(n_i)\neq f(n_j)$ whenever $\{n_i,n_j\}\in E$ (in this case, 
$f$ is said to be a $4$-coloring function).

Let $G=\la N,E \ra$ be a planar \textsc{3-coloring} instance. 
We construct an equivalent $\nmp$ instance $\la \D^p,\IC,D^p, Q, t, k_1, k_2 \ra$ as follows:
\begin{list}{--}
     {\setlength{\rightmargin}{0mm}
      \setlength{\leftmargin}{3mm}
      \setlength{\itemindent}{-1mm} 
			\setlength{\itemsep}{0mm}}
\item
$\D^p$ consists of the probabilistic relation schemas
$R^p_G($\emph{Node, Color, IdEdge, P}$)$ and $R^p_\phi($\emph{Tid,P}$)$.
\item
$D^p$ is the instance of $\D^p$ consisting of the instances $r^p_G$ of $R^p_G$ and 
$r^p_\phi$ of $R^p_\phi$ defined as follows:
\begin{list}{--}
     {\setlength{\rightmargin}{0mm}
      \setlength{\leftmargin}{5mm}
      \setlength{\itemindent}{-1mm} 
			\setlength{\itemsep}{0mm}}
\item
for each node $n\in N$ and for each edge $e\in E$ such that $n\in e$, 
$r^p_G$ contains four tuples of the form $R^p_G(n, c, e, \frac{1}{8})$, 
one for each $c\in\{R,G,B,C\}$;
\item
$r^p_\phi$ consists of the tuples $R^p_\phi(1,\frac{1}{2})$ and $R^p_\phi(2,\frac{1}{2})$ only;
\end{list}
\item
$\IC$ contains the following binary denial constraints:
\begin{list}{--}
     {\setlength{\rightmargin}{0mm}
      \setlength{\leftmargin}{7mm}
      \setlength{\itemindent}{-1mm} 
			\setlength{\itemsep}{0mm}}
\item[$ic_1:$] 
$
\neg [R_G(x_1,x_2,x_3)\wedge R_G(x_1,x_4,x_5)\wedge x_2\neq x_4]$; 
\item[$ic_2:$] 
$
\neg [R_G(x_1,x_2,x_3)\wedge R_G(x_4,x_2,x_3)\wedge x_1\neq x_4]$;
\item[$ic_3:$]
$
\neg [R_G(x_1,x_2,x_3)\wedge R_\phi(2)]$;
\item[$ic_4:$]
$
\neg [R_G(x_1,x_2,C)\wedge R_\phi(1)]$;
\end{list}
\item
$Q(x,y) = R_\phi(x)\wedge R_\phi(y)$;
\item
$t = (1,2)$; 
\item
$k_1=\frac{1}{2}$;
\item
$k_2=1$.
\end{list}

It is easy to see that the fact that $G$ is $4$-colorable implies that $D^p$ is consistent
w.r.t. $\IC$ (it suffices to follow the same reasoning as the proof of Theorem~\ref{theo:CCisNPcomplete}, using $4$ colors instead of $3$).

We first prove that, if $G$ is 3-colorable, then the corresponding instance of $\nmp$ is true.
Let $f$ be a $3$-coloring function over $G$.
Consider an interpretation $Pr$ for $D^p$ which assigns non-zero probability to the following
possible worlds only:\\
\noindent
$w_1 = \{R_G(n, f(n), e) \ | \ n\in N, e\in E \wedge n\in e \}\cup\{R_\phi(1)\}$\\
$w_2 = \{R_G(n,$ \emph{Next}$(f(n)), e) \ | \ n\in N, e\in E \wedge n\in e \}$\\
$w_3 = \{R_G(n,$ \emph{Next}$($\emph{Next}$(f(n))), e) \ | \ n\in N, e\in E \wedge  n\in e \}$\\
$w_4 = \{R_G(n,$ \emph{Next}$($\emph{Next}$($\emph{Next}$(f(n)))), e) \ | \ n\!\in\!N, e\!\in\!E \wedge  n\!\in\!e \}$\\
$w_5 = \{R_\phi(1), R_\phi(2)\}$\\
$w_6 = \{R_\phi(2)\}$\\
where \emph{Next} is a function which receives a color $c\in \{R,G,B,C\}$ and returns 
the next color in the sequence $[R, G, B, C]$ (where \emph{Next}$(C)$ returns $R$).
Furthermore, $Pr$ assigns probability $\frac{1}{8}$ to the possible worlds 
$w_1, w_2, w_3, w_4$ and $w_5$, and probability $\frac{3}{8}$ to the possible world $w_6$.
It is easy to see that $Pr$ is a model of $D^p$ w.r.t. $\IC$
and the probability that the tuple $t=(1,2)$ is an answer of $Q$ assigned by $Pr$ is
$\frac{1}{8}$. 
Hence, the $\nmp$ is true in this case (as $\frac{1}{8}< k_1$).


We now prove that if $G$ is not $3$-colorable, then the corresponding instance of $\nmp$ is
false.
First observe that,
reasoning similarly to in the proof of Theorem~\ref{theo:CCisNPcomplete},
it is possible to show that, for each model $Pr$ of
$D^p$ w.r.t. $\IC$ and for each possible world $w$ such that $Pr(w)>0$,
if $w$ contains at least a tuple of $r_G$, 
then for each node $n\in N$ and for each edge $e\in E$ such that 
$n\in e$, there exists $c\in \{R,G,B,C\}$ such that $w$ contains the tuple $R_G(n, c, e)$.
This is due to the fact that every possible world $w$ such that
$Pr(w)>0$ can not contain two tuples $R_G(n, c', e), R_G(n, c'', e)$
and no tuple in $r_G$ can belong to a possible world which contains the tuple $R_\phi(2)$ too.

Since $G$ is not $3$-colorable, for each model $Pr$ of $D^p$ w.r.t. $\IC$ and
for each possible world $w$ such that $Pr(w)>0$ containing at least a tuple of $r_G$, 
it holds that $w$ contains a tuple $R_G(n, C, e)$. 
This implies that no possible world containing a tuple of $r_G$
can contain the tuple $R_\phi(1)$, as otherwise $ic_4$ would be violated.
Since $ic_1$ and $ic_3$ hold for $Pr$, then the sum of the probability of the possible worlds containing 
at least a tuple of $r_G$ is equal to $\frac{1}{2}$.
Since the possible worlds containing at least a tuple of $r_G$
cannot contain neither $R_\phi(1)$ nor $R_\phi(2)$ (as $ic_4$ holds) and both
$R_\phi(1)$ nor $R_\phi(2)$ has probability $\frac{1}{2}$ it holds that
the probability that both $R_\phi(1)$ and $R_\phi(2)$ are true is $\frac{1}{2}$.
The latter implies that the minimum probability  that $t =(1,2)$ is an answer of $Q$ is 
$\frac{1}{2}$, which is equal to $k_1$. 
Therefore the $\nmp$ is false if $G$ is not 3-colorable.
\end{proof}


\subsection{Proof of Theorem~\ref{theo:fp-at-np}}

\ \\
\noindent{\textbf{Theorem~\ref{theo:fp-at-np}.}}
(\qa\ complexity)
\textit{\qa\ belongs to $F\!P^{N\!P}$ and is $F\!P^{N\!P [\log n]}$-hard.}
\begin{proof}
The membership in $F\!P^{NP}$ follows from \cite{Luk01},
where it was shown that a problem more general than ours (that is, the entailment problem 
for probabilistic logic programs with conditional rules)
belongs to $F\!P^{NP}$ (see \emph{Related Work}).
We prove the hardness for $F\!P^{NP [\log n]}$ by showing a reduction to \qa\ 
from the well-known $F\!P^{NP [\log n]}$-hard problem \textsc{clique size},
that is the problem of determining the size $K^*$ of the largest
clique of a given graph.

Let the graph $G=\la N,E \ra$ be an instance of \textsc{clique size}, 
where $u_1,\dots, u_n$ are the nodes of $G$ (where $n=|N|$).
We construct an equivalent instance $\la \D^p, \IC, D^p, Q \ra$ of \qa\ as follows.
$D^p$ is the database schema consisting of the following relation schemas:
\emph{Node}$^p($\emph{Id}, \emph{P}$)$,
\emph{NoEdge}$^p($\emph{nodeId}$_1$, \emph{nodeId}$_2$, \emph{P}$)$,
\emph{Flag}$^p($\emph{Id}, \emph{P}$)$. 
The database instance $D^p$ consists of the following relation instances.
Relation \emph{node}$^p$ contains a tuple $t_i=$ \emph{Node}$^p(u_i,\nicefrac{1}{n})$
for each node $u_i$ of $G$ (that is, every node of $G$ corresponds to a tuple of
\emph{node}$^p$ having probability $\nicefrac{1}{n}$).
Relation \emph{noEdge}$^p$ contains a tuple \emph{NoEdge}$^p(u_i,u_j,1)$ 
for each pair of distinct nodes of $G$ which are not connected by means of any edge in $E$
(thus, \emph{noEdge}$^p$ represents the complement of $E$, and all of
its tuples have probability $1$).
Finally, relation \emph{flag}$^p$ contains the unique tuple \emph{Flag}$^p(1,\frac{n-1}{n})$.

Let $\IC$ consist of the following denial constraints over $\D^p$:
\begin{itemize}
\item[$ic_1:$] 
$\neg [$\emph{Node}$(x_1)\ \wedge$ \emph{Node}$(x_2)\ \wedge\ $\emph{NoEdge}$(x_1,x_2)]$ 
\item[$ic_2:$] 
$\neg [$\emph{Node}$(x_1)\ \wedge$ \emph{Node}$(x_2)\ \wedge\ $\emph{Flag}$(1)\ \wedge\ x_1\neq x_2]$ 
\end{itemize}

Basically, constraint $ic_1$ forbids that tuples representing distinct nodes 
co-exist if they are not connected by any edge, while $ic_2$ imposes that tuple 
\emph{Flag}$(1)$ can co-exist with at most one tuple representing a node.

To complete the definition of the instance of $\qa$, we define the (boolean) query
$Q()=$\emph{Flag}$(1)\wedge$\emph{Node}$(x)$.

We will show that the size of the largest clique of $G$ is $K^*$ iff the empty tuple 
$t_{\emptyset}$ is an answer of $Q$ over $D^p$ with minimum probability $l^*=\frac{n-K^*}{n}$ 
(i.e., \emph{Ans}$(Q, D^p, \IC)$ consists of the pair 
$\la t_{\emptyset}, [p^{min}, p^{max}]\ra$, with  $p^{min}=\frac{n-K^*}{n}$).

We first show that if $G$ contains a clique of size $K$, then 
$p^{min}\leq \frac{n-K}{n}$.
In fact, if $K$ is the size of a clique $C$ of $G$, then we can construct the following model 
$M$ for $D^p$ w.r.t. $\IC$.
Let $w^c=\{$\emph{Node}$(u_i)| u_i\in C\} \cup$ \emph{noEdge}, 
$w^f=\{$ \emph{Flag}$(1) \} \cup$ \emph{noEdge}, and,
for each $u_i\in N\setminus C$, $w_i= \{$\emph{Node}$(u_i),$ \emph{Flag}$(1) \} \cup$ \emph{noEdge}.
Then, denoting as $w$ the generic possible world, the model $M$ is defined as follows:
$$
M(w)=\left\{
\begin{array}{ll}
\nicefrac{1}{n} & \mbox{if $w=w^c$;}\\
\nicefrac{1}{n} & \mbox{if $w=w_i$, for $i$ s.t. $u_i\in N\setminus C$;}\\
\nicefrac{(K-1)}{n} &  \mbox{if $w=w^f$;}\\
0               & \mbox{otherwise}
\end{array}
\right.$$

It is easy to see that $M$ is a model.
First of all, it assigns non-zero probability only to possible worlds
satisfying the constraints.
Moreover, for any tuple $t$ in $D^p$, summing the probabilities
of the possible worlds containing $t$ results in $t[P]$.
In fact, considering only the possible worlds which have been assigned a non-zero probability
by $M$, every tuple \emph{Node}$(u_i)$ representing a node $u_i\in C$ belongs only to 
$w^c$, which is assigned by $M$ the probability $\nicefrac{1}{n}$ (the same as $p($\emph{Node}$(u_i))$). 
Analogously, every tuple \emph{Node}$(u_i)$ representing a node $u_i\not\in C$ belongs only to 
$w^i$, which is assigned by $M$ the probability $\nicefrac{1}{n}$ (the same as $p($\emph{Node}$(u_i))$). 
Finally, tuple \emph{Flag}$(1)$ occurs only in $w^f$ and in $n-K$ possible worlds of the form
$w_i$, thus the sum of the probabilities of the possible worlds containing \emph{Flag}$(1)$
is $M(w^f)+ (n-K)\cdot \frac{1}{n}=\frac{(n-1)}{n}=p($\emph{Flag}$(1))$. 

It is easy to see that the probability of the answer $t_{\emptyset}$ of $Q$ over the model $M$
is the sum of the probabilities of the possible worlds of the form $w_i$, that is 
$\frac{(n-K)}{n}$.
Hence, from definition of minimum probability, it holds that $p^{min}\leq\frac{(n-K)}{n}$.

To complete the proof, it suffices to show that the following property $\mathcal{P}$ holds
over any model $M'$ for $\D^p$ w.r.t. $\IC$:
``\emph{the probability $l$ of the answer $t_{\emptyset}$ of $Q$ over $M'$ can not be strictly 
less than $l^*=\frac{(n-K^*)}{n}$}''.
Observe that, for every model $M'$, 
the possible worlds which have been assigned a non-zero probability by 
$M'$ can be of three types (we do not consider \emph{noEdge} tuples, as they have probability
$1$, thus they belong to every non-zero-probability possible world):
\begin{list}{--}
     {\setlength{\rightmargin}{0mm}
      \setlength{\leftmargin}{7.5mm}
      \setlength{\itemindent}{-1mm} 
			\setlength{\itemsep}{0mm}}
\item[\emph{Type 1:}]
world not containing \emph{Flag}$(1)$, and containing a non-empty set of tuples representing
the nodes of a clique (the non-emptiness of this set derives from the combination of 
constraint $ic_2$ with the value of the marginal probability assigned to tuple \emph{Flag}$(1)$);
\item[\emph{Type 2:}]
world containing the tuple \emph{Flag}$(1)$ and exactly one \emph{node} tuple;
\item[\emph{Type 3:}]
world containing the tuple \emph{Flag}$(1)$ only.
\end{list}
We will show that property $\mathcal{P}$ holds over any model $M'$ by reasoning 
inductively on the number $x$ of possible worlds of Type 1 which have been assigned a 
non-zero probability by $M'$.

The base case is $x=1$, meaning that $M'$ assigns probability $\nicefrac{1}{n}$ to a unique 
Type-1 world $w^{T1}_1$, and probability $0$ to all the other possible worlds of the same type.
It is easy to see that the sum of the probabilities assigned by $M'$ to the Type-2 worlds 
(which coincides with $l$) is equal to $\frac{1}{n}\cdot (n-|C^{T1}_1|)$,  where 
$C^{T1}_1$ is the clique represented by $w^{T1}_1$. 
Hence, if it were $l<l^*$, it would hold that 
$\frac{1}{n}\cdot (n-|C^{T1}_1|)< \frac{(n-K^*)}{n}$, 
which means that $|C^{T1}_1|>K^*$, thus contradicting that $K^*$ is the size of the maximum clique of $G$.

We now prove the induction step. 
The induction hypothesis is that $\mathcal{P}$ holds over any model assigning non-zero 
probability to exactly $x-1$ Type-1 possible worlds (with $x-1\geq 1$). 
We show that this implies that $\mathcal{P}$ holds also over any model assigning non-zero 
probability to exactly $x$ Type-1 possible worlds. 
Consider a model $M'$ assigning non-zero probability to exactly $x$ Type-1 possible worlds,
namely $w^{T1}_1,\dots, w^{T1}_x$. 
We assume that these worlds are ordered by their cardinality (in descending order),
and denote as $C_i$ the clique represented by $w^{T1}_i$ (with $i\in[1..x]$).
We also denote as $w^{T2}_1,\dots, w^{T2}_n$ the Type-2 possible worlds (where $w^{T2}_i$ contains the 
\emph{node} tuple representing $u_i$).
Moreover, let $l'$ be the probability of the answer $t_{\emptyset}$ of $Q$ over $M'$.
We show that, starting from $M'$, a new model $M''$ for $D^p$ w.r.t. $\IC$ can be constructed 
such that:\\
$i)$ $M''$ assigns non-zero probability to $x-1$ Type-1 possible worlds;\\
$ii)$ the probability $l''$ of the answer \emph{true} of $Q$ over $M''$ satisfies 
$l''\leq l'$.\\
Specifically, $M''$ is defined as follows.
$M''$ coincides with $M'$ on all the Type-1 worlds except for the probabilities assigned to  
$w^{T1}_1$ and $w^{T1}_x$.
In particular, $M''(w^{T1}_1)=M'(w^{T1}_1)+ M'(w^{T1}_x)$, while $M''(w^{T1}_x)=0$.
Moreover, for each Type-2 world $w^{T2}_i$ such that $u_i\in C_1\setminus C_x$, 
$M''(w^{T2}_i)=M'(w^{T2}_i)- M'(w^{T1}_x)$,
and,
for each Type-2 world $w^{T2}_i$ such that $u_i\in C_x\setminus C_1$,
$M''(w^{T2})=M'(w^{T2})+ M'(w^{T1}_x)$.
On the remaining Type-2 worlds, $M''$ is set equal to $M'$.
Finally, denoting the type-3 world as $w^{T3}$, 
$M''(w^{T3})=M'(w^{T3})-|C_x\setminus C_1|\cdot M'(C_x) + |C_1\setminus C_x|\cdot M'(C_x)$.
In brief, $M''$ is obtained from $M'$ by moving the probability assigned to $w^{T1}_x$ to
$w^{T1}_1$, and re-assigning the probabilities of the Type-2 and Type-3 worlds accordingly.
Hence, it is easy to see that $M''$ is still a model (as it can be easily checked that it
makes the sum of the probabilities of the possible worlds containing a tuple equal to 
the marginal probability of the tuple).
Moreover, property $i)$ holds, as $M''$ assigns probability $0$ to the world $w^{T1}_x$, 
while the other worlds of the form $w^{T1}_i$ (with $i<x$) are still assigned by $M''$ 
a positive probability, and the remaining Type-1 worlds are still assigned probability $0$.
Also property $ii)$ holds, since the probability of \emph{true} as answer of $Q$ over $M''$ 
is given by
$l''= l'+|C_x\setminus C_1|\cdot M'(C_x) - |C_1\setminus C_x|\cdot M'(C_x)$.
Since $|C_1|\geq|C_x|$, and thus $|C_1\setminus C_x|\geq|C_x\setminus C_1|$, $l''$ is less 
than or equal to $l'$.
If it were $l'< l^*$ (and thus $l''<l^*$) $M''$ would be a model assigning non-zero probability
to $x-1$ Type-1 possible worlds such that the answer \emph{true} of $Q$ over $M''$ has probability
strictly less than $l^*$, thus contradicting the induction hypothesis.
\end{proof}

\subsection{Proof of Theorem~\ref{theo:qatractable}}
\label{app:tractabilityqa}
The proof of Theorem~\ref{theo:qatractable} is postponed to the end of this section,
after introducing some preliminary lemmas.

\begin{lemma}
\label{lem:pmin-pmax-chain}
Let $D^p$ be a PDB instance of $\D^p$ such that
$HG(D^p,$ $\IC)$ is a graph and $D^p\models \IC$.
Let $t,t'$ be two tuples connected by exactly one path in $HG(D^p, \IC)$. 
Then, $p^{min}(t\wedge t')$ and $p^{max}(t\wedge t')$
can be computed in polynomial time w.r.t. the size of $D^p$.
\end{lemma}
\begin{proof}
Let $\pi$ be the path connecting $t$ and $t'$ in $HG(D^p, \IC)$. 
It is easy to see that the fact that $\pi$ is unique implies that 
$p^{min}(t\wedge t')= p^{min}_\pi(t\wedge t')$
and 
$p^{max}(t\wedge t')= p^{max}_\pi(t\wedge t')$
(in fact, any model for $D^p$ w.r.t. 
$HG(D^p, \IC)$ can  be obtained by refining a model for $D^p$ w.r.t. 
$\pi$ without changing the probabilities assigned to the event $t\wedge t'$,
following a reasoning analogous to that used in the proof of the right-to-left implication
of Theorem~\ref{theo:tractableCC}).

Since the path $\pi$ connecting $t$ and $t'$ in the graph $HG(D^p, \IC)$ is unique, it does not 
contain cycles (otherwise there would be at least two paths between $t$ and $t'$).
Hence, $\pi$ is a chain in a graph (the definition of chain for hypergraph is introduced in
Section~\ref{sec:AppendixRing}).
Therefore, $p^{min}_\pi(t\wedge t')$ can be determined by exploiting Lemma~\ref{lem:pmin-chain-hyper},
which provides the formula for computing the minimum probability that the ears at the endpoints 
of a chain co-exist.
It is trivial to see that, denoting as $\hat{t}$ and $\hat{t}'$ the tuples connected
to $t$ and $t'$ in $\pi$, in our case 
the formula in Lemma~\ref{lem:pmin-chain-hyper}
becomes:
{\small
$$
p^{min}_\pi(t\wedge t')\!=\!\!
\left\{
\begin{array}{l}
\!\!0,  \mbox{ if } (t,t') \mbox{ is an edge of }\pi \\[2pt]
\!\!\max \{0, p(t)\!+\!p(t')\!-\![1\!-\!p^{min}_\pi(\hat{t}\!\wedge\!\hat{t}')]\}, \mbox{ otherwise.}
\end{array}
\right.
$$
}
\noindent
since $\pi$ is a chain in a graph, thus its intermediate edges are hyperedges of 
cardinality $2$ with no ears.

As regards $p^{max}_\pi(t\wedge t')$, it can be evaluated as follows:

\noindent
$p^{max}_\pi(t\wedge t')\!=\!\!
\left\{
\begin{array}{l}
\!\!0,  \mbox{ if } (t,t') \mbox{ is an edge of }\pi \\[2pt]
\!\!\min \{p(t), p(t'), 1\!-\![p(\hat{t})\!+\!p(\hat{t}')\!-\!p^{max}_\pi(\hat{t}\!\wedge\!\hat{t}')]\},\\ 
\mbox{ otherwise.}
\end{array}
\right.
$

In fact, it is easy to see that the maximum probability of the event $t\wedge t'$ is
$\min \{p(t), p(t'), p^{max}_\pi(\neg \hat{t}\!\wedge\!\neg \hat{t}')\}$,
where  $p^{max}_\pi(\neg \hat{t}\!\wedge\!\neg \hat{t}')$ is the maximum probability that 
both the tuples $\hat{t}$ and $\hat{t}'$ (which are mutually exclusive with $t$ and $t'$, respectively) are false.
In turn, $p^{max}_\pi(\neg \hat{t} \wedge\!\neg \hat{t}')= 
1-p^{min}_\pi(\hat{t} \vee \hat{t}')=
1-[p(\hat{t})\!+\!p(\hat{t}')\!-\!p^{max}_\pi(\hat{t} \wedge \hat{t}')]$,
thus proving the above-reported formula.

We complete the proof by observing that  
$p^{min}(t\wedge t')$ and $p^{max}(t\wedge t')$
can be computed in polynomial time w.r.t. the size of $D^p$ by recursively applying
the above-reported formulas for $p^{min}$ and $p^{max}$ starting from $t$ and $t'$, and going further on towards the center of the unique path connecting $t$ and $t'$.
\end{proof}

\begin{lemma}
\label{lem:qatractable-clique}
For projection-free queries, \qa\ is in $PTIME$ if $HG(D^p, \IC)$ is a clique.
\end{lemma}
\begin{proof}
It straightforwardly follows from the fact that, for each pair of
tuples $t,t'$ in $HG(D^p, \IC)$, it holds that $p^{min}(t\wedge t') = p^{max}(t\wedge t') = 0$.
\end{proof}

\begin{lemma}
\label{lem:qatractable-tree}
For projection-free queries, \qa\ is in $PTIME$ if $HG(D^p, \IC)$ is a tree.
\end{lemma}
\begin{proof}
\emph{Ans}$(Q,D^p,\IC)$ can be determined by first evaluating the answer $r_q$ of $Q$ w.r.t. 
$det(D^p)$, and then computing, for each $\vec{t}\in r_q$, the minimum and maximum 
probabilities $p^{min}$ and $p^{max}$ of $\vec{t}$ as answer of $Q$.
Obviously, $r_q$ can be evaluated in polynomial time w.r.t. the size of $D^p$, and 
the number of tuples in $r_q$ is polynomially bounded by the size of $D^p$.

Observe that, every ground tuple $\vec{t}\in r_q$ derives from the 
conjunction of a set of tuples $\{t_1, \dots, t_n\}$ in $det(D^p)$.
Thus, in order to prove the statement, it suffices to prove that, for each set
$\{t_1, \dots, t_n\}$ of tuples in $det(D^p)$, 
computing $p^{min}(t_1\wedge\dots\wedge t_n)$ 
and $p^{max}(t_1\wedge\dots\wedge t_n)$ is feasible in 
polynomial time w.r.t. the size of $D^p$.

For the sake of clarity of presentation, we assume that $HG(D^p, \IC)$ coincides with its own
minimal spanning tree containing all the tuples in $\{t_1, \dots, t_n\}$.
This means that each $t_i$ (with $i\in[1..n]$) is either a leaf node or occurs as
intermediate node in the path connecting two other tuples in $\{t_1, \dots, t_n\}$, and
all the leaf nodes are in $\{t_1, \dots, t_n\}$.
In fact, if this were not the case, it is straightforward to see that nothing would change in evaluating $p^{min}(t_1\wedge\dots\wedge t_n)$ and $p^{max}(t_1\wedge\dots\wedge t_n)$ 
if we disregarded the nodes of $HG(D^p, \IC)$ which are not in $\{t_1, \dots, t_n\}$ and do not belong to any path connecting some pair of nodes in $\{t_1, \dots, t_n\}$.




Before showing how $p^{min}(t_1\wedge\dots\wedge t_n)$ and $p^{max}(t_1\wedge\dots\wedge t_n)$
can be computed, we introduce some notations.
We say that a tuple $t$ is a \textit{branching node} of
$HG(D^p, \IC)$ iff the degree of $t$ is greater than two.
Moreover, a pair of tuples $(t_i,t_j)$ is said to be an
\textit{elementary pair of tuples} of $HG(D^p, \IC)$ if
$(i)$ each of $t_i$ and $t_j$ is either in $\{t_1, \dots, t_n \}$ or a branching node, and 
$(ii)$ the path connecting $t_i$ to $t_j$ contains neither branching nodes nor tuples in 
$\{t_1, \dots, t_n\}$ as intermediate nodes.

The set of the elementary pairs of tuples is denoted as
$EP_{HG(D^p, \IC)}$ (we also use the short notation $EP$, when
$HG(D^p, \IC)$ is understood).
Moreover, we denote the branching nodes of $HG(D^p, \IC)$ which are not in 
$\{t_1, \dots, t_n\}$ as $t_{n+1}, \cdots, t_{n+m}$.
Observe that $m<n$, as $n$ is also greater than or equal to the number of leaves
of $HG(D^p, \IC)$.
Finally, we denote with $B=\{$\emph{true}, \emph{false}$\}$ the boolean domain, 
with $B^{n+m}$ the set of all the tuples of $n+m$ boolean values, 
and use the symbol $\alpha$ for tuples of $n+m$ boolean values 
and the notation $\alpha[i]$ to indicate the value of the $i$-th 
attribute of $\alpha$.


We will show that $p^{min}(t_1\wedge\dots\wedge t_n)$ 
(resp., $p^{max}(t_1\wedge\dots\wedge t_n)$) is a solution of 
the following linear programming problem instance
$LP(t_1\wedge\dots\wedge t_n, \D^p, \IC, D^p)$: 


$$
\mbox{\emph{minimize} (resp., \emph{maximize})} \hspace*{.4cm}
\sum_{\alpha\in B^{n+m}\, |\, \forall i\in[1..n]\, \alpha[i]=true} x_\alpha
$$
$$\mbox{\emph{subject to} }  \hspace*{2cm}S(t_1\wedge\dots\wedge t_n, \D^p, \IC, D^p)
$$

\noindent
where $S(t_1\wedge\dots\wedge t_n, \D^p, \IC, D^p)$ is the following system 
of linear inequalities:
\[
\left\{
\begin{array}{lll}
\!\forall (t_i,t_j)\!\in\! EP & 
\mbox{\large$\sum$}_{\scriptsize \begin{array}{l}
\alpha\in B^{n+m}\, | \\
\alpha[i]=\mbox{\scriptsize{\emph{true}}}\ \wedge \\
\alpha[j]=\mbox{\scriptsize{\emph{true}}}
\end{array}
} x_\alpha = x_{t_i,t_j} & (A) \\[20pt]

\!\forall(t_i,t_j)\!\in\!EP & p^{min}(t_i\wedge t_j)\leq x_{t_i,t_j}\leq p^{max}(t_i\wedge t_j) & (B)\\[5pt]

\!\forall i\in\![1..n\!+\!m] & \mbox{\large$\sum$}_{\alpha\in B^{n+m} \,|\, \alpha[i]=\mbox{\scriptsize{\emph{true}}}} \ x_\alpha = p(t_i) & (C)\\[5pt]

 & \mbox{\large$\sum$}_{\alpha\in B^{n+m}} x_\alpha = 1 & (D) \\
\end{array}
\right.
\]

Therein:
$(i)$
$x_{t_i,t_j}$ is a variable representing the probability that $t_i$ and $t_j$ coexist;  and
$(ii)$ $\forall \alpha\in B^{n+m}$, $x_\alpha$ is a variable
representing the probability that $\forall i\in[1..n+m]$ the truth value of $t_i$ is 
$\alpha[i]$; that is,
$x_\alpha$ is the probability of the event
$\bigwedge_{i|\alpha[i]=\mbox{\scriptsize{\emph{true}}}} t_i \ \wedge\bigwedge_{i|\alpha[i]=\mbox{\scriptsize{\emph{false}}}} \neg t_i$.

Since $HG(D^p, \IC)$ is a tree, Lemma~\ref{lem:pmin-pmax-chain} ensures that, for each $(t_i,t_j)\in EP$,
$p^{min}(t_i\wedge t_j)$ and $p^{max}(t_i\wedge t_j)$
can be computed in polynomial time w.r.t. the size of $D^p$. 
Therefore, we assume that they are precomputed constants 
in $LP(t_1\wedge\dots\wedge t_n, \D^p, \IC, D^p)$.

It is easy to see that $LP(t_1\wedge\dots\wedge t_n, \D^p, \IC, D^p)$ can 
be solved in polynomial time w.r.t. the size of $D^P$,
as it consists of at most $6n-2$ inequalities using $2^{2n-1}+2n-1$ variables,
and $n$ only depends on the number of relations appearing in $Q$
(we recall that we are addressing data complexity, thus queries are of constant arity).

We now show that, for each solution of 
$S(t_1\wedge\dots\wedge t_n, \D^p, \IC,$ $D^p)$, there is a model 
$Pr$ of $D^p$ w.r.t. $\IC$
such that $p(t_1\wedge\dots\wedge t_n)$ w.r.t. $Pr$ is equal to
$\sum_{\scriptsize \begin{array}{l}
\alpha\in B^{n+m}\, |\\
\forall i\in[1..n]\, \alpha[i]=\mbox{\scriptsize{\emph{true}}}
\end{array}} \hspace*{-2mm} x_\alpha$, and vice versa.

Given a solution $\sigma$ of 
$S(t_1\wedge\dots\wedge t_n, \D^p, \IC, D^p)$, 
for each $\alpha\in B^{n+m}$ we denote with $\sigma_\alpha$ the value 
assumed by the variable $x_\alpha$ in $\sigma$; moreover, for each 
$(t_i, t_j)\in EP$ we denote with $\sigma_{t_i,t_j}$ the value 
assumed by the variable $x_{t_i,t_j}$ in $\sigma$.

For each $(t_i, t_j)\in EP$, we denote with $D^p_{t_i,t_j}$ 
the maximal subset of $D^p$ which contains only $t_i$, $t_j$, and the tuples
along the path connecting $t_i$ and $t_j$.

From Proposition~\ref{pro:intquery}, 
the fact that, for each $(t_i, t_j)\in EP$, the value $\sigma_{t_i,t_j}$ is such that 
$p^{min}(t_i\wedge t_j)\leq \sigma_{t_i,t_j} \leq p^{min}(t_i\wedge t_j)$, implies that there is at least a model 
$Pr_{t_i,t_j}$ of $D^p_{t_i,t_j}$ w.r.t. $\IC$ such that
$p(t_i\wedge t_j)$ w.r.t. $Pr_{t_i,t_j}$ is equal to
$\sigma_{t_i, t_j}$.
For each $(t_i, t_j)\in EP$, we consider a model $Pr_{t_i,t_j}$
of $D^p_{t_i,t_j}$ w.r.t. $\IC$ such that
$p(t_i\wedge t_j)$ w.r.t. $Pr_{t_i,t_j}$ is equal to
$\sigma_{t_i, t_j}$.
Moreover, for each possible world $w\in pwd(D^p_{t_i,t_j})$, we 
define the relative weight of $w$ (and denote it by $wr(w)$) as:
$$
wr(w) = \left\{
\begin{array}{ll}
\begin{tabular}{c}
$Pr_{t_i,t_j}(w)$\\
\hline
$\sum_{w'\in D^p_{t_i,t_j} \wedge t_i\in w' \wedge t_j\in w'} Pr_{t_i,t_j}(w')$
\end{tabular} &
\mbox{if $t_i\in w \wedge t_j\in w$}\\
 & \\
\begin{tabular}{c}
$Pr_{t_i,t_j}(w)$\\
\hline
$\sum_{w'\in D^p_{t_i,t_j}\wedge, t_i\in w' \wedge t_j\not\in w'} Pr_{t_i,t_j}(w')$
\end{tabular} &
\mbox{if $t_i\in w \wedge t_j\not\in w$}\\
 & \\
\begin{tabular}{c}
$Pr_{t_i,t_j}(w)$\\
\hline
$\sum_{w'\in D^p_{t_i,t_j}\wedge t_i\not\in w' \wedge t_j\in w'} Pr_{t_i,t_j}(w')$
\end{tabular} &
\mbox{if $t_i\not\in w \wedge t_j\in w$}\\
& \\
\begin{tabular}{c}
$Pr_{t_i,t_j}(w)$\\
\hline
$\sum_{w'\in D^p_{t_i,t_j}\wedge t_i\not\in w' \wedge t_j\not\in w'} Pr_{t_i,t_j}(w')$
\end{tabular} &
\mbox{if $t_i\not\in w \wedge t_j\not\in w$}\\
\end{array}
\right.
$$


It is easy to see that, for each possible world $w\in pwd(D^p)$, there is
for each pair $(t_i, t_j)\in EP$ a possible world 
$w_{t_i,t_j}\in pwd(D^p_{t_i,t_j})$ such that 
$w = \bigcup_{(t_i, t_j)\in EP} w_{t_i,t_j}$, and vice versa.

We consider the interpretation $Pr$ of $D^p$ defined as follows.
For each possible world $w\in pwd(D^p)$, we consider the possible worlds
$w_{t_i,t_j}$ such that $w = \bigcup_{(t_i, t_j)\in EP} w_{t_i,t_j}$ 
and define the interpretation $Pr$ of $D^p$ as:
$$Pr(w) = \sigma_\alpha \prod_{(t_i, t_j)\in EP} wr(w_{t_i,t_j}),$$
\noindent
where $\alpha$ is the tuple in $B^{n+m}$ which agrees with $w$ on the 
presence/absence of $t_1,\cdots, t_{n+m}$
(i.e., $\forall i\in [1..n+m]\, \alpha[i]=\mbox{\scriptsize{true}}$ (resp. \emph{false}) iff 
$t_i\in w$ (resp. $t_i\not\in w$)).
It is easy to see that $Pr$ is a model for $D^p$ w.r.t. $\IC$.
Specifically,  the following conditions hold:
\begin{itemize}
\item 
$Pr$ assigns probability $0$ to every possible world $w$ not satisfying $\IC$. 
This can be proved reasoning by contradiction. 
Assume that $Pr(w)>0$ and $w$ does not satisfy $\IC$.
Consider the possible worlds $w_{t_i,t_j}$ such that 
$$w = \bigcup_{(t_i, t_j)\in EP} w_{t_i,t_j}.$$ 
Since $Pr(w)>0$, for each $(t_i, t_j)\in EP$ it holds that
$$Pr_{t_i,t_j}(w_{t_i,t_j})>0.$$
Hence, since $Pr_{t_i,t_j}$ is a model of $D^p_{t_i,t_j}$, then
$w_{t_i,t_j}$ contains no pair of tuples $t',t''$ connected by 
an edge in $HG(D^p, \IC)$. 
Therefore, $w$ contains no pair of tuples $t',t''$ connected by an edge in $HG(D^p, \IC)$, thus contradicting that $w$ does not satisfy $\IC$.
\item
For each tuple $t\in D^p$, $p(t) =  \sum_{w\in pwd(D^p) \wedge t\in w} Pr(w)$.
This follows from the fact that, given a tuple $t\in D^p$, and such that $t$ belongs 
to a chain whose ends are the tuples $t_i, t_j$, the probability of a tuple $t$
is given by 
$$\sum_{w_{t_i,t_j}\in pwd(D^p_{t_i,t_j})\, s.t.\, t\in w_{t_i,t_j}} Pr_{t_i,t_j}(w_{t_i,t_j}).$$
The latter is equal to  $\sum_{w\in pwd(D^p) s.t. t\in w} Pr(w)$,
since for each $w_{t_i,t_j}\in pwd(D^p_{t_i,t_j})$ it holds that
$$\sum_{w\in pwd(D^p)\, s.t.\, w_{t_i,t_j}\subseteq w} Pr(w) = Pr_{t_i,t_j}(w_{t_i,t_j}).$$
\end{itemize}

Therefore, the interpretation $Pr$ is a model for $D^p$ w.r.t. $\IC$, and the probability
assigned to $t_1 \wedge \cdots, t_n$ by $Pr$ is equal to
$\sum_{\tiny \begin{array}{l}
\alpha\in B^{n+m} \\
\wedge \forall i\in[1..n]\, \alpha[i]=\mbox{\scriptsize{true}}
\end{array}} \sigma_\alpha$.
Hence, it is easy to see that $p^{min}(t_1\wedge\dots\wedge t_n)$ 
(resp. $p^{max}(t_1\wedge\dots\wedge t_n)$) is the optimal solution 
of $LP(t_1\wedge\dots\wedge t_n, \D^p, \IC, D^p)$ and can be computed in polynomial
time w.r.t. the size of $D^p$, which completes the proof.
\end{proof}


\noindent{\textbf{Theorem~\ref{theo:qatractable}.} 
For projection-free queries, \qa\ is in $PTIME$ if $HG(D^p, \IC)$ is a simple graph.}
\begin{proof}
Let $\vec{t}$ be an answer of the projection-free query $Q$ posed on the deterministic version of $D^p$.
The minimum and maximum probabilities $p^{min}$ and $p^{max}$ of $\vec{t}$ as answer of $Q$ 
over $D^p$ can be determined as follows.
Let $T=\{t_1, \dots, t_n\}$ be the set of tuples in $D^p$ such that $Q(\vec{t})=t_1\wedge \cdots\wedge t_n$.
$T$ can be partitioned into the sets $T_1, \dots, T_k$, such that:\\
$1)$ $k$ is the number of distinct (maximal) connected components of $HG(D^p, \IC)$, each of which contains at least one tuple in $T$; \\
$2)$ for each $i\in [1..k]$, $T_i$ contains the tuples of $T$ belonging to the $i$-th 
maximal connected component of $HG(D^p, \IC)$ among those mentioned in $1)$.\\
Let $\vec{t}_i$ be the conjunction of the tuples belonging to the partition $T_i$ of $T$.
Since every maximal connected component of $HG(D^p, \IC)$ is either a clique or a tree,
lemmas~\ref{lem:qatractable-clique} and \ref{lem:qatractable-tree} ensure that
$p^{min}(\vec{t}_i)$ and $p^{max}(\vec{t}_i)$ can be computed in polynomial time w.r.t. the size of $D^p$.
As distinct tuples $\vec{t}_i$ and $\vec{t}_j$, with $i,j\in [1..k]$, belong to distinct maximal connected components of $HG(D^p, \IC)$, they can be viewed as events among which no 
correlation is known.
Hence, $p^{min}(\vec{t})$  (resp., $p^{max}(\vec{t})$) can be determined by applying 
Fact~\ref{fac:boole} to the events $\vec{t}_1, \dots \vec{t}_k$, with the probability of  
$\vec{t}_i$ equal to $p(\vec{t}_i)=p^{min}(\vec{t}_i)$ (resp., 
$p(\vec{t}_i)=p^{max}(\vec{t}_i)$), for each $i\in [1..k]$.
\end{proof}

\subsection{Extending tractable cases of query evaluation}
\label{app:estensionequery}
As discussed in the core of the paper (Section~\ref{app:extensions}), our tractability result on query evaluation can be extended to the cases that:
$i)$ tuples are associated with ranges of probabilities, instead of exact probability values;
$ii)$ denial constraints are probabilistic.
We here give a hint on how the proof of Lemma~\ref{lem:qatractable-tree} can be extended to 
these cases (Lemma~\ref{lem:qatractable-tree} states that projection-free queries can be 
evaluated in \emph{PTIME} if the conflict hypergraph is a tree, and is the core of the 
proof of Theorem~\ref{theo:qatractable}).

As regards extension $ii)$, it is easy to see that, as shown for \cc, any instance $I$ of the
query evaluation problem in the presence of probabilistic constraints is equivalent to an instance $I'$ of \qa, where the conflict hypergraph $H'$ of $I'$ is obtained by augmenting each hyperedge of the conflict hypergraph $H$ of $I$ with an ear. 
The point is that, even if $H$ is a tree,  this reduction makes $H'$ contain hyperedges with 
more than two nodes, thus $H'$ is no more a tree.
However, $H'$ is a hypertree of a particular form: for any pairs of intersecting edges, their intersection consists of a unique node, which is a node inherited from $H$ (the new nodes of 
$H'$ are all ears).
This implies that the minimum and maximum probabilities $p^{min}$ and $p^{max}$ of an answer
can be still computed as solutions of the two variants of the optimization problem 
$LP(t_1\wedge\dots\wedge t_n, \D^p, \IC, D^p)$ introduced in the proof of Lemma~\ref{lem:qatractable-tree}. 
The fact that $LP(t_1\wedge\dots\wedge t_n, \D^p, \IC, D^p)$  can be still written and solved 
in polynomial time derives from the fact that the values $p^{min}(t_i\wedge t_j)$ and 
$p^{max}(t_i\wedge t_j)$ occurring in the inequalities (B) can be still evaluated in 
polynomial time, by observing that both $p^{min}(t_i\wedge t_j)$ and 
$p^{max}(t_i\wedge t_j)$ can be obtained by exploiting Lemma~\ref{lem:pmin-chain-hyper} for the minimum probability value, and an analogous result for the maximum probability value.
Observe that this reasoning does not work (as is) for general hypertrees, as 
in this case we are not assured that the tuples composing the answer are in intersections
between distinct pairs of hyperedges.

As regards extension $i)$, the minimum and maximum probabilities $p^{min}$ and $p^{max}$ 
of an answer can be computed as solutions of the two variants of the optimization problem 
$LP(t_1\wedge\dots\wedge t_n, \D^p, \IC, D^p)$ with the following changes:\\
$1)$ equalities (C) are replaced with pairs of inequalities imposing that, for each $t_i$, its probability ranges between the minimum and maximum marginal probabilities of the range associated with $t_i$ in the PDB;\\
$2)$ the values $p^{min}(t_i\wedge t_j)$ and $p^{max}(t_i\wedge t_j)$ occurring in the inequalities (B) are evaluated by considering the minimum probabilities for the tuples along the path connecting $t_i$ and $t_j$ in the conflict tree.
Moreover, when evaluating $p^{min}(t_i\wedge t_j)$, the minimum marginal probabilities for $t_i$ and $t_j$ are taken into account, while, for $p^{max}(t_i\wedge t_j)$, we have to consider their maximum probabilities.
Therein, the maximum probability of a tuple $t$ is the minimum between 
the upper bound of the probability range of $t$, and the maximum probability value that $t$ 
can have according to the conflict tree (this value is entailed by the tuples connected to 
$t$ by direct edges: as implied by Theorem~\ref{theo:tractableCC}, the sum of the probabilities of two tuples connected through an edge must be less than or equal to $1$).
Intuitively enough, we consider the minimum probabilities for the intermediate 
tuples between $t_i$ and $t_j$ as this allows the greatest degree of freedom in
distributing $t_i$ and $t_j$ in the probability space.
\end{document}